\documentclass[a4paper,onecolumn,11pt,accepted=2022-10-06]{quantumarticle}
\pdfoutput=1
\usepackage{fullpage}
\usepackage[bookmarks]{hyperref}
\usepackage{amssymb}
\usepackage{amsmath}
\usepackage{amsthm}
\usepackage{tcolorbox}
\usepackage{graphicx,color,colordvi}
\usepackage{bbm}
\usepackage{cleveref}
\usepackage{stmaryrd}
\usepackage[utf8]{inputenc}
\usepackage{fullpage}
\usepackage[blocks]{authblk}
\usepackage{dsfont}
\usepackage{mathtools}
\usepackage{authblk}

\usepackage{centernot}

\usepackage{pgfplots}

\def\Re{{\operatorname{Re}}}

\def\openone{\leavevmode\hbox{\small1\kern-3.8pt\normalsize1}}

\def\cm{{\cal M}}

\def\cn{{\cal N}}
\def\cd{{\cal D}}
\def\ce{{\cal E}}

\def\PSD{\operatorname{PSD}}

\def\cc{{\cal C}}

\def\CC{\mathbb{C}}
\def\RR{\mathbb{R}}

\def\NN{\mathbb{N}}

\def\11{\mathbb{I}}

\def\fqln{\succeq_{\operatorname{l.n.}}^{\operatorname{fq}}}
\def\fqmc{\succeq_{\operatorname{m.c.}}^{\operatorname{fq}}}
\def\cmc{\succeq_{\operatorname{m.c.}}^{\operatorname{c}}}
\def\cln{\succeq_{\operatorname{l.n.}}^{\operatorname{c}}}
\def\rln{\succeq_{\operatorname{l.n.}}^{\operatorname{reg}}}
\def\sln{\succeq_{\operatorname{l.n.}}}
\def\mc{\succeq_{\operatorname{m.c.}}}
\def\deg{\succeq_{\operatorname{deg}}}
\newcommand{\efqln}[1]{\succeq_{\operatorname{l.n.}}^{\operatorname{fq},{#1}}}
\newcommand{\ecln}[1]{\succeq_{\operatorname{l.n.}}^{\operatorname{c},{#1}}}
\newcommand{\erln}[1]{\succeq_{\operatorname{l.n.}}^{\operatorname{reg},{#1}}}
\newcommand{\esln}[1]{\succeq_{\operatorname{l.n.}}^{{#1}}}

\newtheorem{definition}{Definition}[section]
\newtheorem{proposition}[definition]{Proposition}
\newtheorem{lemma}[definition]{Lemma}

\newtheorem{theorem}[definition]{Theorem}
\newtheorem{corollary}[definition]{Corollary}
\newtheorem{conjecture}[definition]{Conjecture}

\newtheorem{remark}[definition]{Remark}

\def\eps{\varepsilon}

\newcommand{\supp}{\mathop{\rm supp}\nolimits}

\newcommand{\tr}{\mathop{\rm Tr}\nolimits}

\usepackage{pst-node}
\usepackage{tikz-cd} 

\newcommand{\bra}[1]{\langle#1|}
\newcommand{\ket}[1]{|#1\rangle}

\newcommand{\cA}{{\cal A}}
\newcommand{\cB}{{\cal B}}
\newcommand{\cD}{{\cal D}}
\newcommand{\cE}{{\cal E}}

\newcommand{\cG}{{\cal G}}
\newcommand{\cN}{{\cal N}}
\newcommand{\cT}{{\cal T}}
\newcommand{\cH}{{\cal H}}
\newcommand{\cK}{{\cal K}}

\newcommand{\cV}{{\cal V}}

\newcommand{\cP}{\mathcal{P}}

\newcommand{\cL}{{\cal L}}

\newcommand{\cR}{{\cal R}}

\newcommand{\cM}{{\mathcal{M}}}

\newcommand{\Id}{{\mathds{1}}}

\newcommand{\C}{{\mathbb{C}}}
\newcommand{\R}{{\mathbb{R}}}

\newcommand{\M}{{\mathbb{M}}}

\def\d{\mathrm{d}}
\def\e{\mathrm{e}}

\usepackage{graphicx}
\usepackage{setspace}
\usepackage{verbatim}
\usepackage{subfig}

\newcommand{\CPTP}{\operatorname{CPTP}}

\numberwithin{equation}{section}

\DeclareRobustCommand\openone{\leavevmode\hbox{\small1\normalsize\kern-.33em1}}
\newcommand{\identity}{I}
\newcommand{\id}{{\rm{id}}}
\newcommand{\be}{\begin{equation}}
	\newcommand{\ee}{\end{equation}}
\newcommand{\bea}{\begin{eqnarray}}
	\newcommand{\eea}{\end{eqnarray}}
\newcommand{\beas}{\begin{eqnarray*}}
	\newcommand{\eeas}{\end{eqnarray*}}

\title{On contraction coefficients, partial orders and approximation of capacities for quantum channels}

\begin{document}

\author{Christoph Hirche}
\affiliation{QMATH, Department of Mathematical Sciences, University of Copenhagen, Universitetsparken 5, 2100 Copenhagen, Denmark}
\author{Cambyse Rouz\'{e}}
\author{Daniel Stilck Fran\c ca}
\affiliation{Department of Mathematics, Technische Universit\"at M\"unchen, 85748 Garching, Germany}
\affiliation{Munich Center for Quantum Science and Technology (MCQST), M\"unchen, Germany}

\date{}

\maketitle

\begin{abstract}
The data processing inequality is the most basic requirement for any meaningful measure of information. It essentially states that distinguishability measures between states decrease if we apply a quantum channel and is the centerpiece of many results in information theory. Moreover, it justifies the operational interpretation of most entropic quantities. In this work, we revisit the notion of contraction coefficients of quantum channels, which provide sharper and specialized versions of the data processing inequality. A concept closely related to data processing is partial orders on quantum channels. First, we discuss several quantum extensions of the well-known less noisy ordering and relate them to contraction coefficients. We further define approximate versions of the partial orders and show how they can give strengthened and conceptually simple proofs of several results on approximating capacities. Moreover, we investigate the relation to other partial orders in the literature and their properties, particularly with regards to tensorization. We then examine the relation between contraction coefficients with other properties of quantum channels  such as hypercontractivity. Next, we extend the framework of contraction coefficients to general f-divergences and prove several structural results. Finally, we consider two important classes of quantum channels, namely Weyl-covariant and bosonic Gaussian channels. For those, we determine new contraction coefficients and relations for various partial orders.
\end{abstract}

\tableofcontents

\section{Introduction}

One of the arguably most fundamental concepts in (quantum) information theory is that of data processing\footnote{``If you don't have data processing, you're toast'', \href{https://youtu.be/3gloUVUzgLc?t=1423}{\textit{Mark M. Wilde, ISL Colloquium, 17/09/2020}}}. 
Many relevant quantities are monotone under the application of a quantum channel. That is what allows us to assign them an operational meaning in terms of distinguishability and, in turn, makes them useful in assessing physical properties. A commonly considered quantity is the relative entropy for which, when considering an arbitrary channel $\cn$, data processing manifests as
\begin{align}
D( \cn(\rho) \| \cn(\sigma) ) \leq D(\rho\|\sigma)
\end{align}
for any two states $\rho,\sigma$. This mathematical statement also gives rise to an operational interpretation of data processing: The relative entropy gives the optimal rate at which one can asymptotically discriminate two quantum states in an asymmetric setting~\cite{HP91, ON00}. It is now immediately apparent that applying a quantum channel to the state can never make the discrimination task easier; therefore the relative entropy has to become smaller. 

Since the relative entropy acts as a parent quantity for several other entropic quantities, they also inherit its data processing property. An important example is the mutual information: Given a bipartite quantum state $\rho_{AA'}$ and a quantum channel $\cn:A'\to B$, we always have
\begin{align}
I(A:B)_\rho \leq I(A:A')_\rho\,
\end{align}

Although data processing is a potent tool, often it is not sufficient to know that a quantity is decreasing when applying a channel, one also needs to quantify by how much. When fixing the reference state $\sigma$ and varying $\rho$, in many cases it is possible to show that $D(\cn(\rho)\|\cn(\sigma))$ is strictly smaller than $D(\rho\|\sigma)$ unless $\rho=\sigma$. A natural approach is then to consider contraction coefficients, 
\begin{align}\label{contractioncoeff}
\eta_{\operatorname{Re}}(\cn,\sigma):=\sup_{\rho\ne\sigma}\frac{D(\cn(\rho)\|\cn(\sigma))}{D(\rho\|\sigma)}\,,
\end{align}
and we say that $\cn$ satisfies a \textit{strong data processing inequality} at input state $\sigma$ when $\eta_{\operatorname{Re}}(\cN,\sigma)<1$. These contraction coefficients have already found many applications in information theory~\cite{polyanskiy2017strong,ramakrishnan2020quantum} and were studied in the quantum setting in~\cite{Lesniewski1998a,hiai2016contraction}.

The contraction coefficient \eqref{contractioncoeff} has been extensively studied in the classical setting. In their pioneering work, Ahlswede and G\'{a}cs \cite{ahlswede1976spreading} discovered a deep relationship between $\eta_{\operatorname{Re}}(\cn,\sigma)$ and several other quantities, among which the maximal correlation and hypercontractive properties of the channel. After this, the notion of entropic contraction coefficients were extended to general $f$-divergences and other relevant measures of distinguishability between two probability distributions such as the Dobrushin contraction coefficient (for the trace distance), or that of the $\chi^2$-divergence. \cite{cohen1993relative,choi1994equivalence,miclo1997remarques,cohen1998comparisons,del2003contraction} (see also the more recent contributions \cite{raginsky2016strong,polyanskiy2017strong,7279116,makur2018comparison}).

Another closely related question is to compare channels based on how much they contract a certain quantity. In this case, comparing the contraction coefficients is only of limited use. A more direct approach is to define partial orders on the set of channels based on how much the channels contract each possible input state individually. While many such orders are known, the most commonly used ones are the less noisy and the more capable order~\cite{korner1977comparison}, especially in the classical setting.

In the quantum setting, Watanabe introduced quantum generalizations of the less noisy and more capable orders~\cite{watanabe2012private}. A noteworthy difference to the classical case is that the definitions in~\cite{watanabe2012private} include a regularization, in the sense that the property has to hold not just for $\cn$ but $\cn^{\otimes n}$ for all $n$. This choice proved to be useful in the context of capacities and was therefore operationally motivated. However, in many cases it can also be useful to look at the $n=1$ case, see e.g.~\cite{tikku2020non}. We go a step further and also consider two additional quantum generalizations: a third version with additional reference system that turns out to have some desirable properties such as tensorization and a fourth where all systems are fully quantum.

In this work we revisit the concepts of contraction coefficients and partial orders in quantum information theory. Most importantly we show how both contraction coefficients and the less noisy order can be understood in very similar terms as either relative entropy or mutual information based tools. From this starting point, we discuss properties such as tensorization, bounds and the relation to other partial orders in the spirit of~\cite{makur2018comparison,raginsky2016strong}. 

We then discuss the application of partial orders in the context of quantum capacities, revisiting the recently introduced concept of approximate degradability and the resulting capacity bounds~\cite{sutter2017approximate} and introducing approximate versions of the different less noisy and more capable definitions. This leads to weaker requirements for approximating the quantum, private and classical capacities. Moreover, as corollaries we get the approximate degradability capacity bounds via conceptually straightforward proofs.

The main drawback of the entropy based partial orders discussed here is that compared to (approximate) degradability, they are rather complicated to verify. It can be nontrivial to check whether two quantum channels fulfill the conditions of the partial order of choice. To this end, we investigate alternative characterizations and conditions, providing some progress towards this question. We also draw connections to the notions of hypercontractivity and other functional inequalities.  

Finally, we discuss both generalizations and special cases of the discussed concepts. First, we investigate generalizations of contraction coefficients and the less noisy order to $f$-divergences~\cite[Chapter 7]{OhyaPetz-Entropy-1993}. Then we present several results for the setting where the channels are from the special classes of either Weyl-covariant channels or bosonic Gaussian channels. In particular we find examples of Weyl-covariant channels that satisfy the less noisy ordering without one being the degraded version of the other. We end this paper with a derivation of precise contraction coefficients for the most important classes of quantum Gaussian channels, namely the attenuator, the amplifier and the additive noise channels. We also consider the case of tensor products of the latter, under the constrained minimum output entropy conjecture.

\paragraph{Some Notations:} In what follows, we denote by $\cH$ a finite dimensional Hilbert space. $\cB(\cH)$ is the set of linear operators acting on $\cH$. The set of quantum states is denoted by $\cD(\cH)$, and that of full-rank states by $\cD(\cH)_+$. Given two quantum systems $\cH,\cK$, the set of quantum channels $\cN:\cB(\cH)\to\cB(\cK)$, that is of completely positive, trace preserving superoperators, is denoted by $\CPTP(\cH,\cK)$. For a quantum channel $\cN$, we will denote the channel in the Heisenberg picutre by $\cN^*$. The set of probability mass functions over sets of cardinality $m\in\NN$ is denoted by $\cP_m$. We also use standard quantum information conventions, denoting quantum systems by capital letters $A,B$ and classical ones by $U$ and $X$.

Furthermore, for two states $\rho,\sigma$  we define their relative entropy as 
\begin{align}
D(\rho\|\sigma)=\tr\left[\rho\left(\log(\rho)-\log(\sigma)\right)\right]
\end{align}
if the support of $\rho$ is contained in that of $\sigma$. It this condition is not satisfied, it is defined to be $+\infty$. We define the von Neumann entropy of a quantum state $\rho_A$ as 
\begin{align}
H(A)_{\rho_A}:=-\tr\left[\rho_A\log(\rho_A)\right]
\end{align}
For other standard definitions such as most entropic quantities and channel capacities we refer to~\cite{tomamichel2015quantum, wilde2013quantum}.

\section{Contraction coefficients and partial orders of quantum channels}\label{sec:defs}

In this section we will first review the main partial orders used throughout the manuscript: Degradable, less noisy and more capable. In the case of the latter two, we introduce several variants that are all valid generalizations of the corresponding classical formulations. Next, we review contraction coefficients and in particular, the one based on the relative entropy for which we then show an alternative formulation in terms of mutual information. This  provides us with a theme throughout the rest of the manuscript. Finally, we briefly discuss some connections between contraction coefficients and partial orders.

\subsection{Channel preorders in quantum information theory, and related tasks}\label{chanpreorder}

Channel preorders can be understood as ways of quantifying a channel's noisiness as compared to another one. However, the notion of noise varies according to the tasks for which a given channel is being used. For this reason, a zoo of preorders can be found already in the classical setting \cite{makur2018comparison}. We recall some of the main preorders one can find in the literature and recall their relations here. We start with the most common partial order: Degradability. 

\begin{tcolorbox}
\textbf{Degradation preorder:} \cite{cover1972broadcast,bergmans1973random} A channel $\cn\in\CPTP(\cH_A,\cK_{B'})$ is said to be a \textit{degraded} version of a channel $\cm\in\CPTP(\cH,\cK_{B})$ with same input space, denoted $\cm\succeq_{\operatorname{deg}} \cn$ if the following equivalent conditions hold:
\begin{itemize}
\item The channels are related as 
\begin{align}
\cn=\Theta\circ \cm
\end{align}
for some channel $\Theta\in\CPTP(\cK,\cK')$. 
\item For any c-q state $\rho_{UAR}$,
\begin{align}
H_{\min}(U|BR)_{(\id_{UR}\otimes \cm)(\rho_{UAR})}\leq H_{\min}(U|B'R)_{(\id_{UR}\otimes \cn)(\rho_{UAR})}\,.
\end{align} 
\end{itemize}
\end{tcolorbox}

The first condition is the well-known standard definition of degradability, while the equivalence with the second was recently shown in~\cite{buscemi2016degradable}. We recall that the $\min$ conditional entropy is defined as 
\begin{align}
H_{\min}(A|B)_\rho = -\inf_{\sigma_B} \inf\{\lambda\in\R : \rho_{AB}\leq 2^\lambda \identity_A\otimes\sigma_B\}. 
\end{align}
While the first condition is often more useful due to data processing the second will make the close familiarity to the remaining partial orders easily evident. 

We will continue by defining what can be considered the basic version of the less noisy and more capable partial orders. 
\begin{tcolorbox}
\textbf{Less noisy, more capable:} \cite{korner1977comparison, watanabe2012private} A channel $\cm\in\CPTP(\cH_A,\cK_B)$ is said to be \textit{less noisy} than a channel $\cn\in\CPTP(\cH_A,\cK_{B'})$ with same input space, denoted $\cm\succeq_{\operatorname{l.n.}} \cn$, if for all classical random variables $U$ and any c-q state $\rho_{UA}$
$$I(U:B)_{(\id_U\otimes \cm)(\rho_{UA})}\ge I(U:B')_{(\id_U\otimes \cn)(\rho_{UA})}\,.$$
Moreover, $\cm$ is said to be \textit{more capable} than $\cn$, denoted $\cm\succeq_{\operatorname{m.c.}}\cn$, if for any ensemble $\{p,|\psi_x\rangle_A\}$, and the corresponding c-q state $\rho_{XA}=\sum_x\,p(x)\,|x\rangle\langle x|\otimes |\psi_x\rangle\langle\psi_x|$
\begin{align*}
I(X:B)_{(\id_X\otimes \cm)(\rho_{XA})}\ge I(X:B')_{(\id_X\otimes \cn)(\rho_{XA})}\,.
\end{align*}
\end{tcolorbox}
	We recall that the mutual information $I(A:B)_{\rho}$ of a bipartite state $\rho_{AB}$ is defined as
	\begin{align*}
	I(A:B)_\rho:=H(A)_\rho+H(B)_\rho-H(AB)_\rho\,,
	\end{align*}	 
	where $H(C)_{\rho}:=H(\rho_C)=-\tr[\rho_C\ln\rho_C]$ denotes the von Neumann entropy of a state $\rho$ on a subsystem $C$. In the classical setting, these definitions are well motivated by their applications to wiretap and broadcast channels~\cite{korner1977comparison,csiszar1978broadcast,marton1979coding}. In the quantum setting, the less noisy and more capable relations determine properties and relations of the quantum and private capacities~\cite{watanabe2012private}. Furthermore, it is immediately clear that
\begin{align}
\cm\sln\cn \quad\Rightarrow \quad \chi(\cm)\geq\chi(\cn)\,,
\end{align}
where $\chi(\cn)$ is the Holevo quantity of the channel $\cn$ defined as 
\begin{align*}
	\chi(\cn):=\sup\limits_{p,\rho_x}H\left(\sum_xp(x)\cn(\rho_x)\right)-\sum_xp(x)H(\cn(\rho_x)).
\end{align*}

Note that we stated the condition here in its most common form via the mutual information. It can however easily be written in terms of the conditional entropy, $H(A|B)_\rho:=H(AB)_\rho-H(B)_\rho$,  as
\begin{align}
H(U|B)_{(\id_U\otimes \cm)(\rho_{UA})}\leq H(U|B')_{(\id_U\otimes \cn)(\rho_{UA})}\,,
\end{align}
and similarly for the more capable relation, which bears noticeable similarity with the alternative definition of degradability. The lack of a reference system here is crucial and we will discuss this in more details later. Indeed, one can also find a different partial order in the literature called \textit{less ambiguous} that is defined as a mix of the above characteristics: min-entropy, but no reference system (see e.g.~\cite{buscemi2017comparison}). 
\begin{align}
H_{\min}(U|B)_{(\id_U\otimes \cm)(\rho_{UA})}\leq H_{\min}(U|B')_{(\id_U\otimes \cn)(\rho_{UA})}\,.
\end{align}

As expected in the quantum setting, the notions of less noisy and more capable need in general to be regularized in order to lead to an operational interpretation~\cite{watanabe2012private}\footnote{Strictly speaking, Watanabe defined the notion of regularized less noisy and more capable between a channel and its complementary, which he simply referred to as less noisy and more capable.}: 

\begin{tcolorbox}
\textbf{Regularized less noisy, more capable:}
With the previous notations, the channel $\cm$ is said to be \textit{regularized less noisy} than $\cn$, denoted $\cm\succeq_{\operatorname{l.n.}}^{\operatorname{reg}}\cn$, if for any $n\in\NN$, any classical register $U$ and any c-q state $\rho_{UA^n}$,
\begin{align*}
I(U:B^n)_{(\id_{U}\otimes \cm^{\otimes n})(\rho_{UA^n})}\ge I(U:B'^n)_{(\id_{U}\otimes \cn^{\otimes n})(\rho_{UA^n})}\,.
\end{align*} 
Moreover, $\cm$ is \textit{regularized more capable} than $\cn$, denoted $\cm\succeq_{\operatorname{m.c.}}^{\operatorname{reg}}\cn$, if for any $n\in\NN$, and  ensemble $\{p,|\psi_x\rangle_{A^n}\}$ over $\cH_A^{\otimes n}$, and the corresponding c-q state $\rho_{XA^n}=\sum_{x}\,p(x)\,|x\rangle\langle x|\otimes |\psi_{x}\rangle\langle\psi_{x}|_{A^n}$,
\begin{align*}
I(X:B^n)_{(\id_{X}\otimes \cm^{\otimes n})(\rho_{XA^n})}\ge I(X:B'^n)_{(\id_{X}\otimes \cn^{\otimes n})(\rho_{XA^n})}\,.
\end{align*} 
\end{tcolorbox}
 In general the regularized version of the partial orders give a strictly stronger condition than the unregularized one. We will discuss the problem of tensorization in more details in Section~\ref{Sec:PropertiesAndBounds}. Similar to the standard definition we can also here directly observe that
 \begin{align}
\cm\rln\cn \quad\Rightarrow \quad C(\cm)\geq C(\cn)\,,
\end{align}
where $C(\cn)$ is the classical capacity of the channel $\cn$. 
 An alternative quantum generalization is to provide the channel input with an additional quantum reference system. We will later see that this also leads to an order that obeys additional desirable properties such as tensorization. 

 \begin{tcolorbox}
 \textbf{Completely less noisy, more capable} 
 With the previous notations, the channel $\cm$ is said to be \textit{completely less noisy} than $\cn$, denoted $\cm\succeq_{\operatorname{l.n.}}^{\operatorname{c}}\cn$ if for any classical register $U$ and any c-q state $\rho_{UAR}$,
\begin{align*}
I(U:BR)_{(\id_{UR}\otimes \cm)(\rho_{UAR})}\ge I(U:B'R)_{(\id_{UR}\otimes \cn)(\rho_{UAR})}\,.
\end{align*} 
Moreover, $\cm$ is \textit{completely more capable} than $\cn$, denoted $\cm\succeq_{\operatorname{m.c.}}^{\operatorname{c}}\cn$, if for any ensemble $\{p,|\psi_x\rangle_{AR}\}$ over $\cH_{AR}$, and the corresponding c-q state $\rho_{XAR}=\sum_{x}\,p(x)\,|x\rangle\langle x|\otimes |\psi_{x}\rangle\langle\psi_{x}|$,
\begin{align*}
I(X:BR)_{(\id_{XR}\otimes \cm)(\rho_{XAR})}\ge I(X:B'R)_{(\id_{XR}\otimes \cn)(\rho_{XAR})}\,.
\end{align*} 
\end{tcolorbox}

Note that there is no known restriction on the dimension of the $R$ system, making this condition generally even harder to check than the standard variants. This is a common problem in quantum information theory, see e.g.~\cite{BeigiGohari2013,HW2020,CHW2020,datta2019convexity,wilde2020amortized}.

At this point it should be briefly noticed that, in the classical less noisy setting, allowing a classical reference system does not change the order itself. Even more, one can easily see that even in the quantum setting a classical reference system would not help.
\begin{lemma}
If one restricts the reference system in the completely less noisy ordering to be classical the order would become equivalent to the usual less noisy ordering, meaning
\begin{align}
\cm\cln\cn \Leftrightarrow \cm\sln\cn\,.
\end{align}
\end{lemma}
\begin{proof}
The ``$\Rightarrow$'' direction is obvious. The ``$\Leftarrow$'' direction follows by noticing that
\begin{align*}
\cm\cln\cn ~~ \Leftrightarrow~~ & ~~  I(U:B|R)_{(\id_{UR}\otimes \cm)(\rho_{UAR})}\ge I(U:B'|R)_{(\id_{UR}\otimes \cn)(\rho_{UAR})}  \\
\Leftrightarrow~~ & ~~ \sum_r p(r) I(U:B|R=r)_{(\id_{UR}\otimes \cm)(\rho_{UAR})}\ge \sum_r p(r) I(U:B'|R=r)_{(\id_{UR}\otimes \cn)(\rho_{UAR})}\,.
\end{align*} 
The first equivalence follows by subtracting $I(U:R)_\rho$ on both sides and the second because $R$ is assumed to be classical. Now the proof follows by the definition of the less noisy ordering.
\end{proof}
As quantum conditioning is significantly more complex, the same argument does not hold for quantum reference systems~\cite{BeigiGohari2013,hirche2018bounds}. 

All previously discussed definitions of quantum extensions of the classical less noisy order in this work have in common that the system correlated to the input is classical. However, one could envision a \textit{fully quantum} extension of the less noisy order where this restriction is lifted. 

\begin{tcolorbox}
\textbf{Fully quantum less noisy, more capable} A channel $\cm\in\CPTP(\cH_A,\cK_B)$ is said to be \textit{fully quantum less noisy} than a channel $\cn\in\CPTP(\cH_A,\cK_{B'})$ with same input space, denoted $\cm\fqln \cn$, if for all quantum states $\rho_{AA'}$
$$I(A:B)_{(\id_A\otimes \cm)(\rho_{AA'})}\ge I(A:B')_{(\id_A\otimes \cn)(\rho_{AA'})}\,.$$
Moreover, $\cm$ is said to be \textit{fully quantum more capable} than $\cn$, denoted $\cm\succeq_{\operatorname{m.c.}}^{\operatorname{fq}}\cn$, if for all pure quantum states $\Psi_{AA'}$
\begin{align*}
I(A:B)_{(\id_A\otimes \cm)(\Psi_{AA'})}\ge I(A:B')_{(\id_A\otimes \cn)(\Psi_{AA'})}\,.
\end{align*}
\end{tcolorbox}
Similar to before we can also here directly observe that
 \begin{align}
\cm\fqmc\cn \quad\Rightarrow \quad C_E(\cm)\geq C_E(\cn)\,,
\end{align}
where $C_E(\cn)$ is the entanglement assisted classical capacity of the channel $\cn$. 

We remark here that the fully quantum less noisy order has previously been defined in~\cite{cross2017uniform} where it was called \textit{informationally degradable}. 
Furthermore, in line with the previous discussion, when considering the $\min$ entropy this partial order is the \textit{more coherent order} discussed e.g. in~\cite{buscemi2017comparison}.

The relations between the different partial orders are summarized in Figure~\ref{Fig:RelationsPO} and we defer their proofs to Section~\ref{Sec:PropertiesAndBounds}. 

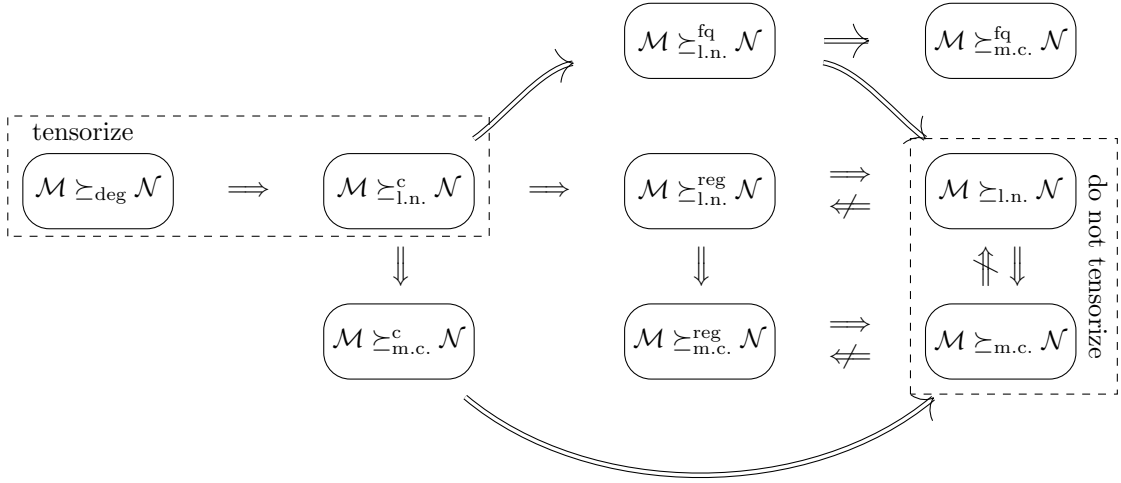
\begin{figure}[t]
\centering
\begin{tikzpicture}
\draw[rounded corners=10pt] (8,2) rectangle (10,3) node[pos=.5](fqln){$\cm\fqln\cn$};
\draw[rounded corners=10pt] (12,2) rectangle (14,3) node[pos=.5](fqmc){$\cm\succeq_{\operatorname{m.c.}}^{\operatorname{fq}}\cn$};
\draw[rounded corners=10pt] (0,0) rectangle (2,1) node[pos=.5](deg){$\cm\succeq_{\operatorname{deg}} \cn$};
\draw[rounded corners=10pt] (4,0) rectangle (6,1) node[pos=.5](cln){$\cm\succeq_{\operatorname{l.n.}}^{\operatorname{c}}\cn$};
\draw[rounded corners=10pt] (8,0) rectangle (10,1) node[pos=.5] {$\cm\succeq_{\operatorname{l.n.}}^{\operatorname{reg}}\cn$};
\draw[rounded corners=10pt] (12,0) rectangle (14,1) node[pos=.5](ln) {$\cm\succeq_{\operatorname{l.n.}}\cn$};
\draw[rounded corners=10pt] (4,-2) rectangle (6,-1) node[pos=.5](cmc) {$\cm\succeq_{\operatorname{m.c.}}^{\operatorname{c}}\cn$};
\draw[rounded corners=10pt] (8,-2) rectangle (10,-1) node[pos=.5] {$\cm\succeq_{\operatorname{m.c.}}^{\operatorname{reg}}\cn$};
\draw[rounded corners=10pt] (12,-2) rectangle (14,-1) node[pos=.5](mc) {$\cm\succeq_{\operatorname{m.c.}}\cn$};
\draw[double,->,double distance=1.3pt, shorten <=20pt,shorten >=21pt] (fqln) to [out=350,in=145] (ln);
\draw[double,->,double distance=1.3pt, shorten <=20pt,shorten >=21pt] (fqln) to [out=360,in=180] (fqmc);
\draw[double,->,double distance=1.3pt, shorten <=20pt,shorten >=21pt] (cln) to [out=35,in=190] (fqln);
\draw[double,->,double distance=1.3pt, shorten <=20pt,shorten >=21pt] (cmc) to [out=320,in=220] (mc);
\node at (3,0.5){$\Longrightarrow$};
\node at (7,0.5){$\Longrightarrow$};
\node at (11,0.7){$\Longrightarrow$};
\node at (11,0.3){$\centernot\Longleftarrow$};
\node at (11,-1.3){$\Longrightarrow$};
\node at (11,-1.7){$\centernot\Longleftarrow$};
\node[rotate=270] at (5,-0.5){$\Longrightarrow$};
\node[rotate=270] at (9,-0.5){$\Longrightarrow$};
\node[rotate=270] at (13.2,-0.5){$\Longrightarrow$};
\node[rotate=270] at (12.8,-0.5){$\centernot\Longleftarrow$};
\draw[dashed] (-0.2,-0.1) rectangle (6.2,1.5);
\node at (0.8,1.3){tensorize};
\draw[dashed] (11.8,-2.2) rectangle (14.55,1.2);
\node[rotate=270] at (14.3,-0.5){do not tensorize};
\end{tikzpicture}
\caption{\label{Fig:RelationsPO} Summary of relations of the partial orders defined in Section~\ref{chanpreorder}. Proofs can be found in Proposition~\ref{inclusions}.}
\end{figure}

As customary we also define the anti-orders, i.e. anti-degradable, anti-less noisy and so on, as above but with the roles of $\cn$ and $\cm$ interchanged. This will become relevant when we discuss capacities where we often consider $\cm$ to be the complementary channel of $\cn$.

Before ending this section we state some elementary facts about the channels dominated by a fixed channel $\cm$. To that end, given a channel $\cm$, define the \textit{less noisy domination region}
\begin{align*}
\cL_\cm:=\big\{  \cn  :\,\cm   \succeq_{\operatorname{l.n.}} \cn\big\}
\end{align*}
as well as the \textit{degradation region} of $\cm$
\begin{align*}
\cD_\cm:=\big\{  \cn  :\,\cm \succeq_{\operatorname{deg}}  \cn\big\}\,.
\end{align*}
\begin{proposition}\label{obviousprops}
	Given the channel $\cm$, its less noisy domination region $\cL_\cm$ and its degradation region $\cD_\cm$ are non-empty, closed, convex and for any channel $\Lambda$
	\begin{align*}
	\cn\in\cL_\cm\Rightarrow \Lambda\circ\cn\in\cL_{\cm}~~~\text{ and }~~~\cn\in\cD_\cm\Rightarrow \Lambda\circ \cn\in \cD_\cm\,.
	\end{align*}
\end{proposition}
\begin{proof}
	Non-emptiness is obvious since $\cm\in \cD_\cm$ and $\cm\in\cL_\cm$. The closure of the sets is proved as follows: first, as we show later in Proposition~\ref{propDmutualinfof} (see also Proposition~\ref{propDmutualinfo}), the condition 
\begin{align}
I(U:B)_{(\id_{U}\otimes \cm)(\rho_{UA})}\ge I(U:B')_{(\id_{U}\otimes \cn)(\rho_{UA})}
\end{align}
for all $\operatorname{c-q}$ state $\rho_{UA}$ with marginal $\sigma_{A}=\tr_U(\rho_{UA})$is equivalent to
\begin{align}\label{equ:rel_entropy_char}
D(\cm(\rho_{A})\|\cm(\sigma_{A}))\ge		D(\cn(\rho_{A})\|\cn(\sigma_{A}))
\end{align}
for any state $\rho_{A}\in\cD(\cH_{A})$ with $\supp(\rho_{A})\subseteq\supp(\sigma_{A})$. Thus, we will work with the equivalent characterization of the less noisy ordering given by Eq.~\eqref{equ:rel_entropy_char} to show closure.
Given two states $\rho,\sigma$ consider a sequence of channels $\cn_k\in\cL_\cn$ such that $\cn_k\to\cn$ (say with respect to the $2\to 2$ norm). Then we also have that $\cn_k(\rho)\to\cn(\rho)$ and $\cn_k(\sigma)\to\cn(\sigma)$ (say in Hilbert-Schmidt norm). Therefore
	\begin{align*}
	D(\cn(\rho)\|\cn(\sigma))\le \liminf_{k\to\infty}\,D(\cn_k(\rho)\|\cn_k(\sigma))\le D(\cm(\rho)\|\cm(\sigma))\,,
	\end{align*}
	where the first inequality arises from lower semicontinuity of Umegaki's relative entropy, and the second one holds because $\cn_k\in\cL_\cm$ for all $k$. Therefore the set $\cL_\cm(\rho,\sigma):=\{\cn:\,D(\cn(\rho)\|\cn(\sigma))\le D(\cm(\rho)\|\cm(\sigma)) \}$ is closed. Since $\cL_\cm=\bigcap_{\rho,\sigma}\mathcal{L}_\cm(\rho,\sigma)$, $\cL_\cm$ itself is closed as an intersection of closed sets. The closure of $\cD_{\cm}$ follows similarly: consider a sequence of channels $\cn_k$ such that $\cn_k\to\cn$. Since for each $k$, $\cm_k=\Lambda_k\circ\cm$ for some $\Lambda_k$
	belonging to the compact set of quantum channels, there exists a subsequence $\Lambda_{k_m}$ that converges by sequential compactness: $\Lambda_{k_m}\to\Lambda$. Hence $\cn\in\cD_{\cm}$ since $\cn_{k_m}=\Lambda_{k_m}\circ\cm\to \Lambda\circ\cm$, proving closedness of $\cD_{\cm}$. The convexity of the sets is obvious from the convexity of the relative entropy for $\cL_{\cm}$ and the linearity of the defining condition for degradability. Invariance under output channels is obvious by DPI for $\cL_\cm$ and using the direct definition for $\cD_\cm$.
\end{proof}

\subsection{Contraction coefficients and the strong data processing inequality}\label{lessnoisySDPI}

The central objects of this section are contraction coefficients and strong data processing constants. Given a fixed state $\sigma\in\cD(\cH_A)$ and a quantum channel $\cn:\cB(\cH_A)\to \cB(\cK_B)$, define the \textit{relative entropy} strong data processing inequality (SDPI) \textit{constants} \cite{ahlswede1976spreading,hiai2016contraction}: 
\begin{align}
\eta_{\operatorname{Re}}(\cn,\sigma):=\sup_{\substack{\rho\in\cD(\cH_A)\\0<D(\rho\|\sigma)<\infty}}\frac{D(\cn(\rho)\|\cn(\sigma))}{D(\rho\|\sigma)},
\end{align}
and the contraction coefficient is then $\eta_{\operatorname{Re}}(\cn):=\sup_{\sigma\in\cD(\cH)}\eta_{\operatorname{Re}}(\cn,\sigma)$.

It turns out that these coefficients and the notion of (regularized/completely) less noisiness $\succeq_{\operatorname{l.n.}}^{\operatorname{(reg/c)}}$ can be understood on a very similar footing, allowing us to treat them with the same tools by expressing the partial orders in terms of relative entropies and the contraction coefficients in terms of mutual information. The key argument here is the following proposition that is a quantum generalization of  \cite[Proposition 4]{watanabe2012private}:
\begin{proposition}\label{propDmutualinfo}
Given two channels $\cm\in\CPTP(\cH_A,\cK_B)$ and $\cn\in\CPTP(\cH_A,\cK_{B'})$, $n\in\NN$, $\sigma_{A}\in\cD(\cH_A)$ and $\eta\ge 0$, the following are equivalent:
	\begin{itemize}
		\item[(i)] For all $\operatorname{c-q}$ state $\rho_{UA}$ with marginal $\sigma_{A}=\tr_U(\rho_{UA})$, where $U$ is an arbitrary classical system, 
		\begin{align*}
		\,\eta \,I(U:B)_{(\id_{U}\otimes \cm)(\rho_{UA})}\ge I(U:B')_{(\id_{U}\otimes \cn)(\rho_{UA})}\,.
		\end{align*}
		\item[(ii)] For any state $\rho_{A}\in\cD(\cH_{A})$ with $\supp(\rho_{A})\subseteq\supp(\sigma_{A})$,
		\begin{align}\label{charactrelativeent}
		\eta\,D(\cm(\rho_{A})\|\cm(\sigma_{A}))\ge		D(\cn(\rho_{A})\|\cn(\sigma_{A}))\,.
		\end{align}
		\end{itemize}
\end{proposition}

\begin{proof}
We postpone the proof to \Cref{fdivergence}, where a more general equivalence between $f$-divergences and the corresponding mutual information quantities is proven in Proposition~\ref{propDmutualinfof}.
\end{proof}

Note that the above Proposition also applies to our regularized quantities by choosing $\cm'=\cm^{\otimes n}$ and to the complete ones by choosing $\cm'=\cm\otimes\id$.
A direct consequence is that, $\cm\succeq_{\operatorname{l.n.}}\cn$ if and only if \Cref{charactrelativeent} holds for $\eta=1$, whereas $\cm\succeq_{\operatorname{l.n.}}^{\operatorname{reg}}\cn$ if and only if \Cref{charactrelativeent} holds for $\cm^{\otimes n}$ and all $n\in\NN$  and $\eta=1$. Finally, $\cm\succeq_{\operatorname{l.n.}}^{\operatorname{c}}\cn$  if and only if \Cref{charactrelativeent} holds for $\cm\otimes\id$ and $\eta=1$.

 The next result is another direct consequence of the above Proposition and provides us with an alternative way to express $\eta_{\operatorname{Re}}(\cn,\sigma)$ in terms of the mutual information.
\begin{lemma}
\begin{align}\label{mutualinfo-etarel}
\eta_{\operatorname{Re}}(\cn,\sigma)=\sup_U\sup_{\substack{\rho_{UA}\\0<I(U:A)<\infty \\ \tr_U(\rho_{UA})=\sigma}}\frac{I(U:B)_{(\id_U\otimes \cn)(\rho_{UA})}}{I(U:A)_{\rho_{UA}}}.
\end{align}
\end{lemma}
\begin{proof}
The proof follows directly from Proposition~\ref{propDmutualinfo}. 
\end{proof}

The above result suggests the introduction of the following so-called \textit{less noisy domination factor} \cite{makur2018comparison}: given two channels $\cm\in\CPTP(\cH_A,\cK_B)$ and $\cn\in\CPTP(\cH_A,\cK_{B'})$, and a state $\sigma\in\cD(\cH_A)$:
\begin{align*}
\eta_{\operatorname{Re}}(\cm,\cn,\sigma):=\sup_{\substack{\rho\in\cD(\cH_A)\\0<D(\cm(\rho)\|\cm(\sigma))<\infty}}\frac{D(\cn(\rho)\|\cn(\sigma))}{D(\cm(\rho)\|\cm(\sigma))}\equiv \sup_U\sup_{\substack{\rho_{UA}\\0<I(U:B)<\infty \\ \tr_U(\rho_{UA})=\sigma}}\,\frac{I(U:B')}{I(U:B)}\,.
\end{align*}
and $\eta_{\operatorname{Re}}(\cm,\cn):=\sup_\sigma\eta(\cm,\cn,\sigma)$. In the case when $\cm=\id$, we retrieve the usual SDPI constants of $\cn$. In particular, from \Cref{propDmutualinfo} we directly have $$\eta_{\operatorname{Re}}(\cm,\cn)\le 1~~~\Leftrightarrow~~~ \cm\succeq_{\operatorname{l.n.}}\cn\,.$$ 
One can of course also define complete and regularized versions of these coefficients. However, it can be easily seen that the \textit{complete} contraction coefficient $\eta^c_{\operatorname{Re}}(\cn)$ is equal to $1$ for any channel $\cN$, which limits its usefulness severely. 

Finally, we want to end this section by pointing out some relations between the less noisy order in relation to erasure channels and contraction coefficients which will serve as an example of the close relation of the two concepts and reappear throughout this work. It is immediately clear from the above discussion that we have
\begin{align}
\cm \sln \cn &\Longrightarrow \eta_{\operatorname{Re}}(\cm,\sigma) \leq \eta_{\operatorname{Re}}(\cn,\sigma)~~~\forall\sigma \\
&\Longrightarrow\; \eta_{\operatorname{Re}}(\cm) \leq \eta_{\operatorname{Re}}(\cn)\,.
\end{align}
Now consider $\cm$ to be the quantum erasure channel $\cm^{\operatorname{er}}_{\eta_e}$ of erasure parameter $\eta_e$: $\cH_{B}=\cH_A\oplus \CC$ and for any $X\in\cB(\cH_A)$,
\begin{align}
\cm^{\operatorname{er}}_{\eta_e}(X):=(1-\eta_e)X+\eta_e|e\rangle\langle e|\,,
\end{align}
where $|e\rangle\perp\cH_A$. 

The following is a direct consequence of \Cref{propDmutualinfo}:

\begin{proposition}\label{Prop:erasureContraction}
	Let $\cn\in\CPTP(\cH_A,\cK_B)$ and $\eta_e\in[0,1]$. Then the following are equivalent:
	\begin{itemize}
	\item[(i)] $\cm^{\operatorname{er}}_{\eta_e}\succeq_{\operatorname{l.n.}}\cn$.
		\item[(ii)] $\eta_{\operatorname{Re}}(\cn)\le \eta_{\operatorname{Re}}(\cm^{\operatorname{re}}_{\eta_e})=(1-\eta_e)$.
	\end{itemize}
\end{proposition}

\begin{proof}
	First, a simple calculation gives the following (see \cite[Proposition 21.6.1]{wilde2013quantum}):
\begin{align*}
I(U:B)_{(\id_U\otimes \cm^{\operatorname{er}}_{\eta_e})(\rho_{UA})}=(1-\eta_e)\,I(U:A)_{\rho_{UA}}\,.
\end{align*}
The result then follows directly from \Cref{propDmutualinfo}.
\end{proof}

Thus, we conclude that a bound on the relative entropy contraction coefficient is equivalent to a less noisy relation for the corresponding channel when compared to the erasure channel. In the next section we will discuss further properties of the concepts introduced in this section.

\section{Properties and Bounds}\label{Sec:PropertiesAndBounds}

In this section we discuss properties of partial orders and contraction coefficients. First we will discuss tensorization properties of different partial orders and their relation with each other, then we discuss bounds on contraction coefficients and finally we illustrate our results with several simple examples. 

\subsection{Tensorization and relationships  of partial orders}
In the classical setting, it is a well-known fact that the notion of less noisiness \textit{tensorizes} (see e.g.~\cite[Proposition 16]{polyanskiy2017strong},\cite[Proposition 5]{sutter2014universal}): Given channels $\cm_1\succeq_{\operatorname{l.n.}}\cn_1$ and $\cm_2\succeq_{\operatorname{l.n.}}\cn_2$, 
\begin{align*}
\cm_1\otimes\cm_2\succeq_{\operatorname{l.n.}}\cn_1\otimes\cn_2\,.
\end{align*}
In particular, we have as expected that the notions of less noisy and regularized less noisy are equivalent. This is false in general in the quantum case, as we will also discuss in more details further down. Often tensorization properties can be recovered when allowing for an additional reference system. The same is true in the case of the less noisy ordering and we show in the following that the completely less noisy ordering tensorizes. 
\begin{lemma}[The complete less noisy preorder tensorizes]\label{lemmatensorcomp}
	 For any four channels with $\cm_1\succeq_{\operatorname{l.n.}}^{\operatorname{c}}\cn_1$ and $\cm_2\succeq_{\operatorname{l.n.}}^{\operatorname{c}}\cn_2$, it follows that $$\cm_1\otimes\cm_2\succeq_{\operatorname{l.n.}}^{\operatorname{c}}\cn_1\otimes\cn_2\,.$$
\end{lemma}
\begin{proof}
	Assume that $\cm_1\succeq_{\operatorname{l.n.}}^{\operatorname{c}}\cn_1$ and $\cm_2\succeq_{\operatorname{l.n.}}^{\operatorname{c}}\cn_2$. Therefore, for any two tripartite states $\rho_{RA_1A_2},\sigma_{RA_1A_2}$, we have
	\begin{align*}
	D((\id_{R}\otimes\cm_1\otimes \cm_2)(\rho_{RA_1A_2})&\|(\id_{R}\otimes\cm_1\otimes \cm_2)(\sigma_{RA_1A_2}))\\
	&\ge D((\id_{R}\otimes\cn_1\otimes \cm_2)(\rho_{RA_1A_2})\|(\id_{R}\otimes\cn_1\otimes \cm_2)(\sigma_{RA_1A_2}))\\
	&\ge D((\id_{R}\otimes\cn_1\otimes \cn_2)(\rho_{RA_1A_2})\|(\id_{R}\otimes\cn_1\otimes \cn_2)(\sigma_{RA_1A_2}))
	\end{align*}
	where we successively used that $\cm_1\succeq_{\operatorname{l.n.}}^{\operatorname{c}}\cn_1$ for the reference system $RB_2$  and $\cm_2\succeq_{\operatorname{l.n.}}^{\operatorname{c}}\cn_2$ for the reference system $RB_1$.
\end{proof}

The previous result shows in particular that the notion of complete less noisiness is more stringent than the one of (regularized) less noisiness, but potentially less than that of degradation. We will now further investigate the implications between the different partial orders. The relations are also summarized in Figure~\ref{Fig:RelationsPO}. 
\begin{proposition}\label{inclusions}
	Given two channels $\cn\in\CPTP(\cH_A,\cH_B)$  and $\cm\in\CPTP(\cH_A,\cH_{B'})$
	
\begin{align}
	&\cn\succeq_{\operatorname{deg}}\cm~~\Rightarrow~~\cn\cln\cm~~\Rightarrow~~\cn\rln\cm~~\Rightarrow~~\cn\sln\cm\,, \label{eqn:relations1}\\
	\label{eqn:relations2}
&\cn\succeq_{\operatorname{deg}}\cm~~\Rightarrow~~\cn\cln\cm~~\Rightarrow~~\cn\fqln\cm~~\Rightarrow~~\cn\sln\cm\,.
\end{align}
\end{proposition}
\begin{proof}
We first prove Equation~\eqref{eqn:relations1}: The last implication is obvious and was already discussed  in \Cref{chanpreorder}. The first implication follows by the data processing inequality: Assume that $\cn\succeq_{\operatorname{deg}}\cm$. Therefore, by definition there exists a channel $\Theta$ such that $\cm=\Theta\circ\cn$, which also implies that for any reference system $R$, $\id_R\otimes \cm=\id_R\otimes \Theta\circ \cn$. Then for any two states $\rho_{AR},\sigma_{AR}$
	\begin{align*}
	D((\id_R\otimes\cm)(\rho_{AR})\|(\id_R\otimes \cm)(\sigma_{AR}))&=
	D((\id_R\otimes\Theta\circ\cn)(\rho_{AR})\|(\id_R\otimes \Theta\circ\cn)(\sigma_{AR}))\\
	&\le D((\id_R\otimes\cn)(\rho_{AR})\|(\id_R\otimes \cn)(\sigma_{AR}))\,.
	\end{align*}
	where the inequality follows by DPI for the channel $\id\otimes \Theta$. Therefore $\cn\cln\cm$. Finally, that then $\cn\rln\cm$ is a simple consequence of  \Cref{lemmatensorcomp} and the characterization of $\rln$ in \Cref{propDmutualinfo}. 
	
For Equation~\eqref{eqn:relations2} the first implication was already shown and the last is obvious because classical-quantum states are a special case of general quantum states. It remains to show the middle one. Using Proposition~\ref{propDmutualinfo} we can express the condition for $\cn\cln\cm$ as one of relative entropies involving two arbitrary states with reference system. Now note that, for a state $\rho_{AB}$, we have 
\begin{align*}
I(A:B) = D(\rho_{AB} \| \rho_A \otimes \rho_B).
\end{align*}
Using this to rewrite the condition for $\cn\fqln\cm$ it becomes clear that it is indeed a special case of the completely less noisy relation. 
\end{proof}
It remains an interesting open problem whether any relation between the regularized and the fully quantum less noisy ordering can be found. 

Sometimes, further relationships between partial orders can be uncovered when one restricts to the commonly used special case when $\cm = \cn^c$ is the complementary channel of $\cN$\footnote{For detailed definitions, we refer to \cite{wilde2013quantum}.}. One such example is given by the following proposition.
\begin{proposition} 
Given a channel $\cn\in\CPTP(\cH_A,\cH_B)$  and its complementary channel  $\cn^c\in\CPTP(\cH_A,\cH_{B'})$ one has
\begin{align}
\cn\cmc\cn^c~~\Leftrightarrow~~\cn\mc\cn^c. 
\end{align}
\end{proposition}
\begin{proof}
Note that the coherent information of a pure state $\Psi_{ABE}$ can be written as a difference of two mutual informations with respect to a classical quantum state $\rho_{XBE}$ (see e.g.~\cite[Theorem 13.6.1]{wilde2013quantum}):
\begin{align}
I(A\rangle B) = I(X:B) - I(X:E).
\end{align}
We can use this in the following chain of equivalences,
\begin{align}
&&I( A : B ) &\leq I(A : E ) \\
&\Leftrightarrow& I( A \rangle B ) &\leq I(A \rangle E )  \\
&\Leftrightarrow& I(X:B) - I(X:E) &\leq I(X:E) - I(X:B)   \\
&\Leftrightarrow& I(X:B)  &\leq I(X:E),  
\end{align}
which holds for all pure input states and therefore the first line corresponds to the fully quantum more capable condition and the last line to the usual more capable condition. 
\end{proof}
It is currently unclear whether something similar holds for an arbitrary second channel $\cm$.

We will now briefly discuss whether the less noisy ordering tensorizes. It does in the classical case~\cite[Proposition 16]{polyanskiy2017strong},\cite[Proposition 5]{sutter2014universal}. In the quantum setting tensorzation does not hold. Namely, we can find $\cm_1,\cm_2,\cn_2,\cn_2$ such that $\cm_1\succeq_{\operatorname{l.n.}}\cm_2$, $\cn_1\succeq_{\operatorname{l.n.}}\cn_2$ but $\cm_1\otimes \cm_2\nsucceq_{\operatorname{l.n.}}\cn_1\otimes\cn_2$. This follows directly from superactivation of the private capacity. If we have $P^{(1)}(\cn_{A\rightarrow BE}) = 0$, it follows that the channel to Eve is less noisy than that to Bob. On the other hand, from $P^{(1)}(\cn^{\otimes 2}_{A\rightarrow BE}) > 0$, it follows that this order does not carry over to the two copy case. Note that the additivity of the private information questions has been recently considered in~\cite{tikku2020non}, showing that when the sender is quantum but the receivers are classical additivity still holds. However, when either one of the receivers is quantum, additivity is violated, even when the sender is classical. Since the counterexamples in~\cite{tikku2020non} exhibit the desired superactivation feature, the results directly carry over to the tenzorization problem for the less noisy order. 

It should be noted that all the examples above are based on wiretap channels that are not necessarily isometric. It would be interesting to find similar examples in the isometric case, i.e. where the channel to Eve is the complement of that to Bob. From the same argument, it can easily be seen that anti-less noisy channels have $P^{(1)}(\cn)=0$, however anti-regularized less noisy is needed to ensure that $P(\cn)=0$. Another consequence is that the less noisy ordering and the completely less noisy ordering (and therefore the degradable ordering) cannot be equivalent. An interesting open problem is whether the regularized or complete less noisy ordering could still be equivalent to the degradable ordering. 

From the above it also follows that when checking for the regularized less noisy order one has to check $n>1$. On the other hand, one can show that it is sufficient to check the condition in the asymptotic limit. 
\begin{lemma}
The following two conditions are equivalent: 
\begin{enumerate}
\item  $\cn\rln\cm$. 
\item For all $\operatorname{c-q}$ states $\rho_{UA^n}$,
\begin{align*}
\lim_{n\rightarrow\infty} I(U:B^n)_{(\id_{U}\otimes \cm^{\otimes n})(\rho_{UA^n})}\ge \lim_{n\rightarrow\infty} I(U:B'^n)_{(\id_{U}\otimes \cn^{\otimes n})(\rho_{UA^n})}\,.
\end{align*} 
\end{enumerate}
\end{lemma}
\begin{proof}
Since regularized less noisy requires the mutual information condition to hold for all $n$ the direction $(1) \Rightarrow (2)$ is obvious. The other direction follows by noting that if the condition holds for some $n$ it also holds for $n-1$: This is simply because if the condition holds for all states $\rho_{UA^n}$ it in particular also holds for those of the form $\rho_{UA^n}=\rho_{UA^{n-1}}\otimes\rho_{A_n}$, for which we have 
\begin{align}
I(U:B^n) = I(U:B^{n-1}). 
\end{align}
This directly implies the result for $n-1$ and therefore the claim follows by starting at $n=\infty$ and stepwise reducing $n$. 
\end{proof}

We could also ask for a weaker requirement than tensorization. We call a partial order $\succeq$ \textit{tensor stable} if for all channels the following holds: If $\cn_1 \succeq \cn_2$ then $\cn_1\otimes\cm \succeq \cn_2\otimes\cm$ for all $\cM$. Every partial order that tensorizes is automatically tensor stable, however also those that do not tensorize can potentially still possess this property. We will show that for the less noisy ordering, while the property holds for some channels, in general the order is not tensor stable. Again an example based on erasure channels can serve as a primer. 
\begin{lemma}
Fix any two erasure channels $\ce_1$ and $\ce_2$ with erasure probabilities $\epsilon_1$ and $\epsilon_2$, respectively. For any channel $\cm$ we have the following:
\begin{align}
\ce_1 \succeq_{\operatorname{l.n.}} \ce_2\quad\Rightarrow\quad\ce_1\otimes\cm \succeq_{\operatorname{l.n.}} \ce_2\otimes\cm. 
\end{align}
\end{lemma}
\begin{proof}
First note that $\ce_1 \succeq_{\operatorname{l.n.}} \ce_2$ is clearly equivalent to $\epsilon_1\leq\epsilon_2$. We now need to show that 
\begin{align}
I(U:E_1B) \ge I(U:E_2B)
\end{align}
We are using
\begin{align}
I(U:E_iB) = (1-\epsilon_i) I(U:A_iB) + \epsilon_i I(U:B),  
\end{align}
which follows essentially from~\cite[Eqn. (20.90) -- (20.94)]{wilde2013quantum}, to show that the above is equivalent to
\begin{align}
0\ge \left(\epsilon_1 - \epsilon_2\right) \left(I(U:A_1B) - I(U:B)\right), 
\end{align}
which holds since the first term in the product is negative by assumption and the second positive by the data processing inequality. This concludes the proof. 
\end{proof}

A direct implication is that for two erasure channels,
\begin{align}
\ce_1 \succeq_{\operatorname{l.n.}} \ce_2\quad\Rightarrow\quad\ce_1 \cln \ce_2. 
\end{align}
The same would generalize to all channels if less noisy were tensor stable in general. In that case, the less noisy and the completely less noisy orderings would be equivalent, which we know they are not. Therefore, the less noisy ordering cannot be tensor stable in general.

\subsection{Bounds on contraction coefficients}
In this section, we will discuss several properties and bounds of SDPI constants and contraction coefficients. 
Let's start with a simple bound for concatenated channels.
\begin{lemma}\label{Lemma:UpperConcatenated}
We have
\begin{align}
&\eta_\Re(\cm\circ\cn, \sigma) \leq \eta_\Re(\cm,\cn(\sigma)) \,\,\eta_\Re(\cn, \sigma), \\
&\eta_\Re(\cm\circ\cn) \leq \eta_\Re(\cm) \,\eta_\Re(\cn). 
\end{align}
\end{lemma}
\begin{proof}
Observe, 
\begin{align}
\eta_\Re(\cm\circ\cn, \sigma) = \frac{D(\cm\circ\cn(\rho)\|\cm\circ\cn(\sigma))}{D(\rho\|\sigma)} = \frac{D(\cm\circ\cn(\rho)\|\cm\circ\cn(\sigma))}{D(\cn(\rho)\|\cn(\sigma))} \,\frac{D(\cn(\rho)\|\cn(\sigma))}{D(\rho\|\sigma)}\,, 
\end{align}
assuming that $\rho$ is the optimizing state, from which the first claim follows directly. The second then follows by taking the supremum over $\sigma$ on both sides. 
\end{proof}
A common approach to bounding capacities is to consider flagged channels. This approach also translates to contraction coefficients. Let $\cn = \sum_i \lambda_i \,\cn_i$ for some positive numbers $\lambda_i$ and completely positive maps $\cN_i$,  and define
\begin{align}
\widehat\cn = \sum_i\,\lambda_i\, \cn_i \otimes |i\rangle\langle i|. 
\end{align}
\begin{lemma}\label{Lemma:UpperFlag}
For all states $\sigma$, we have
\begin{align}
\eta_\Re(\cn,\sigma) \leq \eta_\Re(\widehat\cn,\sigma) \leq  \sum_i \lambda_i \,\eta_\Re(\cn_i,\sigma). 
\end{align}
In fact, the two above bounds are generally true for any contraction coefficient based on a divergence satisfying the data processing inequality. The first inequality also extends to the case of non-trace preserving $\cN_i$.
\end{lemma}
\begin{proof}
The first inequality follows from data processing. The second from the properties of flagged channels and splitting the supremum.
\end{proof} 
Note that, in general, the flags do not have to be orthogonal and one could define an extension 
\begin{align}
\widetilde\cn = \sum_i  \lambda_i\,\cn_i \otimes \tau_i, 
\end{align}
leading to a potentially tighter bound
\begin{align}
\eta_\Re(\cn,\sigma) \leq \eta_\Re(\widetilde\cn,\sigma). 
\end{align}
This approach has recently been employed to find very tight bounds on quantum capacities~\cite{fanizza2019quantum, wang2019optimizing}.

We will now discuss whether contraction coefficients tensorize and give some simples examples to the contrary. We start with the contraction coefficient $\eta_\Re(\cm)$. Our counterexample is going to be based on the erasure channel. As mentioned before, we have
\begin{align}
\eta_\Re(\ce) = 1-\epsilon. 
\end{align}
Our goal will be to show that $\eta_\Re(\ce) \neq \eta_\Re(\ce\otimes\ce)$. First, we will see that the contraction coefficient cannot get smaller under tensorization. 
\begin{lemma}\label{Lemma:LowerBoundCC}
The following holds for any two quantum channels:
\begin{align*}
	\eta_{\operatorname{Re}}(\cm\otimes\cn) \geq \max\{\eta_\Re(\cm),\,\eta_\Re(\cn)\}\,.
\end{align*}
\end{lemma}
\begin{proof}
Recall that 
\begin{align}
\eta_{\Re}(\cm):=\sup_{\rho_{UA}}\frac{I(U:B)_{(\id_U\otimes\cm)(\rho_{UA})}}{I(U:A)_{\rho_{UA}}}\,.
\end{align}
Pick $\rho_{UA_1}$ as the optimizer in $\eta_{\Re}(\cm)$. With $\sigma_{UA_1A_2}=\rho_{UA_1}\otimes\rho_{A_2}$ with arbitrary $\rho_{A_2}$, we immediately have
\begin{align}
	\eta_{\operatorname{Re}}(\cm\otimes\cn) \geq \eta_\Re(\cm)
\end{align}
and
\begin{align}
	\eta_{\operatorname{Re}}(\cm\otimes\cn) \geq \eta_\Re(\cn)
\end{align}
follows by symmetry. 
\end{proof}
Upper bounds are more involved, but one can give a nice one when one channel is an erasure channel. 
\begin{lemma}\label{Lemma:ErasureUpperBound}
For all quantum channel $\cM$, we have
\begin{align}
\eta_{\operatorname{Re}}(\cm\otimes\ce) \leq (1-\epsilon) + \epsilon\,\eta_{\Re}(\cm) \leq (1+\epsilon)\max\{\eta_{\Re}(\ce),\eta_{\Re}(\cm)\}.
\end{align}
\end{lemma}
\begin{proof}
The main ingredient is the following equality
\begin{align}
I(U:B_1B_2) = (1-\epsilon) I(U:B_1A_2) + \epsilon I(U:B_1),  \label{Eq:MutualInfoErasure}
\end{align}
which follows essentially from~\cite[Eqn. (20.90) -- (20.94)]{wilde2013quantum}. We have, 
\begin{align}
\frac{I(U:B_1B_2)}{I(U:A_1A_2)} &= (1-\epsilon) \frac{I(U:B_1A_2)}{I(U:A_1A_2)} + \epsilon \frac{I(U:B_1)}{I(U:A_1A_2)} \\
&\leq (1-\epsilon) + \epsilon \frac{I(U:B_1)}{I(U:A_1)} \\
&\leq (1-\epsilon) + \epsilon \eta_{\operatorname{Re}}(\cm), 
\end{align}
where the first inequality follows from applying data-processing twice and the second by taking the supremum over all $\rho_{UA_1}$. 

Now the second inequality follows by noting that 
\begin{align}
(1-\epsilon) + \epsilon \,\eta_{\operatorname{Re}}(\cm) = \eta_{\Re}(\ce) + \epsilon\, \eta_{\operatorname{Re}}(\cm)
\end{align}
and considering the two cases $\eta_{\Re}(\ce)\geq\eta_{\Re}(\cm)$ and $\eta_{\Re}(\cm)\geq\eta_{\Re}(\ce)$.
\end{proof}
This allows us to calculate the contraction coefficient $\eta_\Re(\ce^{\otimes n})$. 
\begin{lemma}
Given the erasure channel $\cE$ of parameter $\eps$, we have
\begin{align}
\eta_\Re(\ce^{\otimes n}) = 1-\epsilon^n.
\end{align}
\end{lemma}
\begin{proof}
The $\leq$ direction follows from Lemma~\ref{Lemma:ErasureUpperBound}. Set $\cm=\ce^{\otimes n-1}$ and define a function $f(n)$ by $f(1)=1-\epsilon$ and $f(n) = (1-\epsilon) + \epsilon f(n-1)$ for $n\geq 2$. It can easily be shown by induction that $f(n)= 1-\epsilon^n$ and therefore we have $\eta_\Re(\ce^{\otimes n}) = 1-\epsilon^n$. 

 For the opposite direction we have to show that this bound is indeed achievable. Consider the following state
\begin{align}
\rho_{UA_1A_2} = \frac12 |0\rangle\langle 0|_U \otimes |0\rangle\langle 0|_{A_1^n} + \frac12 |1\rangle\langle 1|_U \otimes |1\rangle\langle 1|_{A_1^n}. \label{Eqn:InputStates}
\end{align}
It can be checked that 
\begin{align}
I(U:A_1A_2) &= \log2 \\ 
H(U) &= \log2 \\
H(UB_1B_2) - H(B_1B_2) &= \epsilon^n\log2  \label{entropy-diff}\\
\Rightarrow I(U:B_1B_2) &= (1-\epsilon^n)\log2
\end{align}
where the only non-trivial one is Equation~\ref{entropy-diff}, which follows from observing that the only overlap in support between the two output states of the erasure channels corresponds to the state being completely deleted. From this the desired result follows immediately. 
\end{proof}
This completes our argument that the contraction coefficient does not tensorize. 

However, $\eta_{\Re}(\cm,\sigma)$ might still tensorize. Indeed, this is the case for erasure channels. Consider two erasure channels $\ce_1$ and $\ce_2$ with erasure probability $\epsilon_1$ and $\epsilon_2$, respectively. By Lemma~\ref{Lemma:ErasureUpperBound} we immediately get
\begin{align}
\eta_{\Re}(\ce_1\otimes\ce_2) \leq 1 - \epsilon_1\epsilon_2. 
\end{align}
Opposed to that, in the restricted case, tensorization holds as can be seen in the following lemma. 
\begin{lemma}
For the erasure channels defined above and any two states $\sigma_1$ and $\sigma_2$, we have
\begin{align}
\eta_{\Re}(\ce_1\otimes\ce_2, \sigma_1\otimes\sigma_2) \leq \max\{(1-\epsilon_1),(1-\epsilon_2)\}. 
\end{align}
\end{lemma}
\begin{proof}
As before, we can use Equation~\ref{Eq:MutualInfoErasure}, however, this time we apply it twice and we get
\begin{align}
I(U:B_1B_2) = (1-\epsilon_1)(1-\epsilon_2) I(U:A_1A_2) + (1-\epsilon_2)\epsilon_1 I(U:A_1) + (1-\epsilon_1)\epsilon_2 I(U:A_2).
\end{align}
The second main ingredient is the following argument: 
\begin{align}
I(U:A_1A_2) &= I(U:A_2) + I(U:A_1|A_2) \\
&=  I(U:A_2) + I(UA_2:A_1) - I(A_2:A_1) \\
&=  I(U:A_2) + I(UA_2:A_1) \\
&\geq I(U:A_2) + I(U:A_1), 
\end{align}
where the first two inequalities follow by the chain rule, the third because $\rho_{A_1A_2}=\sigma_1\otimes\sigma_2$ and the inequality by data processing. 

Now, we have to distinguish two cases. First,  consider $\epsilon_1\geq\epsilon_2$. We have, 
\begin{align}
&(1-\epsilon_1)(1-\epsilon_2) I(U:A_1A_2) + (1-\epsilon_2)\epsilon_1 I(U:A_1) + (1-\epsilon_1)\epsilon_2 I(U:A_2) \\
&\leq (1-\epsilon_1)(1-\epsilon_2) I(U:A_1A_2) + (1-\epsilon_2)\epsilon_1 [I(U:A_1A_2)- I(U:A_2)] + (1-\epsilon_1)\epsilon_2 I(U:A_2) \\
&= (1-\epsilon_1) I(U:A_1A_2) + (\epsilon_1-\epsilon_2) I(U:A_2) \\
&\leq (1-\epsilon_2) I(U:A_1A_2), 
\end{align} 
where the first inequality follows by using the previously derived argument, the equality is simple rearranging and the final inequality is data-processing together with the assumption that $(\epsilon_1-\epsilon_2)\geq 0$. 
It follows directly that in this case 
\begin{align}
\eta_{\Re}(\ce_1\otimes\ce_2, \sigma_1\otimes\sigma_2) \leq (1-\epsilon_2). 
\end{align}
The second part follows directly by assuming $\epsilon_2\geq\epsilon_1$ and exchanging the roles of $A_1$ and $A_2$ in the previous derivation. Putting both cases together concludes the proof. 
\end{proof}
Whether SDPI constants based on the relative entropy with respect to a fixed tensor product state tensorize in general remains an important open problem. 

\subsection{Additional examples} 

First, we recall an example from~\cite{hiai2016contraction} for qubit channels. The original theorem applies to several contraction coefficients, however we state it here for the relative entropy. Recall that any unital qubit channel $\cM$ can be represented by a $3\times 3$ real matrix $T$ describing how the channel acts on the Pauli matrices $\sigma$. This representation is usually referred to as the Bloch sphere representation. We then have:
\begin{lemma}[{\cite[Theorem 6.1]{hiai2016contraction}}]
For any unital map $\cm_T : I + w\cdot\sigma \rightarrow I + (Tw)\cdot\sigma$ where $T$ is a real matrix with $\| T \|_\infty\leq 1$, 
\begin{align}
\eta_\Re(\cm_T ) = \| T \|_\infty^2. 
\end{align}
\end{lemma}

Examples to which this result applies include the depolarizing channel,
\begin{align}
\cd_p(\cdot) = (1-p)(\cdot) + p \frac{\Id}{2},
\end{align}
and the dephasing and bit-flip channels,
\begin{align}
\cd^Z_p = (1-p)(\cdot) + p Z(\cdot)Z, \\
\cd^X_p = (1-p)(\cdot) + p X(\cdot)X,  
\end{align}
We get, 
\begin{align}
\eta(\cd_p) &= (1-p)^2, \\
\eta(\cd^X_p) &= \eta(\cd^Z_p) = 1. 
\end{align}

We can combine this with results from the previous section to test our bounds on products of channels. Let $\ce_{1/2}$ be an erasure channel with erasure probability $\epsilon=\frac12$, using the first upper bound in Lemma~\ref{Lemma:ErasureUpperBound} and lower bounding by fixing the input state as that from Equation~\eqref{Eqn:InputStates}, we get
\begin{align}
\frac{\log2 - \frac12 h(\frac p2)}{\log2} \leq \eta_\Re(\ce_{1/2}\otimes\cd_p) \leq \frac12(1+(1-p)^2).  \label{Eqn:BoundsErasureDep}
\end{align}
These simple bounds already limit the possible values significantly to the green area in Figure~\ref{Fig:ErasureDep}, where we also plot the looser second bound from Lemma~\ref{Lemma:ErasureUpperBound} and the lower bound from Lemma~\ref{Lemma:LowerBoundCC} for comparison. 

\begin{figure}
\centering
\includegraphics[scale=0.4]{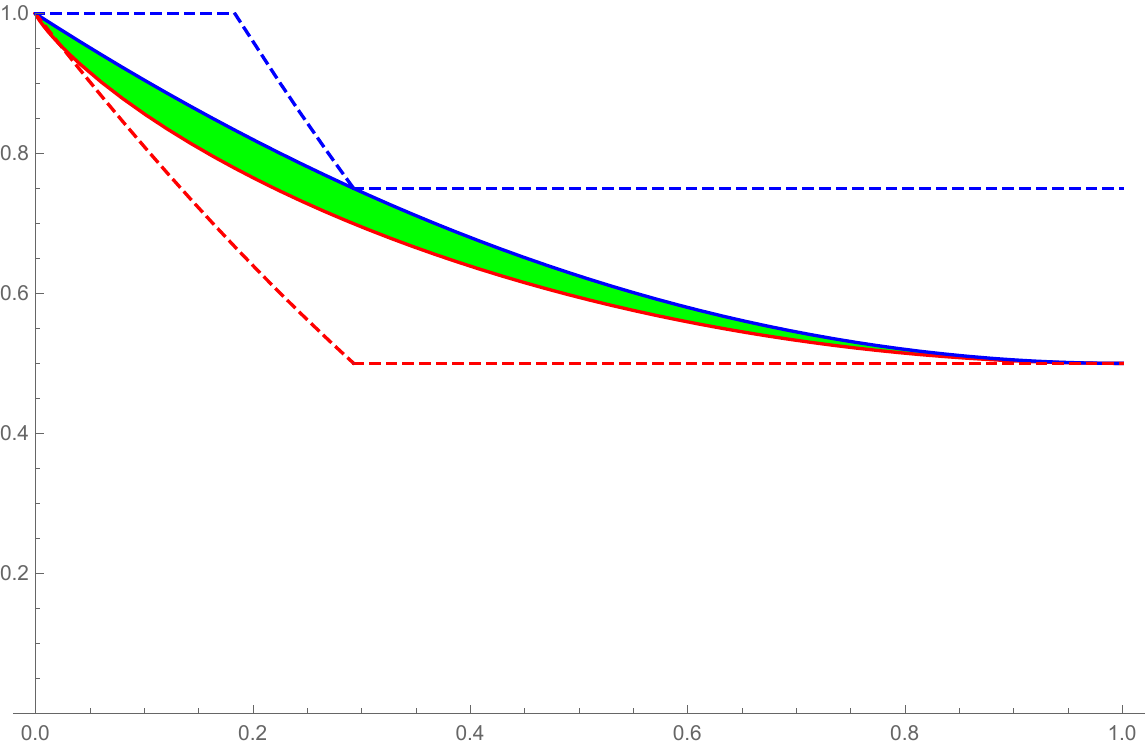}
\caption{\label{Fig:ErasureDep} Bounds on $\eta_\Re(\ce_{1/2}\otimes\cd_p)$ as given in Equation~\eqref{Eqn:BoundsErasureDep}. The green area marks the remaining possible range of the contraction coefficient.}
\end{figure}

Let us consider one more example. A channel that has recently proven useful in investigating the properties of quantum capacities is the dephrasure channel~\cite{leditzky2018dephrasure},
\begin{align}
\cd_{\epsilon, p} = \ce_\epsilon \circ \cd^Z_p = (1-\epsilon) \left( (1-p)(\cdot) + p Z(\cdot)Z)\right) + \epsilon |e\rangle\langle e|. 
\end{align}
We can easily get a lower bound on the contraction coefficient by picking the input state in Equation~\ref{Eqn:InputStates}. An upper bound can either be found from using Lemma~\ref{Lemma:UpperConcatenated} or Lemma~\ref{Lemma:UpperFlag} (for the latter note that the dephrasure channel can directly be written as a flag channel). Together we get, 
\begin{align}
\eta_\Re(\cd_{\epsilon, p}) = 1-\epsilon. 
\end{align}
Finally, a simple observation is that a replacer channel that always outputs a state $\tau$, 
\begin{align}
\cR_\tau(\cdot) = \tau \tr(\cdot)\,,
\end{align}
is the noisiest channel and an isometric channel $\cV(\rho)=V\rho V^\dagger$ for some isometry $V$ is the least noisy channel:
\begin{lemma}
For any channel $\cn$, replacer channel $\cR_\tau$ and isometric channel $\cV$, we have
\begin{align}
&\eta_\Re(\cR_\tau)= 0,  \qquad \eta_\Re(\cV)=1  \\
&\eta_\Re(\cn\otimes\cR_\tau) = \eta_\Re(\cn)   \\
&\cV \deg\cn\deg\cR_\tau\,,
\end{align}
where the final line implies a similar statement for all other partial orders defined here via Proposition~\ref{inclusions}.
\end{lemma}
\begin{proof}
Note that 
\begin{align}
D(\tau\|\tau) &= 0, \\
D(\rho\otimes\tau \| \sigma\otimes\tau) &= D(\rho\|\sigma), \\
D(V\rho V^\dagger\| V\sigma V^\dagger) &= D(\rho\|\sigma),
\end{align}
by additivity, faithfulness and data processing of the relative entropy. The proof of the first two lines follows easily from there. The third line is obvious by choosing in the first step $\cn\circ\cV^{-1}$ and in the second step $\cR_\tau$ as the degrading maps. 
\end{proof}

\section{Approximate partial orders and capacities}

In this section, we explore approximations of the partial orders introduced in \Cref{sec:defs}. These are powerful tools which can be used for instance in approximating  capacities of a quantum channel in terms of the ones of another channel when the latter are simpler to compute. 
\subsection{Approximate partial orders}
We recall that the diamond norm distance between two quantum channels $\cM,\cN\in\operatorname{CPTP}(\cH_A, \cH_B)$ is defined as
\begin{align*}
\|\cM-\cN\|_\diamond:=\sup_{R}\sup_{\rho_{AR}}\,\|(\cN\otimes\id_R-\cM\otimes\id_R)(\rho_{AR})\|_1\,.
\end{align*}
Previously, Sutter et. al.~\cite{sutter2017approximate} defined $\epsilon$-approximate degradation as follows. 
\begin{definition}
	A channel $\cn$ is said to be an $\epsilon$-degraded version of $\cm$ if there exists a channel $\Theta$ such that $\| \cn - \Theta\circ\cm \|_\diamond \leq \epsilon$. 
\end{definition}
Now, one can also define approximate versions of other partial orders, as we will do in the following. 
\begin{definition}
	A channel $\cn$ is said to be \textit{$\epsilon$-completely less noisy} than $\cm$ (denoted $\cn \ecln{\epsilon} \cm$) if 
	\begin{align}
	I(U:ER)_{\cm} \leq I(U:BR)_{\cn} +\epsilon \quad \forall \rho_{UAR}\;, 
	\end{align}
	\textit{$\epsilon$-regularized less noisy} than $\cm$ (denoted $\cn \erln{\epsilon} \cm$) if 
	\begin{align}
	I(U:E^n)_{\cm} \leq I(U:B^n)_{\cn} +n\epsilon \quad \forall \rho_{UA_1^n}\;\forall n\;,
	\end{align}
	\textit{$\epsilon$-less noisy} than $\cm$ (denoted $\cn \esln{\epsilon} \cm$) if 
	\begin{align}
	I(U:E)_{\cm} \leq I(U:B)_{\cn} +\epsilon \quad \forall \rho_{UA}\;,
	\end{align}
	\textit{$\epsilon$-fully quantum less noisy} than $\cm$ (denoted $\cn \efqln{\epsilon} \cm$) if 
	\begin{align}
	I(A:E)_{\cm} \leq I(A:B)_{\cn} +\epsilon \quad \forall \rho_{AA'}\;. 
	\end{align}
One defines similarly the approximate versions of the different \textit{more capable} orders introduced in \Cref{sec:defs}. $\epsilon$-anti orders are defined in the same way by  exchanging $\cn$ and $\cm$. 
\end{definition}
Following the same arguments as in the exact case, see Proposition~\ref{inclusions}, we have immediately that 
\begin{alignat}{2}
\cn \ecln{\epsilon} \cm &\;\Rightarrow\;& \cn \erln{\epsilon} \cm &\;\Rightarrow\; \cn \esln{\epsilon} \cm\,. \label{Eqn:chain} \\
\cn \ecln{\epsilon} \cm &\;\Rightarrow\;& \cn \efqln{\epsilon} \cm &\;\Rightarrow\; \cn \esln{\epsilon} \cm\,. \label{Eqn:chain2}
\end{alignat}

We would of course hope that $\eps$-approximate degradability implies the above orders. This is indeed the case as we will see in a moment. To relate the two partial orders we will make use of the continuity bounds introduced in~\cite{audenaert2007sharp} and~\cite{AFW}. In what follows, $h(x):=-x\ln x-(1-x)\ln (1-x)$, $x\in[0,1]$, denotes the binary entropy. 
\begin{lemma}[Continuity bound, Theorem 1 in~\cite{audenaert2007sharp}]\label{Hsharp}
For quantum states $\rho_{A}$ and $\sigma_{A}$, if $\frac12\| \rho - \sigma\|_1 \leq \epsilon \leq 1$, then
\begin{align}
| H(A)_\rho - H(A)_\sigma | \leq \epsilon\log\left( |A|-1\right) + h\left(\epsilon\right)\,.
\end{align}
\end{lemma}
\begin{lemma}[Alicki-Fannes-Winter bound, Lemma 2 in~\cite{AFW}]\label{AFW}
For quantum states $\rho_{AB}$ and $\sigma_{AB}$, if $\frac12\| \rho - \sigma\|_1 \leq \epsilon \leq 1$, then
\begin{align}
| H(A|B)_\rho - H(A|B)_\sigma | \leq 2\epsilon\log|A| + (1+\epsilon) \,h\Big(\frac{\epsilon}{1+\epsilon}\Big)\,,
\end{align}
where we recall that the conditional entropy of a state $\rho_{AB}$ is defined as $H(A|B)_\rho:=H(AB)_\rho-H(B)_\rho$.
\end{lemma}
This allows us to directly show the following relation. 
\begin{lemma}\label{Lemma:eps-c-LN}
If $\cM:A\to B$ is an $\epsilon$-degraded version of $\cn:A\to B'$ then $\cn$ is $\tilde\epsilon$-completely less noisy than $\cm$, with $$\tilde\epsilon=2\epsilon\log|B| + (2+\epsilon)\,h\Big(\frac{\epsilon}{2+\epsilon}\Big)\,.$$ 
Therefore,
\begin{align*}
\|	\cM-\Theta\circ \cN\|_\diamond\le \eps\;\Rightarrow\;\cn \ecln{\tilde\epsilon} \cm \;\Rightarrow\; \cn \erln{\tilde\epsilon} \cm \;\Rightarrow\; \cn \esln{\tilde\epsilon} \cm\,.
	\end{align*}
\end{lemma}
\begin{proof}
We assume that $\|\cM-\Theta\circ \cN\|_\diamond\le \eps$ and start with an input state $\rho_{UAR}$. Define $\sigma_{UBR}= (\Theta\circ\cn)\otimes \id(\rho_{UAR})$ and $\tau_{UBR}= \cm\otimes \id(\rho_{UAR})$. First, observe the following:
	\begin{align*}
I(U:BR)_{\cM(\rho)}-I(U:B'R)_{\cN(\rho)}&\le 	I(U: BR)_\tau-I(U : BR)_\sigma \\
&=  -H(BR)_\sigma+H(UBR)_\sigma - H(UBR)_\tau+H(BR)_\tau  \\
	&= -H(B|R)_\sigma +H(B|R)_\tau   - H(B|UR)_\tau + H(B|UR)_\sigma \\
	&\leq 2\epsilon\log|B| + (2+\epsilon)\,h\Big(\frac{\epsilon}{2+\epsilon}\Big) \\
	&=: \tilde\epsilon\,,
	\end{align*}
	where the first inequality follows by data-processing to remove the degrading map. The first equality follows by the definition of the mutual information, the second by adding a zero twice and the fact that $\tau_{UR}=\sigma_{UR}$ and the last inequality by applying Lemma~\ref{AFW} twice. The second claim follows from \eqref{Eqn:chain}.
\end{proof}

In the same manner, one can go via the fully quantum less noisy order. 
\begin{lemma}\label{Lemma:epsFQLN1}
If $\cM:A\to B$ is an $\epsilon$-degraded version of $\cn:A\to B'$ then $\cn$ is $\hat\epsilon$-fully quantum less noisy than $\cm$, with $$\hat\epsilon= \frac12\epsilon\log(|B|-1) +\epsilon\log|B| + (1+\frac\epsilon{2})\,h\Big(\frac{\epsilon}{2+\epsilon}\Big) + h\left(\frac{\epsilon}2\right)\,.$$ 
Therefore,
\begin{align*}
\|	\cM-\Theta\circ \cN\|_\diamond\le \eps\;\Rightarrow\;\cn \efqln{\hat\epsilon} \cm \;\Rightarrow\; \cn \esln{\hat\epsilon} \cm\,.
	\end{align*}
\end{lemma}
\begin{proof}
We assume that $\|\cM-\Theta\circ \cN\|_\diamond\le \eps$ and start with an input state $\rho_{AA'}$. Define $\sigma_{AB}= (\Theta\circ\cn)\otimes \id(\rho_{AA'})$ and $\tau_{AB}= \cm\otimes \id(\rho_{AA'})$. First, observe the following:
	\begin{align*}
I(A:B)_{\cM(\rho)}-I(A:B')_{\cN(\rho)}&\le 	I(A: B)_\tau-I(A : B)_\sigma \\
&=  -H(B)_\sigma+H(AB)_\sigma - H(AB)_\tau+H(B)_\tau  \\
	&= -H(B)_\sigma +H(B)_\tau   - H(B|A)_\tau + H(B|A)_\sigma \\
	&\leq \frac12\epsilon\log(|B|-1) +\epsilon\log|B| + (1+\frac\epsilon{2})\,h\Big(\frac{\epsilon}{2+\epsilon}\Big) + h\left(\frac{\epsilon}2\right)\\
	&=: \hat\epsilon\,,
	\end{align*}
	where the first inequality follows by data-processing to remove the degrading map, the first equality is by definition of the mutual information, the second by adding a zero once and the fact that $\tau_{A}=\sigma_{A}$, the last inequality by applying Lemma~\ref{Hsharp} and Lemma~\ref{AFW}. The second claim follows from \eqref{Eqn:chain}.
\end{proof}
Clearly, the diamond norm is symmetric under exchanging the arguments. This motivates the following definition. 
\begin{definition}
Two channels $\cn$ and $\cm$ are said to be comparably $\epsilon$-completely less noisy if $\cn \ecln{\epsilon} \cm$ and also $\cm \ecln{\epsilon} \cn$, and similarly for comparably $\epsilon$- regularized less noisy and comparably $\epsilon$- less noisy.
\end{definition}
\begin{corollary}\label{Lemma:epsDegCompLN}
Assuming $\cN,\cM:A\to B$, if $\|\cm - \cn\|_\diamond\leq\epsilon$, then $\cm$ and $\cn$ are comparably $\hat\epsilon$-completely less noisy, with $\hat\epsilon=2\epsilon\log|B| + (2+\epsilon)h\big(\frac{\epsilon}{2+\epsilon}\big)$. 
\end{corollary}
\begin{proof}
This follows easily by Lemma~\ref{Lemma:eps-c-LN} and the symmetry of the diamond norm. 
\end{proof}
As a special case, of course also the non-complete version $(n=1)$ is implied. Furthermore, it can be easily seen that $\hat\epsilon$-completely less noisy implies $\hat\epsilon$-completely more capable. Again, the same holds for the non-complete case. In a similar way, we can directly derive a few other Lemmas. 
\begin{lemma}
	If $\cm$ is an $\epsilon$-anti-degraded version of $\cn$, then $\cm$ is $\hat\epsilon$-anti-completely less noisy than $\cn$, with $\hat\epsilon=2\epsilon\log|E| + (1+\frac{\epsilon}{2})h\big(\frac{\epsilon}{2+\epsilon}\big)$. 
\end{lemma}
\begin{proof}
	The proof follows by simply exchanging the roles of $\cm$ and $\cn$ in the proof of Lemma~\ref{Lemma:eps-c-LN}. 
\end{proof}
\begin{lemma}
	If $\cm$ is an $\epsilon$-conjugate degraded version of $\cn$, then $\cm$ is $\hat\epsilon$-completely less noisy than $\cn$, with $\hat\epsilon=2\epsilon\log|E| + (1+\frac{\epsilon}{2})\,h\big(\frac{\epsilon}{2+\epsilon}\big)$. 
\end{lemma}
\begin{proof}
	We again use the same technique as in the proof of Lemma~\ref{Lemma:eps-c-LN}. However, we need one additional step. Since we only want to show that a less noisy type ordering is implied we can restrict the input state to be classical-quantum. It needs to be shown that going from $\sigma_{E_1^nU}= \cn^{\otimes n}(\rho_{A_1^nU})$ to $\nu_{E_1^nU}= \cc\circ\cn^{\otimes n}(\rho_{A_1^nU})$ does not change the mutual information, where $\cc$ is the conjugation map. We have that, 
	\begin{align}
	I(E_1^n:U) &= H(E_1^n) - H(E_1^n|U) \\
	&= H(E_1^n) - \sum_u p(u) H(E_1^n|U=u). 
	\end{align}
	Now, since $(\cc^{\otimes n}\circ\cN^{\otimes n})(\rho))^T= \cN^{\otimes n}(\rho)$ it follows that the mutual information remains unchanged. 
\end{proof}

\subsection{Approximating capacities of a quantum channel}
Next, we consider a quantum channel $\cn$ with associated complementary channel $\cn^c$\footnote{For detailed definitions, we refer to \cite{wilde2013quantum}.}. We want to investigate properties of the private capacity $P$ and the quantum capacity $Q$. 
Let us start with the following.
\begin{theorem}\label{theoboundscapacities}
Let $\cn$ be a quantum channel.
\begin{itemize}
	\item[(i)] If $\cn$ is $\epsilon$-more capable, then $Q^{(1)}(\cn) \leq P^{(1)}(\cn) \leq Q^{(1)}(\cn) + \epsilon$.
	\item [(ii)]	If $\cn$ is $\epsilon$-regularized more capable, then $	Q(\cn) \leq P(\cn) \leq Q(\cn) + \epsilon$.
	\item [(iii)]	If $\cn$ is $\epsilon$-fully quantum less noisy, then 
	$Q^{(1)}(\cn) \leq Q(\cn) \leq Q^{(1)}(\cn) + \epsilon.$
	\item [(iv)]	If $\cn$ is $\epsilon$-fully quantum less noisy and $\epsilon$-regularized more capable, then 
	$ P^{(1)}(\cn) \leq P(\cn) \leq P^{(1)}(\cn) + 2\epsilon.$
\end{itemize}
	
\end{theorem} 

\begin{proof}
The proof of the first part is similar to that of \cite[Proposition 1]{watanabe2012private}. The second part is simply the same proof for arbitrary $n$ and taking the limit, see also \cite[Theorem 1]{watanabe2012private}. 	The third part is similar to the proof of \cite[Proposition J.1]{cross2017uniform}. Finally, the last part follows directly from the previous bounds.
\end{proof}
To compare the above results to \cite[Theorem 3.4]{sutter2017approximate}, we can easily make the following corollary:
\begin{corollary}
	Let $\cn$ be either $\epsilon$-degradable or $\epsilon$-conjugate degradable. Then we have
	\begin{align}
&	Q^{(1)}(\cn) \leq P^{(1)}(\cn) \leq Q^{(1)}(\cn) + \hat\epsilon \\
	&Q(\cn) \leq P(\cn) \leq Q(\cn) + \tilde\epsilon \\
	&Q^{(1)}(\cn) \leq Q(\cn) \leq Q^{(1)}(\cn) + \tilde\epsilon \\
	&P^{(1)}(\cn) \leq P(\cn) \leq P^{(1)}(\cn) +2 \tilde\epsilon 
	\end{align}
	with 
	\begin{align}
	\hat\epsilon &= \frac12\epsilon\log(|B|-1) +\epsilon\log|B| + \left(1+\frac\epsilon{2}\right)\,h\Big(\frac{\epsilon}{2+\epsilon}\Big) + h\left(\frac{\epsilon}2\right)\,, \\
	\tilde\epsilon &=2\epsilon\log|B| + (2+\epsilon)\,h\Big(\frac{\epsilon}{2+\epsilon}\Big)\,.
	\end{align}
\end{corollary}
\begin{proof}
	Follows directly from \Cref{theoboundscapacities} and \Cref{Lemma:epsFQLN1}.
\end{proof}
We can furthermore bound the quantum and private capacity of almost anti-less noisy or anti-more capable channels.
\begin{theorem}
Let $\cn$ be a quantum channel.
\begin{itemize}
	\item[(i)] $\cn$ be $\epsilon$-anti more capable. Then $Q^{(1)}(\cn) \leq \epsilon$.
	\item[(ii)] 	If $\cn$ is $\epsilon$-regularized anti more capable, then $Q(\cn) \leq \epsilon$.
	\item[(iii)] 	If $\cn$ be $\epsilon$-anti less noisy, then $	P^{(1)}(\cn) \leq \epsilon$.
	\item[(iv)] 	If $\cn$ is $\epsilon$-regularized anti less noisy, then $	P(\cn) \leq \epsilon$.
\end{itemize}
\end{theorem} 
\begin{proof}
	The first two bounds follow immediately by noting that the coherent information can be written as a difference of two mutual informations (see e.g.~\cite[Theorem 13.6.1]{wilde2013quantum}):
\begin{align}
I(A\rangle B) = I(X:B) - I(X:E).
\end{align}
	The proof of the next two bounds follows immediately from the definitions of the private information and the less noisy partial order.
\end{proof}
Again, this immediately implies the previously known results about degradable channels (and adding analog ones for conjugate degradable).
\begin{corollary}
	Let $\cn$ be either $\epsilon$-anti degradable or $\epsilon$-conjugate anti degradable. Then we have
	\begin{align}
&	Q^{(1)}(\cn) \leq P^{(1)}(\cn) \leq  \hat\epsilon \\
	&Q(\cn) \leq P(\cn) \leq \tilde\epsilon,
	\end{align}
	with 
	\begin{align}
	\hat\epsilon &= \frac12\epsilon\log(|B|-1) +\epsilon\log|B| + \left(1+\frac\epsilon{2}\right)\,h\Big(\frac{\epsilon}{2+\epsilon}\Big) + h\left(\frac{\epsilon}2\right)\,, \\
	\tilde\epsilon &=2\epsilon\log|B| + (2+\epsilon)\,h\Big(\frac{\epsilon}{2+\epsilon}\Big)\,.
	\end{align}
\end{corollary}
This should be compared to \cite[Theorem 3.8]{sutter2017approximate}. 

In a similar fashion we can also reproduce the results in~\cite{leditzky2018approaches}. 
\begin{lemma}
If $\cn$ is $\epsilon$-less noisy than $\cm$, we have
\begin{align}
\chi(\cm) \leq \chi(\cn) + \epsilon. 
\end{align}
If $\cn$ is $\epsilon$-regularized less noisy than $\cm$, we have
\begin{align}
C(\cm) \leq C(\cn) + \epsilon. 
\end{align}
If $\cn$ and $\cm$ are comparably $\epsilon$-regularized less noisy and $\cn$ is weakly additive, we have
\begin{align}
C(\cm) \leq \chi(\cm) + 2\epsilon. 
\end{align}
If $\cn$ is $\epsilon$-fully quantum more capable than $\cm$, we have
\begin{align}
C_E(\cm) \leq C_E(\cn) + \epsilon. 
\end{align}
\end{lemma}
\begin{proof}
The first two statements essentially follow by definition of the involved quantities. The third statement follows by chaining them together and applying the additivity:
\begin{align}
C(\cm) \leq C(\cn) + \epsilon = \chi(\cN) + \epsilon \leq \chi(\cm) + 2\epsilon. \label{Eq:proveWA}
\end{align}
The final statement again follows trivially from the definitions.
\end{proof}
Combining the previous results we easily get the following corollary. 
\begin{corollary}
If $\|\cm - \cn\|_\diamond\leq\epsilon$ and $\cn$ is weakly additive, we have
\begin{align}
C(\cm) \leq \chi(\cm) + \hat\epsilon+\tilde\epsilon ,
\end{align}
	\begin{align}
	\hat\epsilon &= \frac12\epsilon\log(|B|-1) +\epsilon\log|B| + \left(1+\frac\epsilon{2}\right)\,h\Big(\frac{\epsilon}{2+\epsilon}\Big) + h\left(\frac{\epsilon}2\right)\,, \\
	\tilde\epsilon &=2\epsilon\log|B| + (2+\epsilon)\,h\Big(\frac{\epsilon}{2+\epsilon}\Big)\,.
	\end{align}
\end{corollary}
\begin{proof}
This follows along Equation~\eqref{Eq:proveWA} and noting that the final inequality only requires the unregularized $\epsilon$-less noisy property. 
\end{proof}
This should be compared to \cite[Corollary II.7]{leditzky2018approaches}. Moreover, we note that given channels $\cn,\cm$ it is possible to efficiently compute the minimal $\epsilon$ such that $\cN$ is an $\epsilon$ degraded version of $\cM$~\cite{leditzky2018approaches,sutter2017approximate}.

\section{Characterizations of strong data processing}
As we discussed in Sec.~\ref{lessnoisySDPI}, it is possible to show that a channel $\cN$ is less noisy than another channel $\cM$ by considering their respective contraction coefficients. Thus, contraction coefficients can be used to determine if two channels are related in the order. In this section we will now gather characterizations and new properties of contraction coefficients that are of interest beyond the study of partial orders. Indeed, SDPIs are now also becoming a standard tool to understand the computational power of noisy quantum devices~\cite{Franca2021,Wang2021,Wang2021a,DePalma2022} and for that application tensorization is also key.
\subsection{SDPI via hypercontractivity}
Let us prove some characterizations of strong data processing inequalities (SDPI), i.e. the relative entropy contraction coefficient. 
Obtaining characterizations of strong SDPIs has been the focus of recent activity in quantum information theory. For instance, in~\cite{berta_brascamp} the authors show that for a channel $\cM:\cB(\cH_A)\to\cB(\cH_B)$:
\begin{align*} 
D(\cM(\rho) \| \cM(\sigma)) & \leq \eta D(\rho \| \sigma) \quad \forall \rho \in \mathcal{D}(\cH_A) 
\end{align*}
and 
\begin{align*}
\operatorname{tr} \exp \left(\log \sigma+\cM^{\dagger}(\log \omega)\right) & \leq\left\|\exp \left(\log \omega+\frac{1}{\eta} \log \cM(\sigma)\right)\right\|_{\eta} \quad \forall \omega \in \mathcal{D}(\cH_B) 
\end{align*}
are equivalent. Here we will focus on a hypercontractive approach. In short, hypercontractive inequalities are bounds on various $p\to q$ norms of quantum channels $\cN$. As we will see, these bounds are closely related to entropic inequalities for the channel. But by reducing an entropic inequality to a norm inequality it is often easier to obtain tensorization.

In the classical setting, results connecting hypercontractivity of a Markov kernel and SDPIs go back to Alswede and G\'{a}cs~\cite{ahlswede1976spreading} and have also been the focus of the more recent paper~\cite[Theorem 6]{Raginsky_2013}. In order to state our results, let us define the $\sigma-$weighted $p$-quasi-norms for $\sigma\in\cD_+(\cH)$ and $X\in\cB(\cH)$ as:
\begin{align*}
\|X\|_{p,\sigma}^p=\textrm{tr}\left (|\sigma^{\frac{1}{2p}}X\sigma^{\frac{1}{2p}}|^p\right).
\end{align*}
For $p\geq1$ these quantities are norms, while for $p<1$ only quasi-norms.
They are closely related to the sandwiched R\'enyi divergences $D_p$ of~\cite{Muller-Lennert2013,Wilde2014}, which can be defined as 
\begin{align*}
D_p(\rho\|\sigma)=\frac{p}{p-1}\log(\|\sigma^{-\frac{1}{2}}\rho\sigma^{-\frac{1}{2}}\|_{p,\sigma})
\end{align*}
for $p\in(0,1)\cup(1,\infty)$ and are known to satisfy a data processing inequality for $p\in(\tfrac{1}{2},1)\cup(1,\infty)$~\cite{Beigi2013}. 
Moreover, we have that as $p\to1$, they converge to the usual relative entropy and that these entropies are monotone increasing in the parameter $p$~\cite{Muller-Lennert2013}.
The connection between these norms and R\'{e}nyi divergences will allow 
us to obtain several conditions that imply bounds on the relative entropy contraction coefficients and the 
less noisy ordering based on the Petz recovery map of the channel. For a quantum channel $\cM$ and a reference state $\sigma$, define $\Gamma_{\sigma}^p(X)=\sigma^{\frac{p}{2}}X\sigma^{\frac{p}{2}}$ with the convention  $\Gamma_{\sigma}=\Gamma_{\sigma}^1$ and the Petz recovery map as $$\hat{\cM}=\Gamma_{\sigma}\circ \cM^*\circ\Gamma_{\cM(\sigma)}^{-1}\,.$$ We omit the reference state $\sigma$ in the definition, as it should always be clear from the context.
One can readily check that restricting to the support of $\sigma$ it is a quantum channel that satisfies $\hat{\cM}\circ \cM(\sigma)=\sigma$ and it is known that if $D(\cM(\rho) \| \cM(\sigma)) =  D(\rho \| \sigma) $, then $\hat{\cM}\circ \cM(\rho)=\rho$~\cite{Petz1986}. That is, the relative entropy does not contract under $\cM$ if and only if one can reverse the action of the channel on two input states simultaneously with the Petz recovery map. 

Our techniques will mostly be based on studying how the $p,\sigma$ norms contract for different values of $p$ and reference states $\sigma$ under the Petz recovery map. Such statements are intimately related to strong data processing inequalities. Indeed, as shown in~\cite{muller2017monotonicity}, the data-processing inequality for the R\' enyi divergences is equivalent to the statement that $\|\hat{\cM}^*\|_{(p,\sigma)\to (p,\cM(\sigma))}\leq 1$ for all choices of input state $\sigma$ and channels $\cM$. The basic idea of this section is that if the inequality 
\begin{align}\label{equ:hypercontractive}
\|\hat{\cM}^*\|_{(p,\sigma)\to (q,\cM(\sigma))}\leq1
\end{align}
holds for $q>p$, then we also obtain a strong data processing inequality for $\cM$ and some R\' enyi entropy. 
Note that this statement is stronger than the one with $p=q$, as the norms are monotone increasing in $p$. Denote by $p',q'$ the H\"older conjugates of $p,q$.
One can show by duality of norms that Eq.~\eqref{equ:hypercontractive} is in fact equivalent to
\begin{align}\label{equ:petz}
\|\Gamma_{\sigma}^{-\frac{p'-1}{p'}}\circ\hat{\cM}(\rho)\|_{p'}\leq\|\Gamma_{\cM(\sigma)}^{-\frac{q'-1}{q'}}(\rho)\|_{q'}.
\end{align}
for all states $\rho$. As we will see in Proposition~\eqref{prophyper}, the expression above is equivalent to an SDPI for R\' enyi entropies.

We start with the following standard Lemma:
\begin{lemma}\label{lem:derivative}
	For any two quantum states $\rho,\sigma$ s.t. $\operatorname{supp}\rho\subset\operatorname{supp}\sigma$ we have:
	\begin{align*}
	\frac{d}{dp}\|\Gamma_{\sigma}^{-1}(\rho)\|_{p,\sigma}\big |_{p=1}=D(\rho||\sigma).
	\end{align*}
\end{lemma}
\begin{proof}
	Note that
	\begin{align*}
	\frac{p}{p-1}\log(\|\Gamma_{\sigma}^{-1}(\rho)\|_{p,\sigma})=D_p(\rho||\sigma),
	\end{align*}
	where $D_p$ is the sandwiched Renyi divergence. Define the functions
	\begin{align*}
	\phi(p)=p\log(\|\Gamma_{\sigma}^{-1}(\rho)\|_{p,\sigma}),\quad\varphi(p)=\|\Gamma_{\sigma}^{-1}(\rho)\|_{p,\sigma}.
	\end{align*}
	As $\lim_{p\to1}D_p(\rho||\sigma)=D(\rho||\sigma)$ and $\phi(1)=0$ we have:
	\begin{align*}
	\frac{d}{dp}\phi(p)|_{p=1}=D(\rho||\sigma).
	\end{align*}
	As $\phi(p)=p\log(\varphi(p))$, it follows that
	\begin{align*}
	\frac{d}{dp}\phi(p)=\log(\varphi(p))+p\frac{\varphi'(p)}{\varphi(p)},
	\end{align*}
	which yields the claim after noting that $\varphi(1)=1$.
\end{proof}
The following characterization of a strong DPI then follows:
\begin{proposition}\label{prop4}
	For some $\sigma>0$ and a quantum channel $\cM$, let $\hat{\cM}$ be the Petz recovery map with respect to $\sigma$. 
	Then the following are equivalent:
	\begin{enumerate}
	\item We have for all $\rho \in \cD(A)$
		\begin{align}\label{equ:strondpi12}
	D(\cM(\rho) \| \cM(\sigma)) & \leq \eta D(\rho \| \sigma)
	\end{align}
	\item  there exists an $\epsilon>0$ s.t. for all $0\leq\tau\leq\epsilon$ we have
	\begin{align}\label{equ:hyperdpi}
	\|\hat{\cM}^*\|_{(1+\eta\tau,\sigma)\to(1+\tau,\cM(\sigma))}\leq 1
	\end{align}
	\item there exists an $\epsilon>0$ s.t. for all $0\leq\tau\leq\epsilon$ we have
	\begin{align}\label{equ:recovery}
	D_{\frac{\eta\tau+1}{\eta\tau}}(\hat{\cM}(\rho) \| \sigma) \leq \frac{1+\eta\tau}{1+\tau}D_{\frac{\tau+1}{\tau}}(\rho\| \cM(\sigma))
	\end{align}
	\end{enumerate}
\end{proposition}
\begin{proof}
	Let us start by showing that~\eqref{equ:hyperdpi} implies~\eqref{equ:strondpi12}. 
	Note that by setting $X=\Gamma^{-1}(\rho)$, ~\eqref{equ:hyperdpi} gives:
	\begin{align*}
	\|X\|_{1+\eta\tau,\sigma}\geq \|\Gamma^{-1}_{\cM(\sigma)}(\cM(\rho))\|_{1+\tau,\cM(\sigma)}.
	\end{align*}
	Taking the logarithm and rewriting the Equation above in terms of sandwiched Renyi divergences gives
	\begin{align*}
	\frac{(1+\tau)}{\eta(1+\eta\tau)}D_{1+\tau}(\cM(\rho) \| \cM(\sigma)) & \leq D_{1+\eta\tau}(\rho \| \sigma)
	\end{align*}
	and Eq.~\eqref{equ:strondpi1} follows by taking the limit $\tau\to0$.
	Let us show the other direction. 
	It follows from Lemma~\ref{lem:derivative} and a Taylor expansion that
	\begin{align*}
	&\|X\|_{1+\eta\tau,\sigma}=1+\eta\tau D(\rho \| \sigma)+\mathcal{O}(\tau^2),\\
	&\|\Gamma^{-1}_{\cM(\sigma)}(\cM(\rho))\|_{1+\tau,\cM(\sigma)}=1+\tau D(\cM(\rho) \| \cM(\sigma))+\mathcal{O}(\tau^2),
	\end{align*}
	from which it follows that 
	\begin{align*}
	\|X\|_{1+\eta\tau,\sigma}\geq\|\Gamma^{-1}_{\cM(\sigma)}(\cM(\rho))\|_{1+\tau,\cM(\sigma)}
	\end{align*}
	for all $\tau$ small enough. The claim follows from the fact that is suffices to restrict to positive operators when computing $p\to q$ norms of completely positive maps~\cite{Olkiewicz1999}.
	Finally, Eq.~\eqref{equ:recovery} follows from a duality argument. Indeed, it is easy to show that Eq.~\eqref{equ:hyperdpi} is equivalent to Eq.~\eqref{equ:petz}. Taking the logarithm, we obtain the inequality in terms of the Petz recovery map.
\end{proof}
Thus, we see that an SDPI is equivalent to a nontrivial $p\to q$ inequality for the Petz recovery map for $p,q$ for all $p$ slightly larger than $1$. 
We also immediately obtain a contractive characterization of the less noisy order:
\begin{corollary}
Let $\cm,\cn$ be two quantum channels. Then the following are equivalent:
\begin{enumerate}
\item $\cn \succeq_{\operatorname{l.n.}} \cm$ 
\item For all $\sigma\in\cD(\cH)$ there and $X>0$ there exists an $\epsilon$ such that for all $\tau\leq \epsilon$ we have:
\begin{align}\label{equ:normlessnoisy}
\|\hat{\cm}^*(X)\|_{1+\tau,\cM(\sigma)}\leq \|\hat{\cn}^*(X)\|_{1+\tau,\cN(\sigma)}
\end{align}
\end{enumerate} 
\end{corollary}
\begin{proof}
By rewriting the inequality in Eq.~\eqref{equ:normlessnoisy} in terms of the R\'enyi divergences and taking the limit $\tau\to0$ we see that they imply that
\begin{align*}
D(\cn(\rho)\|\cn(\sigma))\geq D(\cm(\rho)\|\cm(\sigma))
\end{align*}
for all $\rho,\sigma$. By Proposition~\ref{propDmutualinfo}, this is equivalent to $ \cn \succeq_{\operatorname{l.n.}} \cm$. 
To show the other direction we may again resort to a Taylor expansion. 
\end{proof}

It is straightforward to adapt the results above to the completely less noisy ordering by suitably adapting the involved channels and states.

However, it might be difficult to compute the $1+\tau$ norms analytically, as it is more convenient to work with integer values of $p$.
Thus, we will show that any nontrivial $p\to q$ inequality gives rise to an SDPI, possibly with an error term. 

\begin{proposition}\label{prophyper}

	Suppose that for some $1<p\leq q<\infty$, $C\geq 1$ and two channels $\cM,\cN$, where we assume that $\cN$ is invertible, we have
	\begin{align}\label{equ:hypercontractivitychannel}
	\|\hat{\cM}^*\circ (\hat{\cN}^*)^{-1}\|_{(p,\cN(\sigma))\to(q,\cM(\sigma))}\leq C.
	\end{align}
	and 
	\begin{align}\label{equ:contraction-trace}
	\|\hat{\cM}^*\circ (\hat{\cN}^*)^{-1}\|_{(1,\cN(\sigma))\to(1,\cM(\sigma))}\le 1\,.
	\end{align}
	Then
	\begin{align}\label{equ:strondpi1}
	D(\cM(\rho) \| \cM(\sigma)) & \leq \frac{(p-1)q}{p(q-1)}D(\cN(\rho) \| \cN(\sigma))+\frac{q}{q-1}\log(C) \quad \forall \rho \in \cD(\cH).
	\end{align}
\end{proposition}
\begin{proof}
Combining Eq.~\eqref{equ:hypercontractivitychannel}, Eq.~\eqref{equ:contraction-trace} and the Stein-Weiss interpolation theorem (see e.g.~\cite{muller2017monotonicity} for a quantum information friendly exposition) yields that:
\begin{align}\label{equ:interpolationresult}
\|\hat{\cM}^*\circ (\hat{\cN}^*)^{-1}\|_{(p_\theta,\cN(\sigma))\to(q_\theta,\cM(\sigma))}\leq C^{1-\theta}
\end{align}
with
\begin{align*}
\frac{1}{p_\theta}=\theta+\frac{1-\theta}{p},\quad\frac{1}{q_\theta}=\theta+\frac{1-\theta}{q}.
\end{align*}
for $0\leq \theta\leq1$, which can be rewritten as:
\begin{align}\label{equ:rewrittenentropy}
\|\hat{\cM}^*(X)\|_{(q_\theta,\cM(\sigma))}\leq C^{1-\theta}\|\hat{\cN}^*(X)\|_{(q_\theta,\cN(\sigma))}
\end{align}
for all $X$.
Solving for $p_\theta$ and $q_\theta$ we see that
\begin{align*}
p_\theta=\frac{p}{\theta(p-1)+1},\quad q_\theta=\frac{q}{\theta(q-1)+1}.
\end{align*}
Taking the logarithm of Eq.~\eqref{equ:rewrittenentropy} and rewriting in terms of R\' enyi entropies we see that
\begin{align}\label{equ:beforetakinglimit}
D_{q_\theta}(\cM(\rho) \| \cM(\sigma)) & \leq \frac{(p_\theta-1)q_\theta}{p_\theta(q_\theta-1)}D_{p_\theta}(\cN(\rho) \| \cN(\sigma))+\frac{q_\theta(1-\theta)}{q_\theta-1}\log(C).
\end{align}
A close inspection of the expressions involved in the formula above shows that:
\begin{align*}
\frac{(p_\theta-1)q_\theta}{p_\theta(q_\theta-1)}=\frac{(p-1)q}{p(q-1)},\quad \frac{q_\theta(1-\theta)}{q_\theta-1}=\frac{q}{q-1}.
\end{align*}
The claim then follows by taking the limit $\theta\to 1$ in Eq.~\eqref{equ:beforetakinglimit} and noting that $p_\theta\to1$ and $q_\theta\to1$.
\end{proof}

For any quantum channel $\cM$ we have by the Russ-Dye theorem that
\begin{align*}
\|\hat{\cM}^*\|_{(1,\sigma)\to(1,\cM(\sigma))}\leq 1
\end{align*}
and, thus, Eq.\eqref{equ:contraction-trace} is always satisfied for $\cN=\operatorname{id}$. That is, if we wish to prove an SDPI for a quantum channel with and additive error term, a $p\to q$ inequality suffices.

The case of $2\to2$ norms of the Proposition above is particularly interesting, as it leads to entropic inequalities that tensorize:

\begin{corollary}\label{prop:strongdata}
	Suppose that $\cN$ is invertible as a linear map. Then for some $\sigma>0$, 
\begin{align}\label{equ:contractionrenyi2}
\|\Gamma_{\cM(\sigma)}^{-\frac{1}{2}}\circ \cM\circ \cN^{-1}\circ\Gamma_{\cN(\sigma)}^{\frac{1}{2}}\|_{2\to 2}\leq C.
\end{align}
implies that for all states $\rho$
	\begin{align}
D_2(\cM^{\otimes n}(\rho)\|\cM(\sigma)^{\otimes n})\le \,D_2(\cN^{\otimes n}(\rho)\|\cN(\sigma)^{\otimes n})+2n\log (C)\,.
\end{align}
\end{corollary}
\begin{proof}
Define $\hat{\cM}=\Gamma^{-1}_{\cM(\sigma)}\circ \cM\circ \Gamma_\sigma$ and $\hat{\cN}=\Gamma_{\cN(\sigma)}^{-1}\circ \cN\circ \Gamma_\sigma$. 
Note that Eq.~\eqref{equ:contractionrenyi2} is equivalent to
\begin{align*}
\|\Gamma_{\cM(\sigma)}^{\frac{1}{2}}\circ \hat{\cM}\circ \hat{\cN}^{-1}\circ\Gamma_{\cN(\sigma)}^{-\frac{1}{2}}\|_{2\to 2}\leq C.
\end{align*}
We can further massage this to
\begin{align}
\|\Gamma_{\cM(\sigma)}^{\frac{1}{2}}\circ \hat{\cM}(X)\|_{2}\leq C\|\Gamma_{\cN(\sigma)}^{\frac{1}{2}}\circ \hat{\cN}(X)\|_{2}.
\end{align}
for all operators $X$.
To see this, set $X=\hat{\cN}^{-1}\circ\Gamma_{\cN(\sigma)}^{-\frac{1}{2}}(Y)$.
Moreover, we have 
\begin{align*}
\|\Gamma_{\cN(\sigma)}^{\frac{1}{2}}\circ \hat{\cN}(X)\|_{2}=\|\hat{\cN}(X)\|_{2,\cN(\sigma)},\quad \|\Gamma_{\cM(\sigma)}^{\frac{1}{2}}\circ \hat{\cM}(X)\|_{2}=\|\hat{\cM}(X)\|_{2,\cM(\sigma)}.
\end{align*}
Thus, we conclude that for all $X$ we have:
\begin{align*}
\|\hat{\cM}(X)\|_{2,\cM(\sigma)}\leq C\|\hat{\cN}(X)\|_{2,\cN(\sigma)}.
\end{align*}
Picking $X=\sigma^{-\frac{1}{2}}\rho\sigma^{-\frac{1}{2}}$ and taking the logarithm yields the claim for $n=1$. The claim for other $n>1$ follows from noting that the condition in Eq.~\eqref{equ:contractionrenyi2}  tensorizes. Indeed, note that 
for $\cN^{\otimes n}$ and $\cM^{\otimes n}$ we have:
\begin{align*}
&\|\Gamma_{\cM(\sigma)^{\otimes n}}^{-\frac{1}{2}}\circ \cM^{\otimes n}\circ \left(\cN^{-1}\right)^{\otimes n}\circ\Gamma_{\cN(\sigma)^{\otimes n}}^{\frac{1}{2}}\|_{2\to 2}=\|\left(\Gamma_{\cM(\sigma)}^{-\frac{1}{2}}\circ \cM\circ \cN^{-1}\circ\Gamma_{\cN(\sigma)}^{\frac{1}{2}}\right)^{\otimes n}\|_{2\to 2}\\&=\|\left(\Gamma_{\cM(\sigma)}^{-\frac{1}{2}}\circ \cM\circ \cN^{-1}\circ\Gamma_{\cN(\sigma)}^{\frac{1}{2}}\right)\|_{2\to 2}^n.
\end{align*}
The last equality follows from the fact that the $2\to2$ norm just corresponds to the operator norm of the map, which is multiplicative.
\end{proof}
Note that the condition in Eq.~\eqref{equ:contractionrenyi2} is equivalent to the operator norm of $\Gamma_{\cM(\sigma)}^{-\frac{1}{2}}\circ \cM\circ \cN^{-1}\circ\Gamma_{\cN(\sigma)}^{\frac{1}{2}}$ being bounded by $C$. In particular, this means that it can be verified efficiently given the channels and a target state $\sigma$. Thus, we have identified a condition that implies a strong data processing inequality with the feature that it can both be verified efficiently and tensorizes.

This statement can also be used to obtain contraction coefficients that tensorize:
\begin{corollary}
For some $p>0$ let $\cN_p$ be the channel $\cN_p(\rho)=p\rho+(1-p)\sigma$ and suppose that
\begin{align}\label{equ:lessnoisydepolarizing}
\|\Gamma_{\cM(\sigma)}^{-\frac{1}{2}}\circ \cM\circ \cN_p^{-1}\circ\Gamma_{\sigma}^{\frac{1}{2}}\|_{2\to 2}\leq 1
\end{align}
Then for all states $\rho$
\begin{align*}
D_2(\cM^{\otimes n}(\rho)\|\cM^{\otimes n}(\sigma^{\otimes n}))\le \alpha(p,\sigma)\,D_2(\rho\|\sigma^{\otimes n})
\end{align*}
with $\alpha(p,\sigma)=\operatorname{exp}\left(2(1-\|\sigma^{-1}\|^{-1})\frac{\log(p)}{\log(\|\sigma^{-1}\|)}\right)$.
\end{corollary}
\begin{proof}
By Prop.~\ref{prop:strongdata} we have that Eq.~\eqref{equ:lessnoisydepolarizing} implies that
\begin{align*}
D_2(\cM^{\otimes n}(\rho)\|\cM^{\otimes n}(\sigma^{\otimes n}))\le \,D_2(\cN_p^{\otimes n}(\rho)\|\cN_p^{\otimes n}(\sigma^{\otimes n})).
\end{align*}
In~\cite{[MF16]} the authors show that 
\begin{align*}
\,D_2(\cN_p^{\otimes n}(\rho)\|\cN_p^{\otimes n}(\sigma^{\otimes n}))\leq \alpha(p,\sigma)\,D_2(\rho\|\sigma^{\otimes n}),
\end{align*}
which completes the proof.
\end{proof}
The statement above can be used to bound capacities and entropic quantities of a quantum channel efficiently. For instance, if the channel $\cM$ on a $d$-dimensional space is doubly stochastic, we immediately see that if $\|\cM\circ \cN_p^{-1}\|_{2\to 2}\leq 1$, then 
\begin{align*}
D_2\left(\cM^{\otimes n}(\rho)\middle\|\frac{I}{d^n}\right)\le\text{exp}\left(2(1-d^{-1})\frac{\log(p)}{\log(d)}\right) \,D_2\left(\rho\middle\|\frac{I}{d^n}\right).
\end{align*}
By noting that $D_2(\cM^{\otimes n}(\rho)\|\frac{I}{d^n})\geq n-S(\cM^{\otimes n}(\rho))$, we immediately obtain that the minimum output entropy of $\cM^{\otimes n}(\rho)$ is  lower bounded by 
\begin{align*}
S(\cM^{\otimes n}(\rho))\geq n\left(1-\text{exp}\left(2(1-d^{-1})\frac{\log(p)}{\log(d)}\right)\right).
\end{align*}
Inspecting the expression $\|\cM\circ \cN_p^{-1}\|_{2\to 2}$ more closely, we see that:
\begin{align*}
\|\cM\circ \cN_p^{-1}\|_{2\to 2}=\frac{1}{p}\|\cM-(1-p)\cN_0\|_{2\to2}.
\end{align*}
Thus, we see that $\|\cM\circ \cN_p^{-1}\|_{2\to 2}\leq1$ if and only if the singular values of $\cM-(1-p)\cN_0$ are contained in a ball of radius $p$ around the origin in the complex plane. The largest value of $p$ for which this holds is called the spectral gap of $\cM$. 
\black

\section{Generalized contraction coefficients and partial orders}\label{fdivergence}

In this section we will discuss extensions of the concepts discussed to other divergences, in particular $f$-divergences and $\chi^2$-divergences. Note that both partial orders and contraction coefficients can be naturally defined in terms of any divergence by simply replacing the relative entropy by the desired divergence. 

Interestingly, classically contraction coefficients (and partial orders respectively) based on many different divergences have been shown actually to be the same, see e.g.~\cite{makur2018comparison}. In the quantum case this remains a mostly open problem with some of what is known being summarized in the following section.

\subsection{f-divergences}

Here, we extend the discussions of the previous sections to the setting of arbitrary $f$-divergences~\cite[Chapter 7]{OhyaPetz-Entropy-1993}. The two most commonly used ones are the standard $f$-divergence: given $\rho,\sigma\in\cD(\cH)$ with $\sigma$ full-rank:
\begin{align}
D_f(\rho \| \sigma ) := \tr \sigma^{1/2} f( \Delta_{\rho,\sigma})(\sigma^{1/2})\,,
\end{align}
 where $\Delta_{\rho,\sigma}(X):=\rho\,X\,\sigma^{-1}$ is the so-called relative modular operator between states $\sigma$ and $\rho$ and can be interpreted as a non-commutative generalization of the Radon-Nikodym derivative between two probability mass functions. Similarly, the maximal $f$-divergence is defined as
\begin{align}
\widehat D_f(\rho \| \sigma ) := \tr\sigma f(\sigma^{-1/2}\rho\sigma^{-1/2}).
\end{align}
Note that, often, only those functions that obey $f(1)=0$ are considered valid for f-divergences as it ensures that $D_f(\rho \| \rho) = 0$. However, this excludes e.g. R\'enyi divergences. 

Of course, we can define contraction coefficients $\eta_{D_f}(\cn)$ and $\eta_{{\widehat D}_f}(\cn)$ as for other divergences. 
Next, we want to consider mutual information like quantities based on f-divergences. Consider the following definitions,
\begin{align}
I_f(A:B) &:= D_f(\rho_{AB}\|\rho_A\otimes\rho_B), \\
\widehat I_f(A:B) &:= \widehat D_f(\rho_{AB}\|\rho_A\otimes\rho_B). 
\end{align} 
In what follows, we denote $\widetilde{I}_f,$ resp. $\widetilde{D}_f$ for either $I_f$ or $\widehat{I}_f$, resp. $D_f$ or $\widehat{D}_f$.
Note, that just like for the relative entropy, the following holds. Consider a classical quantum state $\rho_{UB}$ and its marginal $\rho_B$, we have
\begin{align}
\widetilde I_f(U:B) = \sum_u p(u) \,\widetilde D_f( \rho_B^u \| \rho_B ). 
\end{align} 
As a direct consequence, for a quantum channel $\cn$, we have
\begin{align}
\widetilde I_f(U:B) = \sum_u p(u) \widetilde D_f( \rho_B^u \| \rho_B ) \leq \eta_{{\widetilde D}_f}(\cn) \sum_u p(u) \widetilde D_f( \rho_A^u \| \rho_A ) = \eta_{{\widetilde D}_f}(\cn) \widetilde I_f(U:A),  
\end{align}
resulting in 
\begin{align}
\eta_{{\widetilde D}_f}(\cn) \geq \sup_{\rho_{UA}}\frac{\widetilde I_f(U:B)}{\widetilde I_f(U:A)}. 
\end{align}
 Before proving the equivalence of \Cref{propDmutualinfof}, we recall this very simple technical lemma:
\begin{lemma}\label{lemmachateauderive}
	Let $A$ be a self-adjoint operator over a finite dimensional Hilbert space, and let $f$ be a function on $\RR_+$ that is differentiable at $\Id$. Then 
	\begin{align*}
	Df[\Id](A)=f'(1)A\,.
	\end{align*} 
\end{lemma}
\begin{proof}
	Simply write the eigenvalue decomposition of $A:=\sum_i a_i|i\rangle\langle i|$, so that
	\begin{align*}
	Df[\Id](A)&=\lim_{\eps\to 0}\,\frac{f(A+\eps\Id)-f(A)}{\eps}	=\lim_{\eps\to 0}\,\sum_i\frac{f(a_i+\eps)-f(a_i)}{\eps}\,|i\rangle\langle i|=f'(1)\,A\,.
	\end{align*}
\end{proof}
The next Proposition is a direct generalization of Theorem 5.2 in \cite{raginsky2016strong}
\begin{proposition}\label{propDmutualinfof}
Let $f$ be differentiable, and let $\cm\in\CPTP(\cH_A,\cK_B)$ and $\cn\in\CPTP(\cH_A,\cK_{B'})$, $\sigma_{A}\in\cD(\cH_A)$ and $\eta\ge 0$. The following are equivalent: 
	\begin{itemize}
		\item[(i)] For all $\operatorname{c-q}$ states $\rho_{UA}$ with marginal $\sigma_{A}$, where $U$ is an arbitrary classical system, 
		\begin{align*}
		\,\eta \,\widetilde{I}_f(U;B)_{(\id_{U}\otimes \cm)(\rho_{UA})}\ge \widetilde{I}_f(U;B')_{(\id_{U}\otimes \cn)(\rho_{UA})}\,.
		\end{align*}
		\item[(ii)] For any state $\rho_{A}\in\cD(\cH_{A})$ with $\supp(\rho_{A})\subseteq\supp(\sigma_{A})$,
		\begin{align}\label{charactrelativeentf}
		\eta\,\widetilde{D}_f(\cm(\rho_{A})\|\cm(\sigma_{A}))\ge		\widetilde{D}_f(\cn(\rho_{A})\|\cn(\sigma_{A}))\,.
		\end{align}
	\end{itemize}
	Therefore, $\Phi\succeq_{\operatorname{l.n.}}\Psi$ if and only if \Cref{charactrelativeentf} holds for $\eta=1$.
\end{proposition}

\begin{proof}
 We first prove (ii)$\Rightarrow$(i): for any $U\sim p_U$ and conditional states $\{\rho_{A^u}\}\in\cD(\cH_A)$, we have
	\begin{align*}
	&\widetilde{I}_f(U;B)=\sum_{u}P_U(u)\,\widetilde{D}_f(\cm(\rho_A^u)\|\cm(\sigma_A))\\
	&\widetilde{I}_f(U;B')=\sum_{u}P_U(u)\,\widetilde{D}_f(\cn(\rho_A^u)\|\cn(\sigma_A))\,,
	\end{align*}
	where $\sigma_A=\sum_u\,P_U(u)\rho_A^u$  is defined to be the marginal of $\rho_{AU}$. Then the result follows directly from (ii).
	Next, we prove that (i)$\Rightarrow$(ii): without loss of generality we assume that $\sigma_A$ is full-rank. Then, for any $\rho_A\in\cD(\cH_A)$, and $0\le \lambda\le \epsilon$, where $\epsilon$ is small enough to ensure that $\sigma_A-\epsilon\rho_A\geq0$,
	we let $U\equiv U_\lambda$ be the binary random variable of distribution $P_{U_\lambda}(0)=\lambda$, $P_{U_\lambda}(1)=1-\lambda$, and conditional states $\rho_A^0=\rho_A$, $\rho_A^1\equiv \rho_A^1(\lambda)=(1-\lambda)^{-1}(\sigma_A-\lambda\rho_A)$. Clearly we have that $\tr_U(\rho_{UA})=\sigma_A$.
	
	Then, let
	\begin{align*}
	\varphi(\lambda)&:=\eta\,\widetilde{I}_f(U_\lambda;B)-\widetilde{I}_f(U_{\lambda};B')\\
	&=\eta\,\lambda\,\widetilde{D}_f(\cm(\rho_A)\|\cm(\sigma_A))+\eta\,(1-\lambda)\widetilde{D}_f(\cm(\rho_A^1)\|\cm(\sigma_A))\\
	&-\lambda \widetilde{D}_f(\cn(\rho_A)\|\cn(\sigma_A))-(1-\lambda)\widetilde{D}_f(\cn(\rho_A^1)\|\cn(\sigma_A))\,.
	\end{align*}
Since $\rho_A^1=\sigma_A$ when $\lambda=0$, we have that $\varphi(0)=0$. Moreover, (i) implies that $\varphi(\lambda)\ge 0$ for any $\epsilon\geq\lambda\ge 0$. Since $\varphi(0)=0$ from $\rho_A^1(0)=\sigma_A$, we have that  $\varphi'(0)\ge 0$. The result follows after computing the latter derivative:
	\begin{align*}
	\varphi'(0)=&\eta \widetilde{D}_f(\cm(\rho_A)\|\cm(\sigma_A))-\widetilde{D}_f(\cn(\rho_A)\|\cn(\sigma_A)) \\
	&+\,\left.\frac{d}{d\lambda}\right|_{\lambda=0}\eta\widetilde{D}_f(\cm(\rho_A^1)\|\cm(\sigma_A))-\widetilde{D}_f(\cn(\rho_A^1)\|\cn(\sigma_A))\,,
	\end{align*}
	 We will be done as soon as we can show that the derivative above is equal to $0$. By linearity of the channels $\cm$ and $\cn$, it is enough to show that 
	 \begin{align*}
	 \left.\frac{d}{d\lambda}\right|_{\lambda=0}\,\widetilde{D}_f(\rho_A^1\|\sigma_A)=0\,.
	 \end{align*}
	 We first focus the case of the maximal $f$-divergences: we have
	 \begin{align}
	 	 \left.\frac{d}{d\lambda}\right|_{\lambda=0}\,\widetilde{D}_f(\rho_A^1\|\sigma_A)&=\tr\Big[   \rho_A^1 Df[\Id]\Big(\left.\frac{d}{d\lambda}\right|_{\lambda=0}  \sigma_A^{-\frac{1}{2}} \rho_A^1\sigma_A^{-\frac{1}{2}} \Big)\Big]\\
	 	 &=f'(1)\,\left.\frac{d}{d\lambda}\right|_{\lambda=0}\,\tr\big[\rho_A^1\sigma_A^{-\frac{1}{2}} \rho_A^1\sigma_A^{-\frac{1}{2}}\big]\\
	 	 &=2f'(1)\,\left.\frac{d}{d\lambda}\right|_{\lambda=0}\frac{1-2\lambda+\lambda^2\tr\big[\rho_A\sigma_A^{-\frac{1}{2}} \rho_A\sigma_A^{-\frac{1}{2}}\big]}{(1-\lambda)^2}
	 \end{align}
	where the second identity comes from \Cref{lemmachateauderive}. It is then straightforward to evaluate the derivative above and verify that
\begin{align*}
\left.\frac{d}{d\lambda}\right|_{\lambda=0}\frac{1-2\lambda+\lambda^2\tr\big[\rho_A\sigma_A^{-\frac{1}{2}} \rho_A\sigma_A^{-\frac{1}{2}}\big]}{(1-\lambda)^2}=0
\end{align*}
which gives what we needed to prove.
The case of the standard $f$-divergence  can be treated similarly: 
	\begin{align*}
\left.	\frac{d}{d\lambda}\right|_{\lambda=0}\,D_f(\rho_A^1\|\sigma_A)&=\tr\sigma_A^{\frac{1}{2}}\left.\frac{d}{d\lambda}\right|_{\lambda=0}f(L_{\rho_A^1} R_{\sigma_A^{-1}})\sigma_A^{\frac{1}{2}}\\
&\overset{(1)}{=}\tr\sigma_A^{\frac{1}{2}}Df[\id]\Big( \left. \frac{d}{d\lambda}\right|_{\lambda=0}L_{\rho_A^1} R_{\sigma_A^{-1}}\Big)\sigma_A^{\frac{1}{2}} \\
&=f'(1)\,\tr\sigma_A^{\frac{1}{2}}(\mathbb{I}-\rho_A\sigma_A^{-1})\sigma_A^{\frac{1}{2}}=0\,.
	\end{align*}
	where in $(1)$ we used once again \Cref{lemmachateauderive} in (1).

\end{proof}

	\begin{remark}
	The previous result can be easily extended to the case of optimized quantum $f$ divergences as defined in \cite{wilde2018optimized}.
	\end{remark} 
		\begin{remark}
	Note that we only require $f$ to be differentiable. Thus, Proposition~\ref{propDmutualinfof} holds for a larger class of functionals than $f-$divergences.
	\end{remark} 
As a result we can express the contraction coefficient based on $f$-divergences in terms of the corresponding $f$-mutual information
\begin{align}
\eta_{{\widetilde D}_f}(\cn) = \sup_{\rho_{UA}}\frac{\widetilde I_f(U:B)}{\widetilde I_f(U:A)}, 
\end{align}
as well as define partial orders based on $f$-mutual informations that can be equivalently formulated in terms of $f$-divergences. We will denote these by $\succeq_{\operatorname{D_f, l.n.}}$.
	
Let us end this section by discussing our favorite example: the erasure channel. One can easily find the following.
\begin{lemma}
For any function $f$ with $f(1)=0$, we have
\begin{align}
D_f(\ce(\rho) \| \ce(\sigma) ) &= (1-\epsilon) D_f(\rho \| \sigma ), \\
\widehat D_f(\ce(\rho) \| \ce(\sigma) ) &= (1-\epsilon) \widehat D_f(\rho \| \sigma )\,,
\end{align}
which immediately implies that for all functions $f$ and $g$ with the above property, we have
\begin{align}
\eta_{D_f}(\ce) = \eta_{{\widehat D}_g}(\ce) = \eta_{\Re}(\ce) = 1-\epsilon. 
\end{align}
\end{lemma}
\begin{proof}
By direct calculation. 
\end{proof}	
This implies that we can easily extend our previous Proposition~\ref{Prop:erasureContraction}. 
\begin{proposition}\label{Prop:erasureContractionF}
	Let $\cn\in\CPTP(\cH_A,\cK_B)$, $\eta_e\in[0,1]$ and $f$ any function with $f(1)=0$. Then the following are equivalent:
	\begin{itemize}
	\item[(i)] $\cm^{\operatorname{er}}_{\eta_e}\succeq_{\operatorname{D_f, l.n.}}\cn$.
	\item[(ii)] $\eta_{D_f}(\cn)\le \eta_{D_f}(\cm^{\operatorname{re}}_{\eta_e})=(1-\eta_e)$.
	\end{itemize}
\end{proposition}
\begin{proof}
Analog to that of Proposition~\ref{Prop:erasureContraction}. 
\end{proof}
And the same holds for the maximal $f$-divergence.

	 \subsection{Spectral approaches to quantum less-noisy pre-orders}

	 In Theorem 1 of \cite{makur2018comparison}, it was shown in the classical setting that the less noisy pre-order $\cM\succeq_{\operatorname{l.n.}}\Psi$ is equivalent to the following  condition: for any probability mass functions $p,q$
	 \begin{align*}
	 \chi^2(\cM(p),\cM(q))\ge \chi^2(\cn(p),\cn(q))\,,
	 \end{align*}
	 where the $\chi^2$ divergence is defined as follows:
	 \begin{align}\label{chisquareordering}
	 \chi^2(p,q):=\sum_{i}\,\frac{(p(i)-q(i))^2}{q(i)}\,.
	 \end{align}
	 This characterization is powerful as it allows for a spectral analysis of  the less noisy pre-order \cite{makur2018comparison}. The proof of this fact relies on two observations: to prove that less noisiness implies (\ref{chisquareordering}), it uses the characterization of less noisy seen in \Cref{lessnoisySDPI}, namely that $\cM\succeq_{\operatorname{l.n.}}\cn$ is equivalent to 
	 \begin{align}
	 D(\cM(p)\|\cM(q))\ge D(\cn(p)\|\cn(q))\,,
	 \end{align}
	 for any two probability mass functions $p$ and $q$, together with the well-known fact that the $\chi^2$ divergence locally approximates the relative entropy:
	 \begin{align}\label{relattochi2}
	 \lim_{\lambda\to 0^+}\frac{2}{\lambda^2}\,D(\lambda p+(1-\lambda)q\|q)=\chi^2(p,q)\,.
	 \end{align}
	 On the other hand, the relative entropy can be rewritten in terms of the following integral \cite{makur2018comparison} (see also \cite{melbourne2019relationships,nishiyama2019new,nishiyama2020relations} for slightly different related expressions):
	 \begin{align}\label{chi2torelat}
	 D(p\|q):=\int_0^\infty\,\frac{\chi^2\big(p,\frac{tp+q}{1+t}\big)}{1+t}\,dt\,.
	 \end{align}
	 The characterization of less noisiness in terms \Cref{chisquareordering} follows from the joint use of \Cref{relattochi2,chi2torelat}. This result implies as a special case the following result \cite{choi1994equivalence} (see also refinements, e.g. Theorem 3.6 of  \cite{raginsky2016strong}, and \cite{raginsky2016strong,ahlswede1976spreading} for the link to the contraction coefficient for the maximal correlation):
	 \begin{align}\label{chisquarecontractioncoeff}
	 \eta_{\operatorname{Re}}(\cM)=\eta_{\chi^2}(\cM):=\sup_{\substack{p,q\\0<\chi^2(p,q)<\infty}}\,\frac{\chi^2(\cM(p),\cM(q))}{\chi^2(p,q)}\,.
	 \end{align}
	 The $\chi^2$ quantity is known classically to locally approximate any $f$ divergence \cite{campbel1996extended,cencov2000statistical}. This is no longer the case in the non-commutative world, where a complete characterization of $\chi^2$ quantities, also known as Fisher information metrics, was found \cite{LR99,[TKRWV10]}. First of all, we define the following set of functions:
	 \begin{align*}
	 \cK:=\big\{ k;\,-k \text{ is operator monotone}, k(w^{-1})=wk(w), k(1)=1\big\}\,.
	 \end{align*}
	 Next, for $\rho,\sigma\in\cD(\cH)$, and $k\in\cK$, define the quantum $\chi^2$-divergence
	 \begin{align*}
	 \chi_k^2(\rho,\sigma):=\langle \rho-\sigma,\,\Omega_\sigma^k(\rho-\sigma)\rangle\,,
	 \end{align*}
when $\supp(\rho)\subseteq\supp(\sigma)$, and infinity otherwise. In the previous case, the inversion $\Omega_\sigma^k$ is defined as
\begin{align*}
\Omega_\sigma^k:=R_\sigma^{-1}\,k(\Delta_{\sigma.\sigma})\,.
\end{align*}
Here, we also restrict ourselves to a subclass of $f$-divergences: consider the class $\cG$ of continuous operator convex functions $g$ from $\RR_+$ to $\RR$ that satisfy $g(1)=0$. These functions can all be expressed in terms of the following integral representation:
\begin{align}\label{integralrep}
g(w)=a(w-1)+b(w-1)^2+c\,\frac{(w-1)^2}{w}\,+\int_0^\infty\,\frac{(w-1)^2}{w+s}\,d\nu(s)\,,
\end{align}
where $a,b,c>0$ and the integral of the positive measure $d\nu(s)$ on $(0,\infty)$ is bounded. Then, for any $k\in\cK$, there exists $g\in\cG$ such that \cite{LR99,[TKRWV10]}:
\begin{align}\label{eqchidf}
\chi^2_k(\rho,\sigma)=-\left.\frac{\partial^2}{\partial\alpha\partial\beta}D_g(\sigma+\alpha (\rho-\sigma)\|\sigma+\beta(\rho-\sigma))\right|_{\alpha=\beta=0}\,.
\end{align}
The function $k$ is related to $g$ by 
\begin{align}\label{kintermsofg}
k(w):=\frac{g(w)+wg(w^{-1})}{(w-1)^2}\,.
\end{align}
	Hence, like its classical restriction, the quantum $\chi^2$  divergence can be understood as a local approximation of the quantum relative entropy. However, \Cref{chi2torelat} is not known to hold in general in the quantum case. Rather, the following integral representation of the quantum relative entropy is known and has been recently used to get tightenings of the data processing inequality in \cite{carlen2017recovery} as a special case of the integral representation \eqref{integralrep}:
	 \begin{align}\label{relativeentkernel}
	 D(\rho\|\sigma)=\int_0^\infty\tr\big[(t+\Delta_{\sigma,\rho})^{-1}(\rho) \big]-\frac{1}{1+t}\, dt\,,
	 \end{align}
	It can be easily shown that the integrands in \Cref{relativeentkernel} and \Cref{chi2torelat} coincide when $\rho$ and $\sigma$ commute and can be identified with $p$ and $q$. However, the integrand \eqref{relativeentkernel} cannot be identified in general with an adequately normalized $\chi^2$ divergence. In particular, \Cref{chisquarecontractioncoeff} is not known to hold in general in the quantum setting, despite some
	 recent progress in that direction \cite{Lesniewski1998a,hiai2016contraction}. We also refer to the more recent paper \cite{cao2019tensorization} where the strong data processing for the $\chi^2$ divergence was shown to tensorize in the case where the state $\sigma$ is taken to be  a tensor product.
	 
	 For the reasons mentioned above, it seems unlikely that the equivalence between less noisy and the ordering of $\chi^2$ is true in the quantum setting. However, the following result is a direct consequence of the above discussion: 

	 \begin{proposition}\label{theorelat_chisquare}
	 	For any two quantum channels $\cM\in\CPTP(\cH_A,\cK_B),\,\cN\in\CPTP(\cH_A,\cK_{B'})$ and $\eta\ge 0$, $\operatorname{(i)}\Rightarrow\operatorname{(ii)}\Rightarrow\operatorname{(iii)}$, where:
	 	\begin{itemize}
	 		\item[$\operatorname{(i)}$] For any two quantum states $\rho,\sigma\in\cD(\cH_A)$ and any $t\ge  0$:
	 		\begin{align*}
	 	&\tr\cN(\sigma)^{\frac{1}{2}}(\Delta_{\cN(\rho),\cN(\sigma)}-\id)^2(\cN(\sigma)^{\frac{1}{2}})\le \eta	\tr\cM(\sigma)^{\frac{1}{2}}(\Delta_{\cM(\rho),\cM(\sigma)}-\id)^2(\cM(\sigma)^{\frac{1}{2}})\,;~~\operatorname{and}\\
	 	&\tr\cN(\sigma)^{\frac{1}{2}}\frac{(\Delta_{\cN(\rho),\cN(\sigma)}-\id)^2}{\Delta_{\cN(\rho),\cN(\sigma)}+t}(\cN(\sigma)^{\frac{1}{2}})\le\eta \tr\cM(\sigma)^{\frac{1}{2}}\frac{(\Delta_{\cM(\rho),\cM(\sigma)}-\id)^2}{\Delta_{\cM(\rho),\cM(\sigma)}+t}(\cM(\sigma)^{\frac{1}{2}})\,.
	 		\end{align*}
	 		\item[$\operatorname{(ii)}$] For all $g\in \cG$ and any two states $\rho,\sigma\in\cD(\cH_A)$:
	 		 \begin{align*}
	 		 D_g(\cN(\rho)\|\cN(\sigma))\le \eta  D_g(\cM(\rho)\|\cM(\sigma))\,.
	 		 \end{align*}
	 		\item[$\operatorname{(iii)}$] For any $k\in\cK$ and any two quantum states $\rho,\sigma\in\cD(\cH_A)$\,,
	 		\begin{align*}
	 		\chi^2_k(\cN(\rho),\cN(\sigma))\le \eta 	\chi^2_k(\cM(\rho),\cM(\sigma)) \,.
	 		\end{align*}
	 	\end{itemize}
	 	Moreover, if (ii) holds only for a given $g\in\cG$, then (iii) holds for $k\in\cK$ satisfying \eqref{kintermsofg}. Finally, obvious extensions of the above chain of implications hold true in the case of regularized and complete less noisy pre-orders, as well as for the related contraction coefficients.
	 \end{proposition}
 
\begin{proof}
	This is a direct consequence of the expression \eqref{integralrep} as well as the relation \eqref{eqchidf}. 
\end{proof}

Finally, we will once more mention the erasure channel as an illustrating example. 
	 \begin{lemma}
	 	Let $\ce:A\rightarrow B$ be the quantum erasure channel of parameter $\eps$, $\rho_{A}$ and $\sigma_A$ be some quantum states. Then for any $k\in\cK$ we have that
	 	\begin{align}
	 	\chi_k^2(\ce(\rho),\ce(\sigma)) = (1-\epsilon)\,\chi_k^2(\rho,\sigma). 
	 	\end{align}
	 	Therefore, 
	 	\begin{align*}
	 	\eta_{\operatorname{Re}}(\ce) = \eta_{\chi^2}(\ce) = 1-\epsilon. 
	 	\end{align*}
	 \end{lemma}

This means that at least for the erasure channel the contraction coefficients for relative entropy, f-divergences and $\chi_k^2$ divergence are the same.

\section{Functional inequalities}

Another standard way of obtaining strong data processing is through functional inequalities~\cite{Raginsky_2013},  which we will describe in a bit more detail below. This approach is particularly effective whenever we have a semigroup $(\cP_t)_{t\ge 0}$  converging to a unique full rank state $\sigma$. Such semigroups are called primitive. In~\cite{makur2018comparison} the authors relate such functional inequalities to the less noisy ordering in the classical setting, especially when comparing to symmetric channels. Here we will recover their results in the quantum setting, especially Proposition 12, while also revisiting the connections between functional and entropic inequalities.

We will start by defining logarithmic Sobolev inequalities and showing their consequences. Let $\cN$ be a primitive channel whose unique invariant state is $\sigma$, and denote by $\cN^*$ the corresponding completely, positive unital dual map. Corresponding to the map $\cN^*$, we construct a quantum Markov semigroup (QMS) of completely positive unital maps $\cP_t=\e^{-t(\id-\cN^*)}$ for all $t\ge 0$. In other words, the resulting semigroup $(\cP_t)_{t\ge 0}$ has $\mathcal{S}:=\cN^*-\id$ as its generator. The semigroup $(\cP_t)_{t\ge 0}$ is by construction also primitive, and $\sigma$ is its unique invariant state. 
We can then define a scalar product
\begin{align*}
\langle X,\,Y\rangle_\sigma:=\tr\big( \sigma^{\frac{1}{2}}X^\dagger\,\sigma^{\frac{1}{2}}Y\big)
\end{align*}
which induces the $\|\cdot\|_{2,\sigma}$ norm.
Next we define the Dirichlet form $\cE_{\mathcal{S}}$ corresponding to $\mathcal{S}$ as follows: for all $X,Y\in\cB_{sa}(\cH)$,
\begin{align*}
\cE_\mathcal{S}(X,Y)=-\frac{1}{2}\langle X,\,\mathcal{S}+\hat{\mathcal{S}}(Y)\rangle_\sigma\equiv \big\langle\,X,\,\big(\textrm{id}-\frac{\cN^*+\hat{\cN}^*}{2}\big)(Y)\big\rangle_\sigma.
\end{align*}
Next, we define the $2$-entropy as follows: for any $X\in\cB_{sa}(\cH)$,
\begin{align*}
\operatorname{Ent}_{2,\sigma}(X):=D\big(|\sigma^{\frac{1}{4}}X\sigma^{\frac{1}{4}}|^2\big\| \sigma\big)-\tr\big(|\sigma^{\frac{1}{2}} X\sigma^{\frac{1}{2}}|^2\big)\ln\tr\big(|\sigma^{\frac{1}{2}} X\sigma^{\frac{1}{2}}|^2\big)\,,
\end{align*}
where here the relative entropy $D(\cdot\|\cdot)$ has been extended to non-normalized positive operators $D(A\|B)=\tr(A(\ln A-\ln B))$ whenever $\supp(A)\subset \supp(B)$. The 
QMS $(\cP_t)_{t\ge 0}$ satisfies the \textit{logarithmic Sobolev inequality} (LSI) if  there exists $\alpha>0$ such that for every $X\in\cB_{sa}(\cH)$
\begin{align}\label{LSI}
\alpha \,\operatorname{Ent}_{2,\sigma}(X)\le \cE_\mathcal{S}(X,X)\,.
\end{align}
The largest $\alpha$ such that \Cref{LSI} holds 
\begin{align*}
\alpha(\cN^*):=\inf_{\substack{X\in\cB_{sa}(\cH)\\\operatorname{Ent}_2(X)\ne 0}}\,\frac{\cE_{\mathcal{S}}(X,X)}{\operatorname{Ent}_{2,\sigma}(X)}
\end{align*}
is called the \textit{logarithmic Sobolev constant} of  $(\cP_t)_{t\ge 0}$. Likewise, the logarithmic Sobolev inequality for the \textit{discrete-time} Markov chain with constant $\alpha>0$ states that for every $X\in\cB_{sa}(\cH)$,
\begin{align*}
\alpha\,\operatorname{Ent}_{2,\sigma}(X)\le\,\cE_{\id-\cN^*\circ\hat{\cN}^*}(X,X)\,,
\end{align*}
where $\cE_{\id-\cN^*\circ\hat{\cN}^*}(X,Y):=\big\langle X,\,\big(\id-\cN^*\hat{\cN}^*)(Y)\big\rangle_\sigma$ is sometimes referred to as the discrete Dirichlet form. 
Logarithmic Sobolev inequalities are particularly useful for reversible semigroups, i.e. semigroups  such that their Petz recovery map with respect to the invariant state is the semigroup itself. In this case, it is known that a LSI implies the following hypercontractive inequality:
\begin{align}\label{equ:hypercontracitve}
\|\cP_t^*\|_{2\to q(t)}\leq 1,\quad  q(t)=1+e^{2\alpha t}.
\end{align}
By the results of Proposition~\ref{prophyper}, we know that this can be used to obtain strong data processing inequalities involving R\'enyi divergences. Indeed, we immediately obtain:
\begin{align*}
D_2(\cP_t(\rho)\|\sigma)\leq \frac{1+e^{-2\alpha t}}{2}D_{q(t)}(\rho\|\sigma).
\end{align*}
Nevertheless, the inequality above is a bit unsatisfactory, as it does not give complete contraction even as $t\to\infty$.  But this can be easily remedied by  resorting to reversibility. As the semigroup is reversible, Eq.~\eqref{equ:hypercontracitve} also implies
\begin{align}
D_{2}(\cP_t(\rho)\|\sigma)\leq \frac{2e^{-2\alpha t}}{1+e^{-2\alpha t}}D_{1+e^{-\alpha t}}(\rho\|\sigma).
\end{align}
by duality. By Proposition~\ref{prophyper}, we conclude that:
\begin{align*}
D(\cP_t(\rho)\|\sigma)\leq \frac{2e^{-2\alpha t}}{1+e^{-2\alpha t}}D(\rho\|\sigma).
\end{align*}
Exploiting the fact that $\frac{2e^{-2\alpha t}}{1+e^{-2\alpha t}}\leq e^{-\alpha t}$ we then get an exponential decay of the relative entropy from the logarithmic Sobolev inequality. Although this result is by no means new~\cite{DS96,[OZ99],[KT13]}, we believe that the approach presented here showcases transparently how a LSI implies several different entropic inequalities.
Computing the LSI constant is, in general, an arduous task. For the depolarizing semigroup in dimension $d$, it is known that:
\begin{align}\label{equ:LSIdep}
\alpha=\frac{2(1-2d^{-1})}{\log(d-1)}.
\end{align}
It is possible to obtain estimates of LSI through~\ref{equ:LSIdep} for any doubly stochastic channel through so-called comparison techniques. The same is true for generalized depolarizing channels.

Indeed, an easy way of estimating the LSI constant of a semigroup assuming that we know the LSI constant of another QMS is via the comparison of their Dirichlet forms: assume that $\mathcal{S}_1$ and $\mathcal{S}_2$ have the same unique invariant state $\sigma$, then:
\begin{align*}
\exists\lambda>0,\,\forall X\in\cB_{sa}(\cH),\,\cE_{\mathcal{S}_1}(X,X)\le \,\lambda\cE_{\mathcal{S}_2}(X,X)~~~\Rightarrow~~~ \alpha(\mathcal{S}_2)\ge \frac{\alpha(\mathcal{S}_1)}{\lambda}\,. 
\end{align*}

In the next theorem, we show that the information-theoretic notion of less noisy domination is a sufficient condition for the comparison of Dirichlet forms. First, we recall the following Lemma from \cite{makur2018comparison}:
\begin{lemma}\label{polynormal}
	Given a positive semi-definite real matrix $A$ and a normal real matrix $B$,
	\begin{align*}
	A^2=AA^*\ge_{\PSD}BB^T~~~\Rightarrow~~~A=\frac{A+A^T}{2}\ge_{\PSD}\frac{B+B^T}{2}\,.
	\end{align*}
\end{lemma}
From this we get:
\begin{theorem}
	Let $\cN,\cM$ be two primitive quantum channels, with a common unique full-rank invariant state $\sigma$. Then the following are true:
	\begin{itemize}
		\item[(i)] If $\cM\sln\cN$, then:
		\begin{align*}
		\forall X\in\cB(\cH)_+,~~\cE_{\id-\cN^*\hat{\cN}^*}(X,X)\ge \cE_{\id-\cM^*\hat{\cM}^*}(X,X),
		\end{align*}
		or equivalently
		\begin{align*}
		D_2(\cN(\rho)\|\sigma)\leq D_2(\cM(\rho)\|\sigma)
		\end{align*}
		for all states $\rho$.
		\item[(ii)] Assume  that $\cM$ is positive semi-definite on $(\cB_{sa}(\cH),\langle\cdot,\,\cdot\rangle_\sigma)$ and that $\cN^*\hat{\cN}^*=\hat{\cN}^*\cN^*$. If $\cM\sln\cN$, then for all $X\in\cB_{sa}(\cH)$
		\begin{align*}
		\cE_{\mathcal{S}_\cN}(X,X)\ge \cE_{\mathcal{S}_{\cM}}(X,X)
		\end{align*}
		\item[(iii)] Assume that $\cN=\cN_{p,\sigma}$, with $\cN_{p,\sigma}$ the generalized depolarizing channel with depolarizing probability $p$ to a state $\sigma>0$, and $\cN_p\sln \cM$ for some channel $\cM$ with  $\sigma$ as its stationary state. Then:
		\begin{align*}
		\forall X\in\cB(\cH)_+,~~p\cE_{\id-\cN_{1,\sigma}^*}(X,X)\leq \cE_{\id-\cM^*}(X,X),
		\end{align*}
	\end{itemize}
\end{theorem}

\begin{proof}
	(i) We use the implication $(i)\Rightarrow (ii)$ of \Cref{theorelat_chisquare} in order to get that $\cM\sln\cN$ implies that for any two states $\rho_1,\rho_2$
	\begin{align*}
	\chi^2(\cM(\rho_1),\cM(\rho_2))\ge 	\chi^2(\cN(\rho_1),\cN(\rho_2))\,.
	\end{align*}
	This implies, taking $\rho_1:=\rho\equiv \sigma^{\frac{1}{2}}X\sigma^{\frac{1}{2}}$ and $\rho_2=\sigma$ that
	\begin{align*}
	\tr\big(\cM(\rho)\sigma^{-\frac{1}{2}}\cM(\rho)\sigma^{-\frac{1}{2}}\big)\ge 	\tr\big(\cN(\rho)\sigma^{-\frac{1}{2}}\cN(\rho)\sigma^{-\frac{1}{2}}\big)
	\end{align*}
	Note that the adjoint of $\cM^*$ with respect to the weighted scalar product is $\hat{\cM}^*$. Thus, we can rewrite this last condition as follows:
	\begin{align*}
	\langle \cM^*(X),\,\cM^*(X)\rangle_\sigma\ge \langle\cN^*(X),\,\cN^*(X)\rangle_\sigma\Leftrightarrow \cE_{\id-\cM^*\hat{\cM}^*}(X,X)\le \cE_{\id-\cN^*\hat{\cN}^*}(X,X)
	\end{align*}

	(ii) The previously proved result implies in particular that the maps $\cM^*$ and $\cN^*$ seen as matrices acting on the real vector space of  self-adjoint operators provided with the $\sigma$-KMS inner product $\langle\cdot,\,\cdot\rangle_\sigma$, 
	\begin{align*}
	\cN^*\hat{\cN}^*\le_{\PSD} 	\cM^*\hat{\cM}^* 
	\end{align*}
	From \Cref{polynormal}, we find that 
	\begin{align*}
	\frac{\cM^*+\hat{\cM}^*}{2}\ge_{\PSD}\frac{\cN^*+\hat{\cN}^*}{2}~~~\Rightarrow~~~\forall X\in\cB_{sa}(\cH),\,\cE_{\mathcal{S}_\cN}(X,X)\ge \cE_{\mathcal{S}_{\cM}}(X,X)
	\end{align*}
	
	(iii) Note that $\cN_{p,\sigma}^*\hat{\cN}_{p,\sigma}^*=\hat{\cN}_{p,\sigma}^*\cN_{p,\sigma}^*$ for the depolarizing channel. Thus, the assumptions of (ii) are fulfilled. The claim then follows by a simple manipulation.
\end{proof}

Note that $p\cE_{\id-\cN_{1,\sigma}^*}(X,X)$ is nothing but the variance of the observable $X$ with respect to the state $\sigma$. Thus, we conclude that from a less-noisy comparison between channels with the same invariant state we immediately get a comparison of Dirichlet forms and a bound on logarithmic Sobolev constants if one of them is known.
Unfortunately, such techniques do not yield sharp inequalities for the SDPI constant for channels. This is because it is known that for generalized depolarizing channels, we have that the optimal SDPI constant~\cite{muller2016relative} is given by:
\begin{align*}
D(\cN_{p,\sigma}(\rho)\|\sigma)\leq (1-p)^{1+\alpha(\sigma)}D(\rho\|\sigma)
\end{align*} 
where 
\begin{align*}
\alpha(\sigma)=\min _{x \in[0,1]} \left(1+q_{s_{\min }(\sigma)}(x)\right)
\end{align*}
with $q_{s_{\min }(\sigma)}$ the smallest eigenvalue of $\sigma$ and
\begin{align*}
q_{y}(x)=\begin{cases}
\frac{D_{2}(y \| x)}{D_{2}(x \| y)}, & x \neq y \\
1, & x=y
\end{cases},
\end{align*}
where $D_{2}$ is the binary relative entropy.
Moreover, for qubit depolarizing channels we have $\alpha(I/2)=1$ and we know that $\cN^{\otimes n}_{p,\sigma}$ satisfies a SDPI with constant $(1-p)$~\cite{beigi2018quantum}. This is to be contrasted with the LS estimate, which decays with the local dimension.

\section{Special classes of channels}

\subsection{Weyl-covariant channels}\label{depol}

In this section, we provide new comparison bounds for Weyl-covariant channels  \cite{weyl1950theory,davies1970repeated,holevo1993note,1996covariant} (see also \cite{watrous2018theory}). 
We consider the finite group $G:=\mathbb{Z}_n\times \mathbb{Z}_n$ with the following projective representation on $\cH:=\mathbb{C}^{n}$: 
\begin{align*}
(a,b)\in\mathbb{Z}_n\times \mathbb{Z}_n\mapsto W_{a,b}:=U^a\,V^b\,,~~~\text{ where }~~~U:=\sum_{c\in\mathbb{Z}_n}\ket{c+1}\bra{c}~~\text{ and }~~ V:=\sum_{c\in\mathbb{Z}_n}\,\,\e^{\frac{2\pi c i}{n}}\ket{c}\bra{c}\,.
\end{align*}
Next, define the \textit{discrete Fourier transform} $F$ to be the following unitary matrix on $\CC^n$ 
\begin{align*}
F:=\frac{1}{\sqrt{n}}\,\sum_{a,b\in\mathbb{Z}_n}\,\e^{\frac{2\pi i a b}{n}}\,E_{a,b}\,.
\end{align*}
A map $\cm$ on $\cB(\cH)$ is a \textit{Weyl-covariant map} if 
\begin{align*}
\cm\big( W_{a,b}\,X\,W_{a,b}^* \big)=W_{a,b}\,\cm(X)\,W_{a,b}^*
\end{align*}
for every $X\in\cB(\cH)$ and $(a,b)\in\mathbb{Z}_n\times \mathbb{Z}_n$. In particular, any additive noise channel of the form 
\begin{align}\label{Phif}
\cm_f(\rho):=\sum_{(a,b)\in\mathbb{Z}_n\times \mathbb{Z}_n}\,f(a,b)\,W_{a,b}\,\rho\, W_{a,b}^\dagger\,,
\end{align}
where $f:\mathbb{Z}_n\times \mathbb{Z}_n\to\RR_+$ is a probability mass function, is  Weyl-covariant, by the irreducibility of the representation $W$. Moreover, one can show that any Weyl-covariant channel has this form (see Corollary 4.15 in \cite{watrous2018theory}).  The next Lemma is an extension of  Lemma 1 in \cite{makur2018comparison}:
\begin{lemma}\label{lemma1}
	Any two Weyl-covariant noise channels $\cm_f$ and $\cm_h$ satisfy $\cm_h\succeq_{\operatorname{deg}}\cm_f$ if and only if there exists a probability distribution function $k$ on $\mathbb{Z}_n\times \mathbb{Z}_n$  such that $\cm_f=\cm_h\circ\cm_k=\cm_k\circ \cm_h$.
\end{lemma}
\begin{proof}
	We simply need to prove the forward implication. By definition, there exists a channel $\cn$ such that $\cm_f=\cn\circ \cm_h$. Now, for any $(a,b)\in \mathbb{Z}_n\times  \mathbb{Z}_n$, since $W_{a,b}W_{c,d}=\alpha\,W_{c,d}W_{a,b}$ where $\alpha:=\e^{\frac{2\pi i(ad-bc)}{n}}$,
	\begin{align}
	\cn\circ\cm_h(W_{a.b})&=\sum_{(a',b')\in\mathbb{Z}_{n}\times \mathbb{Z}_n}h(a',b')\,\cn\big( W_{a',b'}  \,W_{a,b}\,W_{a',b'}^\dagger\big)\\
	&=\sum_{(a',b')\in\mathbb{Z}_n\times \mathbb{Z}_n}\,h(a',b')\,\e^{\frac{2\pi i(a'b-b'a)}{n}}\cn\big(  W_{a,b}\big)\,\\
	&\equiv A_h(a,b)\,\cn(W_{a,b})\,.
	\end{align}
	On the other hand, $\cn\circ \cm_h(W_{a,b})=\cm_f(W_{a,b})=A_f(a,b)\,W_{a,b}$. Assume first that $A_h(a,b)=0$ for some $(a,b)\in \mathbb{Z}_n\times \mathbb{Z}_n$. Then $A_f(a,b)=0$ from the degradability condition, and we can choose without loss of generality $\cn(W_{a,b})=0$. Now, assume that $A_h(a,b)\ne 0$. Then we have $\cn(W_{a,b})=A_f(a,b)/A_h(a,b)\,W_{a,b}$. Therefore, we found that there exists a function $C:\mathbb{Z}_n\times \mathbb{Z}_n\mapsto \CC$ such that for all $(a,b)\in \mathbb{Z}_n\times \mathbb{Z}_n$, $\cn(W_{a,b})=C(a,b)\,W_{a,b}$. We conclude by an appeal to Theorem 4.14 in \cite{watrous2018theory}.
	
\end{proof}

In \Cref{chanpreorder}, we introduced the notion of a less noisy domination region. Restricting ourselves to comparisons with  additive noise channels, we can also define the so-called \textit{additive less noisy domination region} of a channel $\cm$ acting on system $\cH$ as
\begin{align*}
\cL_{\cm}^{\operatorname{add}}:=\big\{  \cm_f:\,\cm \succeq_{\operatorname{l.n.}}\cm_f \big\}
\end{align*}
and its \textit{additive degradation region}
\begin{align*}
\cD_{\cm}^{\operatorname{add}}:=\big\{  \cm_f:\,\cm \succeq_{\operatorname{deg}}\cm_f \big\}\,,
\end{align*}
where both are subsets of the set of Weyl-covariant channels of the form of \eqref{Phif}. 
\begin{proposition}\label{propositioncovariance}
	Let $\cm_k$ be a Weyl-covariant quantum channel corresponding to the probability distribution function $k$. Then, $\cL_{\cm_f}^{\operatorname{add}}$ and $\cL_{\cm_f}^{\operatorname{add}}$ are non-empty, closed, convex, and invariant under the conjugate action of $\{W_{a,b}\}_{(a,b)\in\mathbb{Z}_n\times\mathbb{Z}_n}$: for all $(a,b)\in \mathbb{Z}_n\times \mathbb{Z}_n$,
	\begin{align*}
	\cm_h\in\cL_{\cm_f}^{\operatorname{add}}\Rightarrow \mathcal{W}_{a,b}\circ \cm_h\in\cL_{\cm_f}^{\operatorname{add}}~~~\text{ and }~~~	 \cm_h\in\cD_{\cm_f}^{\operatorname{add}}\Rightarrow \mathcal{W}_{a,b}\circ \cm_h\in\cD_{\cm_f}^{\operatorname{add}}\,,
	\end{align*}
	where $\mathcal{W}_{a,b}(X):=W_{a,b}XW_{a,b}^\dagger$. Moreover, $\cD_{\cm_f}^{\operatorname{add}}=\operatorname{conv}\big( \big\{  \mathcal{W}_{a,b}\circ \cm_f:\,(a,b)\in \mathbb{Z}_n\times\mathbb{Z}_n   \big\}  \big)$.
\end{proposition}
\begin{proof}
	Non-emptiness, closure and convexity follow in the same way as in \Cref{obviousprops} in addition with the fact that the set of additive noise channels is closed and convex.  Invariance under the conjugate action of the Weyl operators is also apparent. The identity in (2) follows from \Cref{lemma1}.
\end{proof}
In the case when $f$ is the distribution $$\omega_\delta:=\Big(1-\delta,\frac{\delta}{n^2-1},...,\frac{\delta}{n^2-1}\Big)\,,$$ the map $\cm_\delta:=\cm_f$ reduces to
\begin{align*}
\cm_\delta(\rho)=(1-\delta)\rho+ \frac{\delta}{n^2-1}\sum_{(a,b)\ne (0,0)}\,W_{a,b}\,\rho\, W_{a,b}^\dagger&=(1-\delta)\,\rho+\frac{\delta}{n^2-1} (n\tr(\rho)\Id-\rho)\\
&=\rho p +(1-p)\tr(\rho)\,\frac{\Id}{n}\,,
\end{align*}
with $p=1-\frac{\delta n^2}{(n^2-1)}$. This is the depolarizing channel $\cN^{\operatorname{depol}}_p$ of noise parameter $p$. We recall that this map is  completely positive if and only if $-\frac{1}{n^2-1}\le p\le 1$, i.e. if and only if $0\le \delta\le 1$ \cite{horodecki1999reduction,lami2016bipartite}. Moreover, the family $(\cN^{\operatorname{depol}}_p)_{p\in\RR}$ forms a semigroup: $\cN^{\operatorname{depol}}_{p_1p_2}=\cN^{\operatorname{depol}}_{p_1}\circ\cN^{\operatorname{depol}}_{p_2}=\cN^{\operatorname{depol}}_{p_2}\circ \cN^{\operatorname{depol}}_{p_1}$. From this property, one can easily derive that for $\delta\le \frac{n^2-1}{n^2}$,
\begin{align}\label{degregdepol}
\delta\le \gamma\le 1-\frac{\delta}{n^2-1}~~~\Leftrightarrow~~~ \cm_\delta \succeq_{\operatorname{deg}}\cm_\gamma\,. 
\end{align}
In the next result, we prove that a depolarizing channel can never be less noisy than an erasure channel, hence extending Proposition 7 of \cite{makur2018comparison}:
\begin{proposition}
	For $n\ge 2$, given an erasure channel $\cE:\cB(\CC^n)\to\cB(\CC^{n+1})$ with erasure probability $\eps\in (0,1)$, there does not exist $\delta\in (0,1)$ such that the corresponding depolarizing channel $\cm_\delta:\cB(\CC^n)\to\cB(\CC^n)$ satisfies $\cm_\delta\succeq_{\operatorname{l.n.}} \cE$.
\end{proposition}
\begin{proof}
	For an erasure channel $\cE$ with $\eps\in(0,1)$, we always have $D(\cE(\Id/n)\|\cE(|1\rangle\langle 1|))=+\infty$. On the other hand, for any $\delta\in (0,1)$, we have $D(\cm_\delta(\rho)\|\cm_\delta(\sigma))<+\infty$ for all $\rho,\sigma\in\cD(\CC^n)$. Hence $\cm_\delta\not{\succeq_{\operatorname{l.n.}}} \cE $ for any $\delta\in (0,1)$ by the characterization in \Cref{propDmutualinfo}.
\end{proof}

	As we saw in \Cref{inclusions}, degradability is at least as strong as any notion of less noisy ordering. This raises the question of whether one can find pair of channels $\cm,\cn$ such that $\cm\succeq_{\operatorname{l.n.}}^{\operatorname{c}}\cn$ without $\cm\succeq_{\operatorname{deg}}\cn$. In the next result, we make one step towards this direction by proving a separation between the (non-complete) less noisy ordering and degradability between two quantum channels. This result extends Proposition 11 of  \cite{makur2018comparison}:
	\begin{proposition}\label{lnnotdeg}
		Given a Weyl-covariant channel $\cm_f$ over $\mathbb{Z}_n\times \mathbb{Z}_n$, $n\ge 2$, $\delta\in\big[ 0,\frac{n^2-1}{n^2} \big]$ and $$\gamma_0:=\frac{1-\delta}{1-\delta+\frac{\delta}{(n^2-1)^2}}\,.$$ If  $\cm_f$ is a convex mixture of the channels $\mathcal{W}_{a,b}\circ \cm_\delta$ and $\mathcal{W}_{a',b'}\circ \cm_{\gamma_0}$, $(a,b),(a',b')\in\mathbb{Z}_n\times\mathbb{Z}_n$, then $$\cm_\delta \succeq_{\operatorname{l.n.}} \cm_f\,.$$
	\end{proposition}
\begin{remark}
	Choosing $\cm_f$ to be the depolarizing channel $\cm_\delta$, we already saw in \eqref{degregdepol} that $\cm_\delta \succeq_{\operatorname{deg}}\cm_\gamma$, and hence $\cm_\delta \succeq_{\operatorname{l.n.}}\cm_\gamma$, whenever $\delta\le \gamma\le 1-\frac{\delta}{n^2-1}$. Therefore,  since $\gamma_0\ge 1-\frac{\delta}{n^2-1}$, \Cref{lnnotdeg} provides us with a class of channels that belong to $\cL^{\operatorname{add}}_{\cm_\delta}\backslash \cD^{\operatorname{add}}_{\cm_\delta}$. However the problem of finding two channels $\cm,\cn$ such that $\cm\succeq_{\operatorname{l.n.}}^{\operatorname{c}}\cn$ but $\cm\not{\succeq_{\operatorname{deg}}}\cn$ still remains open, even when the channels $\cm$ and $\cn$ are taken to be classical.
\end{remark}

\begin{proof}
	Thanks to \Cref{propositioncovariance}, it is enough to prove that $\cm_\delta \succeq_{\operatorname{l.n.}}\cm_{\gamma_0}$. Then, by the characterization in \Cref{propDmutualinfo}, we simply need to show that for any two states $\rho,\sigma\in\cD(\CC^n)$,
	\begin{align*}
	D(\cm_\delta(\rho)\|\cm_\delta(\sigma))\ge	D(\cm_{\gamma_0}(\rho)\|\cm_{\gamma_0}(\sigma))\,.
	\end{align*}
	Denoting by $\rho:=\sum_j\,\lambda_j|e_j\rangle\langle e_j|$ and $\sigma:=\sum_{k}\mu_k|f_k\rangle\langle f_k|$ the eigenvalue decompositions of the states, we can easily show that the above relative entropies are linear functions of the square overlaps between the eigenvectors of $\rho$ and those of $\sigma$. For instance, denoting by $p=1-\frac{\delta n^2}{n^2-1}$, we have
	\begin{align}\label{eqtoprove}
\cm_\delta(\rho)=\sum_j\,\lambda_j(p)\,|e_j\rangle\langle e_j|\,,~~~\cm_\delta(\sigma)=\sum_k\,\mu_k(p)\,|f_k\rangle\langle f_k|\,,
	\end{align}
	where $\lambda_j(p):=p\lambda_j+(1-p)\frac{1}{n}$ and $\mu_k(p):=p\mu_k+(1-p)\frac{1}{n}$. Hence, 
	\begin{align*}
	D(\cm_\delta(\rho)\|\cm_\delta(\sigma))=\sum_j\lambda_j(p)\ln\lambda_j(p)-\sum_{j,k}\lambda_j(p)\ln\mu_k(p)\,|\langle e_j|f_k\rangle|^2\,.
	\end{align*}
	Similarly, we denote by $q=1-\frac{\gamma_0 n^2}{n^2-1}$ and define $\lambda_j(q)$ and $\mu_j(q)$ as the corresponding convex mixtures of the eigenvalues of $\rho$, resp. $\sigma$ and the ones of the completely mixed state. Denote by $A$ the $n\times n$ matrix whose entropies are given by $a_{jk}:=|\langle e_j|f_k\rangle|^2$. Since both $\{|e_j\rangle \}$ and $\{|f_k\rangle\}$ form orthonormal bases, $A$ is doubly stochastic. Then, by Birkhoff's theorem, $A$ can be written as a convex combination of permutation matrices. Therefore, we can conclude that if an inequality of the form
	\begin{align*}
	\sum_j\lambda_j(p)\ln\lambda_j(p)-\sum_{j,k}\lambda_j(p)\ln\mu_k(p)\,a_{jk}\ge	\sum_j\lambda_j(q)\ln\lambda_j(q)-\sum_{j,k}\lambda_j(q)\ln\mu_k(q)\,a_{jk}
	\end{align*}
	holds for all permutation matrices $A$, then it holds for all doubly stochastic matrices $A$, and therefore \eqref{eqtoprove} holds for all states $\rho,\sigma$. However, restricting to the case of permutation matrices means that we can, without loss of generality, assume that the states $\rho$ and $\sigma$ commute. In other words, we reduced the proof to that of classical depolarizing channels over classical probability mass functions, for which we can directly invoke Proposition 11 in \cite{makur2018comparison}.
\end{proof}

\begin{remark}
	The use of Birkhoff's theorem for reducing a proof to the commutative regime is not new. It was recently used in the context of functional inequalities in order to find the exact (modified) logarithmic Sobolev constant of the depolarizing semigroup \cite{muller2016relative,beigi2018quantum}.
\end{remark}

\subsection{Gaussian channels}

In this section, we consider the setting of  Gaussian quantum channels \cite{holevo2012quantum}: an \textit{$n$-mode} quantum system is modeled by the Hilbert space $L^2(\RR^n)\simeq L^2(\RR)^{\otimes n}$ of complex-valued, square-integrable functions on $\RR^n$. On this space, we define the so-called \textit{creation} $a_j^\dagger$ and \textit{annihilation} $a_j$ operators, $j\in [n]$, by means of their commutation relations
\begin{align*}
[a_i,a_j^\dagger]=\delta_{ij}\Id\,,~~~~~[a_i,a_j]=[a_i^\dagger,a_j^\dagger]=0\,,~~~~~i,j=1,\dots, n\,,
\end{align*}
where for each $j\in[n]$, $a_j$ and $a_j^\dagger $ act nontrivially on the $j$-th copy of the system. The number operator is defined as $$N^{(n)}:=\sum_{j=1}^na^\dagger_ja_j$$ and corresponds to the energy of the system. In the case $n=1$, it is known to have the following spectral decomposition onto the Fock basis $\{|k\rangle\}_{k\in\NN}$: $$N^{(1)}=\sum_{k=0}^\infty \,k\,|k\rangle\langle k|\,.$$ Next, we denote by $|z\rangle$ the coherent state of parameter $z\in\CC$. These are the eigenvectors of the one-mode annihilation operator, $a|z\rangle=z|z\rangle$.  In what follows, for a vector $\mathbf{z}\in\CC^n\equiv \RR^{2n}$, the tensor product of coherent states $|z_j\rangle$ will be denoted by $|\mathbf{z}\rangle$. The Weyl displacement operators  $D(\mathbf{z})$, $\mathbf{z}\in \mathbb{C}^n$, are unitary operators that rotate the $n$-mode vacuum state $|\mathbf{0}\rangle$ to the coherent state $|\mathbf{z}\rangle=D(\mathbf{z})|\mathbf{0}\rangle$. In terms of the creation and annihilation operators, they are defined as
\begin{align}
D(\mathbf{z})=\exp\Big[ \,\mathbf{a}^\dagger\cdot\mathbf{z}-\mathbf{z}^*\cdot\mathbf{a}  \Big]\,.
\end{align}

Next, an $n$-mode quantum Gaussian state is proportional to the exponential of a quadratic polynomial in the creation and annihilation operators. 
Among these states, thermal Gaussian states play an important role. They are defined by $n$-fold products of the following geometric probability distribution for the energy:
\begin{align*}
\sigma(E):=\frac{1}{E+1}\,\sum_{k=0}^\infty\,\Big(\frac{E}{E+1}\Big)^k\,|k\rangle\langle k|\,,~~~E>0\,.
\end{align*}
 The average energy of $\sigma(E)$ is equal to $\tr(\sigma(E)N^{(1)})=E$, whereas its entropy is equal to 
\begin{align}\label{entropygaussstate}
g(E):=H(\sigma(E))=(E+1)\ln(E+1)-E\ln(E)\,.
\end{align}
Quantum Gaussian channels are those quantum channels $\cN:\cD(L^2(\RR^n))\to \cD(L^2(\RR^m))$ that preserve the set of quantum Gaussian states. The most important classes of such channels are the beam-splitter, the squeezing, quantum Gaussian attenuators and quantum Gaussian amplifiers. The first two are unitary quantum analogs of the operation of linearly mixing random variables: a \textit{two-mode beam-splitter} $U^{\operatorname{bs}}_{\lambda}:\cD(L^2(\RR^2))\to\cD( L^2(\RR^2))$ of transmissivity $0<\lambda<1$ is defined by the following action on the ladder operators:
\begin{align*}
&U_\lambda^{\operatorname{bs}\dagger}\,a_1\,U_\lambda^{\operatorname{bs}}=\sqrt{\lambda}\,a_1+\sqrt{1-\lambda}\,a_2\\
&U_\lambda^{\operatorname{bs}\dagger}\,a_2\,U_\lambda^{\operatorname{bs}}=-\sqrt{1-\lambda}\,a_1+\sqrt{\lambda}\,a_2\,.
\end{align*}
This is a simple example of a passive Gaussian channel in the sense that it preserve the total energy: $[U_\lambda^{\operatorname{bs}} ,\,N^{(2)}]=0$. On the other hand, the \textit{two-mode squeezing} $U_\kappa^{\operatorname{sq}}$ of parameter $\kappa\ge 1$ increases the energy of the input. It is defined similarly to the beam-splitter as follows:
\begin{align*}
&U_\kappa^{\operatorname{sq}\dagger}\,a_1\,U_\kappa^{\operatorname{sq}}=\sqrt{\kappa}\,a_1+\sqrt{\kappa-1}\,a_2\\
&U_\kappa^{\operatorname{sq}\dagger}\,a_2\,U_\kappa^{\operatorname{sq}}=\sqrt{\kappa-1}\,a_1+\sqrt{\kappa}\,a_2\,.
\end{align*}
In both cases, it is standard to interpret the first factor of the tensor product as the system and the second factor as the environment. Next, for $\lambda\ge 0$, we denote by $\cB_\lambda:\cD(L^2(\RR^2))\to \cD(L^2(\RR))$  the reduced state on the system of either the beam-splitter ($\lambda\le 1$) or the squeezer ($\lambda\ge 1$). The \textit{quantum Gaussian attenuator} $\cE_{\lambda,E}:\cD(L^2(\RR))\to \cD(L^2(\RR))$, of parameters $0\le \lambda\le 1$ and $E\ge0$,  corresponds to the channel $\rho\mapsto \cB_\lambda(\rho\otimes \sigma(E))$. Similarly, the \textit{quantum Gaussian amplifier} $\cA_{\kappa,E}:\cD(L^2(\RR))\to \cD(L^2(\RR))$ of parameters $\kappa\ge 1$ and $E\ge0$,  corresponds to the channel $\rho\mapsto \cB_\kappa(\rho\otimes \sigma(E))$. 

Another important Gaussian channel considered in the literature is the \textit{Gaussian additive noise channel}. It is generated by a convex combination of displacement operators with a Gaussian probability measure:
\begin{align}\label{AGNC}
\cN_E(\rho):=\int_{\CC}\,D(\sqrt{E}\mathbf{z})\,\rho\,D(\sqrt{E}\mathbf{z})^\dagger\,\e^{-|\mathbf{z}|^2}\,\frac{dz}{\pi}\,.
\end{align}
In fact, the Gaussian additive noise channel can be seen as a large particle number limit of the attenuator channel (see \cite{giovannetti2004minimum}). 

All the channels mentioned above are typical examples of a \textit{gauge-covariant} Gaussian channel \cite{holevo2012quantum}: A channel $\cm:\cD(L^2(\RR^n))\to \cD(L^2(\RR^m))$ is called gauge-covariant if for all input state $\rho\in \cD(L^2(\RR^n))$,
\begin{align*}
\cm\big( U^{(n)}_\phi\,\rho\,(U^{(n)}_\phi)^\dagger \big)=U^{(m)}_\phi\,\cm(\rho)(U^{(m)}_\phi)^\dagger
\end{align*}
where $U^{(n)}_\phi:=\e^{-i\phi\,N^{(n)}}$ is the gauge transformation corresponding to the $n$-mode harmonic oscillator. 
In fact, every phase-covariant quantum Gaussian channel can be expressed as a quantum-limited amplifier composed with a quantum-limited attenuator.

The following conjecture, stated in \cite{guha2007classical}, was recently proved in the restricted one-mode case \cite{de2017gaussian,de2016gaussian} as well as for the range of parameters for which the channels become entanglement-breaking \cite{DePalmaEB18}; that is 
\begin{itemize}
\item[(i)] For any quantum Gaussian attenuator $\cE_{\lambda,E}$ with $E\ge\frac{\lambda}{1-\lambda}$;
\item[(ii)] For any quantum Gaussian amplifier $\cA_{\kappa,E}$ with $E\ge \frac{1}{\kappa-1}$;
\item[(iii)] For any quantum Gaussian additive noise channel $\cN_E$ with $E\ge 1$.
\end{itemize}

\begin{conjecture}[Constrained minimum output entropy conjecture for phase-covariant quantum channels]\label{MOE}
For any $n\in\NN$ and any $\rho\in\cD(L^2(\RR^n))$ with finite entropy, let $\sigma$ be the one-mode thermal Gaussian state with entropy $\frac{H(\rho)}{n}$. Then, for any $E\ge 0$, $0\le \lambda\le 1$, $\kappa \ge 1$: 
\begin{align}
&H(\cE^{\otimes n}_{\lambda,E}(\rho))\ge H(\cE_{\lambda,E}^{\otimes n}(\sigma^{\otimes n}))=n\,g\Big( \lambda\,g^{-1}\Big(\frac{H(\rho)}{n}\Big)+(1-\lambda)E\Big)\,,\label{EPNIatt}\\
&H(\cA^{\otimes n}_{\kappa,E}(\rho))\ge H(\cA_{\kappa,E}^{\otimes n}(\sigma^{\otimes n}))=n\,g\Big( \kappa\,g^{-1}\Big(\frac{H(\rho)}{n}\Big)+(\kappa-1)(E+1)\Big)\,,\label{EPNIampl}\\
&H(\cN_E^{\otimes n}(\rho))\ge H(\cN_E^{\otimes n}(\sigma^{\otimes n}))=n\,g\Big(g^{-1}\Big(\frac{H(\rho)}{n} \Big)   + E\Big)\label{EPNIadd}\,.
\end{align}
\end{conjecture}
 In the next result, we provide an exact expression for contraction coefficients of any of these Gaussian channels and their tensor products under the condition that \Cref{MOE} holds. More precisely, let $\sigma_j\in\cD(L^2(\RR))$, $j\in[n]$, be one-mode Gaussian states, $\cG$ be a one-mode Gaussian quantum channel, and define the contraction coefficients for the relative entropy as
 \begin{align}
 &\eta_{\operatorname{Re}}\big( \bigotimes_j \sigma_j,\cG^{\otimes n}\big):=\sup_{\substack{\rho\in\cD(L^2(\RR^n))\\0<D(\rho\|\bigotimes_j\sigma_j)<\infty}}\,\frac{D(\cG^{\otimes n}(\rho)\|\bigotimes_j\cG(\sigma_j))}{D(\rho\|\bigotimes_j\sigma_j)}\,.\label{ReGauss}
 \end{align}
 Here, we define the relative entropy as Lindblad in \cite{lindblad1974expectations,lindblad1973entropy}: given any two positive, trace-class operators $A,B$, if $\{|a_i\rangle\}$, resp. $\{|b_j\rangle\}$,  is a complete orthonormal set of eigenvectors of $A$ with corresponding eigenvalues $\{a_i\}$, resp. $B$ with eigenvalues $\{b_j\}$,
 \begin{align*}
 D(A\|B)&:=\sum_i\,\langle a_i|\,(A\ln A-A\ln B+B-A)\,|a_i\rangle\\
 &=\sum_{j}\,\langle b_j|\,(A\ln A-A\ln B+B-A)\,|b_j\rangle\\
 &=\sum_{i,j}|\langle a_i|b_j\rangle|^2\,\big( a_i\ln a_i-a_i\ln b_j+b_j-a_i \big)\,.
 \end{align*}
 where the sum is taken to be $+\infty$ if the series diverges.
 Furthermore, we define the \textit{energy-constrained} relative entropy contraction coefficient as follows: for $p>0$
 \begin{align*}
\eta_{\operatorname{Re}}^{p}\big(\cG^{\otimes n},\bigotimes_j\sigma_j\big):=\sup_{\substack{\rho\in\cD(L^2(\RR^n)),\,\tr(\rho\,N^{(n)})\le p\\0<D(\rho\|\bigotimes_j\sigma_j)<\infty}}\,\frac{D(\cG^{\otimes n}(\rho)\|\bigotimes_j \cG(\sigma_j))}{D(\rho\|\bigotimes_j\sigma_j)}\,.
 \end{align*}

 \begin{theorem}\label{tensorizationGausss}
Assume that \Cref{MOE} holds true. Then for any family $\{E_j\}_{j\in [n]}$ of energies, any $0\le \lambda\le 1$, $1\ge \kappa$, $E\ge 0$, and $p\ge \max_{j}\,E_j$: 
\begin{align}
&\eta_{\operatorname{Re}}\big(\cE_{\lambda,E}^{\otimes n},\bigotimes_j\sigma(E_j)\big)=\eta_{\operatorname{Re}}^p\big(\cE_{\lambda,E}^{\otimes n},\bigotimes_j\sigma(E_j)\big)=\lambda\,\max_{j\in[n]}\,   \frac{\ln\Big(\frac{\lambda E_j+(1-\lambda)E+1}{\lambda E_j+(1-\lambda)E}\Big)}{\ln\Big(\frac{E_j+1}{E_j}\Big)}\, \label{contrac}\,,\\
&\eta_{\operatorname{Re}}\big(\cA_{\kappa,E}^{\otimes n},\bigotimes_j\sigma(E_j)\big)=\eta_{\operatorname{Re}}^p\big(\cA_{\kappa,E}^{\otimes n},\bigotimes_j\sigma(E_j)\big)=\kappa\,\max_{j\in[n]}\,   \frac{\ln\Big(\frac{\kappa E_j+(\kappa-1)(E+1)+1}{\kappa E_j+(\kappa-1)(E+1)}\Big)}{\ln\Big(\frac{E_j+1}{E_j}\Big)}\,,\nonumber\\
&\eta_{\operatorname{Re}}\big(\cN_{E}^{\otimes n},\bigotimes_j\sigma(E_j)\big)=\eta_{\operatorname{Re}}^p\big(\cN_{E}^{\otimes n},\bigotimes_j\sigma(E_j)\big)=\,\max_{j\in[n]}\,   \frac{\ln\Big( \frac{ E_j+E+1}{ E_j+ E}\Big)}{\ln\Big( \frac{E_j+1}{E_j}\Big)}\, .\nonumber
\end{align}
 \end{theorem}
 Before proving \Cref{tensorizationGausss}, we state and prove a couple of technical lemmas. The first one  allows us to restrict the optimization \eqref{ReGauss} to finite-rank input states (see Lemma 21 in \cite{makur2015linear} for a classical analogue):
\begin{lemma}[Finite-rank state characterization of $\eta_{\operatorname{Re}}$]\label{approxcontracReLemm} Let $\cT$ be the set of finite-rank $n$-mode quantum states
$\rho\in\cD(L^2(\RR^n))$ supported on $\cH_0:=\operatorname{lin}\{\,|k_1\rangle \otimes \dots\,\otimes |k_n\rangle ,\,k_j\in\NN\}$. Then for any $1$-mode bosonic channel $\cN$, and any thermal Gaussian states $\sigma_1,\dots,\sigma_n$:
\begin{align}\label{approxcontracRe}
\eta_{\operatorname{Re}}\big(\cN^{\otimes n},\bigotimes_j\sigma_j\big)= \sup_{\substack{\rho\in\cT\\D(\rho\|\bigotimes_j\sigma_j)<\infty}}\,\frac{D(\cN^{\otimes n}(\rho)\|\bigotimes_j\cN(\sigma_j))}{D(\rho\|\bigotimes_j\sigma_j)}\,.
\end{align}
 \end{lemma}
 \begin{remark}
 	Remark that the optimization does not need to assume that $D(\rho\|\bigotimes_j\sigma_j)>0$: this is automatically true since $\bigotimes_j\sigma_j$ is faithful, and hence cannot be equal to any finite-rank state. 
 \end{remark}
 
\begin{proof}
Consider a quantum state $\rho\in\cD(L^2(\RR^n))$ such that $0<D(\rho\|\bigotimes_j\sigma_j)<\infty$, and define a sequence $\Pi_k$ of  finite-dimensional orthogonal projections on $L^2(\RR^n)$ which commute with $\bigotimes_j\sigma_j$ and such that $\Pi_k\to \Id$ strongly as $k\to\infty$ (i.e. $\lim_{k\to\infty}\|(\Pi_k-\Id)\psi\|=0$ for all $\psi\in L^2(\RR^n)$). Then, for $\rho^{(k)}:=\frac{\Pi_k\rho\Pi_k}{\tr(\rho\Pi_k)}$,
\begin{align*}
D(\rho^{(k)}\|\bigotimes_j\sigma_j) =\frac{1}{\tr(\Pi_k\rho)} \,D(\Pi_k\rho\,\Pi_k\|\Pi_k\bigotimes_j\sigma_j\Pi_k)+1-\frac{\tr(\Pi_k\sigma)}{\tr(\Pi_k\rho)}-\ln\tr(\Pi_k\rho)\,,
\end{align*}	
where the identity follows from the fact that $[\Pi_k,\bigotimes_j\sigma_j]=0$ by assumption. Now by strong convergence of the sequence $\{\Pi_k\}_{k\in\NN}$, $\tr(\Pi_k\sigma),\tr(\Pi_k\rho)\to 1$ as $k\to\infty$. Moreover, by Lemma 4 in \cite{lindblad1974expectations}, we have that $D(\Pi_k\rho\,\Pi_k\|\Pi_k\bigotimes_j\sigma_j\Pi_k)\to D(\rho\|\sigma)$ as $k\to\infty$. Hence,
\begin{align}\label{eqcont}
D(\rho^{(k)}\|\bigotimes_j\sigma_j)\,\underset{k\to\infty}{\to}\, D(\rho\|\bigotimes_j\sigma_j)\,.
\end{align}
Moreover, by lower semi-continuity of the quantum relative entropy, we also have that 
\begin{align}\label{eqsemicont}
\liminf_{k\to\infty}
\,D(\cN^{\otimes n}(\rho^{(k)})\|\bigotimes_j\cN(\sigma_j))\ge D(\cN^{\otimes n}(\rho)\|\bigotimes_j\cN(\sigma_j))\,.
\end{align}
Combining \Cref{eqcont} and \Cref{eqsemicont}, we have that 
\begin{align}\label{eq:convapprox}
\liminf_{k\to\infty}\,\frac{D(\cN^{\otimes n}(\rho^{(k)})\|\bigotimes_j\cN_j(\sigma))}{D(\rho^{(k)}\|\bigotimes_j\sigma_j)}\,\ge\, \frac{D(\cN^{\otimes n}(\rho)\|\bigotimes_j\cN(\sigma_j))}{D(\rho\|\bigotimes_j\sigma_j)}\,.
\end{align}
The rest of the proof relies on a diagonalization argument that was already used in Lemma 21 in \cite{makur2015linear}: suppose that $\{\rho_{m}\}_{m\in\NN}$ is a sequence of quantum states that satisfies $0<D(\rho_m\|\bigotimes_j\sigma_j)<\infty$ for all $m\in\NN$ and achieves the supremum in \eqref{ReGauss}:
\begin{align*}
\lim_{m\to\infty}\,\frac{D(\cN^{\otimes n}(\rho_m)\|\bigotimes_j\cN(\sigma_j))}{D(\rho_m\|\bigotimes_j\sigma_j)}=\eta_{\operatorname{Re}}(\cN^{\otimes n},\bigotimes_j\sigma_j)\,.
\end{align*}
Then, by \eqref{eq:convapprox}, we can construct a sequence $\{\rho_m^{(k(m))}\in\cT\}_{m\in\NN}$ where each $k(m)$ is chosen such that for every $m\in\NN$:
\begin{align*}
\frac{D(\cN^{\otimes n}(\rho_m^{(k(m))})\|\bigotimes_j\cN(\sigma_j))}{D(\rho_m^{(k(m))}\|\bigotimes_j\sigma_j)}\,\ge \, \frac{D(\cN^{\otimes n}(\rho_m)\|\bigotimes_j\cN(\sigma_j))}{D(\rho_m\|\bigotimes_j\sigma_j)}-\frac{1}{2^m}\,.
\end{align*}
Letting $m\to\infty$, we end up with
\begin{align*}
\liminf_{m\to\infty}\,\frac{D(\cN^{\otimes n}(\rho_m^{(k(m))})\|\bigotimes_j\cN(\sigma_j))}{D(\rho_m^{(k(m))}\|\bigotimes_j\sigma_j)}\,\ge\,\eta_{\operatorname{Re}}(\cN^{\otimes n},\bigotimes_j\sigma_j)\,.
\end{align*}
Since the supremum in \eqref{ReGauss} is over a smaller set than in \eqref{approxcontracRe}, the above inequality is actually an equality, and the result follows. 
\end{proof} 
The next technical lemma will serve as a replacement of  Bernoulli's inequality in \cite{makur2015linear}:
 \begin{lemma}
 	For any three numbers $\nu,\,0\le\beta\le\alpha$, 
 	\begin{align}\label{additivenoiseineq}
 	g(\nu+\alpha)-g(\nu+\beta)\le\,\frac{\ln\Big( \frac{\nu+\alpha+1}{\nu+\alpha} \Big)}{\ln\Big(\frac{\alpha+1}{\alpha}\Big)}\,\big( g(\alpha)-g(\beta) \big)\,,
 	\end{align}
 	where $g(x):=(x+1)\ln(x+1)-x\ln(x)$. Similarly, for any $\nu,\,0\le\beta\le\alpha$ and $0\le \eta$, we have
 	\begin{align}\label{attenuatorineq}
 	 	g(\eta\alpha+|1-\eta|\nu)-g(\eta\beta+|1-\eta|\nu)\le\,\eta\,\frac{\ln\Big( \frac{\eta\alpha+|1-\eta|\nu+1}{\eta\alpha+|1-\eta|\nu} \Big)}{\ln\Big(\frac{\alpha+1}{\alpha}\Big)}\,\big( g(\alpha)-g(\beta) \big)\,,
 	\end{align} 
 \end{lemma}

\begin{proof}
	The proof of \eqref{additivenoiseineq} amounts to proving that the function
	\begin{align*} \nu\mapsto \frac{g(\nu+\alpha)-g(\nu+\beta) }{\ln(\nu+\alpha+1)-\ln(\nu+\alpha)}\equiv \frac{g(\nu+\alpha)-g(\nu+\beta) }{g'(\nu+\alpha)}
	\end{align*}
	is monotone decreasing. Now, wince $g$ is differentiable, we have
\begin{align*}
g(\nu+\alpha)-g(\nu+\beta)=\int_{\nu+\beta}^{\nu+\alpha}\,\,g'(s)\,ds=\int_{\beta}^\alpha\,g'(s+\nu)\,ds=\int_\beta^\alpha\,\ln\Big( \frac{s+\nu+1}{s+\nu}  \Big)\,ds\,.
\end{align*}
Therefore, after dividing by $g'(\nu+\alpha)$, we reduced the proof to that of showing that the function 
$\nu\mapsto\frac{\ln(s+\nu+1)-\ln(s+\nu)}{\ln(\nu+\alpha+1)-\ln(\nu+\alpha)}$ is monotone decreasing for all $\beta\le s\le \alpha$. After differentiation, we further simplify the proof to that of showing that the function
\begin{align}\label{func}
x\mapsto -\frac{1}{(x+1)x}\,\frac{1}{\ln(x+1)-\ln(x)}
\end{align}
is increasing. This last assertion can be showed to hold after differentiating the function defined in \eqref{func}. To obtain \eqref{attenuatorineq}, we simply operate the following replacements in \eqref{additivenoiseineq}: $\alpha\to \eta\alpha$, $\nu\to|1-\eta|\nu$ and $\beta\to\eta\beta$. Then we are reduced to having to prove that 
\begin{align*}
\frac{g(\eta\alpha)-g(\eta\beta)}{\ln(\eta\alpha+1)-\ln(\eta\alpha)}\le\,\eta\, \frac{g(\alpha)-g(\beta)}{\ln(\alpha+1)-\ln(\alpha)}\,.
\end{align*}
Using once again the fundamental theorem of calculus and an obvious rescaling, we reduce the problem to that of proving that for all $u\le \alpha$ the function
\begin{align*}
\eta\mapsto \frac{\ln\Big(  \frac{\eta u+1}{\eta u} \Big)}{\ln\Big( \frac{\eta\alpha+1}{\eta\alpha}\Big)}
\end{align*}
is increasing. By differentiating this function, it is enough to show that the function $$x\mapsto (x+1)\ln\Big(\frac{x+1}{x}\Big)$$ on $\RR_+$ is decreasing. One last differentiation reduces the problem to the basic inequality $\ln\big(1+u)\le u$ on $\RR_+$.
\end{proof}

We are now ready to prove our main result.
\begin{proof}[Proof of \Cref{tensorizationGausss}]
\textbf{Step 1: Upper bounding $\eta_{\Re}$}: by \Cref{approxcontracReLemm}, we can restrict the optimization to that over finite-rank input states $\rho\in\cD(L^2(\RR^n))$ supported on $\cH_0=\operatorname{lin}\{|k_1\rangle\otimes  \dots\otimes| k_n\rangle,\,k_j\in\NN\}$, so that $\tr(\rho\,N^{(n)})<\infty$. Such states also have finite von Neumann entropy. Moreover, the energy at the output of any Gaussian quantum channel is also finite \cite{shirokov2019extension}. Then given any sequence $\{\sigma(E_j)\}$ of thermal Gaussian states, and denoting by 
$q_j:=\tr\big(\big(\rho_j-\sigma(E_j)\big)N^{(1)}\big)$ the difference in energy between $\rho_j=\tr_{\{j\}^c}(\rho)$ and $\sigma(E_j)$,
	\begin{align}
	D(\rho\|\bigotimes_j\sigma(E_j))&=-H(\rho)-\tr\big[ \rho\,\ln\,\big(\bigotimes_j\sigma(E_j)\big) \big]\nonumber\\
	&=-H(\rho)+\sum_j\Big(\, S\big(\sigma(E_j)\big)+q_j\ln\Big( \frac{E_j+1}{E_j}\Big)\Big)\nonumber\\
	&=-H(\rho)+\sum_j\Big(\,g(E_j)+q_j\ln\Big( \frac{E_j+1}{E_j}\Big)\Big)\,.
	\end{align}
	Let us now consider the attenuator channel $\cE_{\lambda,E}$. Since $\cE_{\lambda,E}^\dagger(N^{(1)})=\lambda\,N^{(1)}+(1-\lambda)E\,\Id$, we have
	\begin{align}
	&D(\cE_{\lambda,E}^{\otimes n}(\rho)\|\bigotimes_j\cE_{\lambda,E}(\sigma(E_j)))\nonumber\\	&=-H(\cE_{\lambda,E}^{\otimes n}(\rho))+\sum_j\Big(H(\cE_{\lambda,E}(\sigma(E_j)))+\lambda\,q_j\ln\Big( \frac{\lambda E_j+(1-\lambda)E+1}{\lambda E_j+(1-\lambda)E} \Big)   \Big)\nonumber\\
	&=-H(\cE_{\lambda,E}^{\otimes n }(\rho))+\sum_j\Big(\,g(\lambda E_j+(1-\lambda)E)+\lambda\,q_j\ln\Big( \frac{\lambda E_j+(1-\lambda)E+1}{\lambda E_j+(1-\lambda) E}\Big)\Big)\,.\label{eq:attenuator}
	\end{align}
Hence, for $\eta_{\operatorname{Re}}\big(\bigotimes_j\sigma(E_j),\cE_{\lambda,E}^{\otimes n}\big)$ to be upper bounded by the right-hand side of \eqref{contrac}, it suffices to prove that for any $j\in[n]$:
\begin{align*}
g(\lambda E_j+(1-\lambda)E)-\frac{H(\cE_{\lambda,E}^{\otimes n}(\rho))}{n}\le \frac{\ln\Big(\frac{\lambda E_j+(1-\lambda)E+1}{\lambda E_j+(1-\lambda)E}\Big)}{\ln\Big(\frac{E_j+1}{E_j}\Big)}\, \lambda\,\Big(g(E_j)-\frac{H(\rho)}{n}  \Big)\,.
\end{align*}
By the constrained minimum entropy conjecture \eqref{EPNIatt}, we can further simplify the problem to that of proving, for any ${\beta}:=g^{-1}\big(  \frac{H(\rho)}{n}\big)$,
\begin{align}
g(\lambda E_j+(1-\lambda)E)-\,g\big( \lambda\,{\beta}+(1-\lambda)E\big)\le\lambda\, \frac{\ln\Big(\frac{\lambda E_j+(1-\lambda)E+1}{\lambda E_j+(1-\lambda)E}\Big)}{\ln\Big(\frac{E_j+1}{E_j}\Big)}\,\big(g(E_j)-g({\beta})  \big)\,.
\end{align}
The result directly follows from \eqref{attenuatorineq} for $\alpha:= E_j$ and $\nu:=E$. The case of the amplifier $\cA_{\kappa,E}$ follow the exact same proof, since $\cA_{\kappa,E}^{\dagger}(N^{(1)})=\kappa N^{(1)}+(\kappa-1)(E+1)\Id$.

Now, we turn our attention to the additive noise channel $\cN_E$: here, since $\cN_E^\dagger(N^{(1)})=N^{(1)}+E\,\Id$, \eqref{eq:attenuator} has to be replaced by 
	\begin{align}
D(\cN_{E}^{\otimes n}(\rho)\|\bigotimes_j\cN_{E}(\sigma(E_j)))	&=-H(\cN_{E}^{\otimes n}(\rho))+\sum_j\Big(H(\cN_{E}(\sigma(E_j)))+q_j\ln\Big( \frac{ E_j+E+1}{ E_j+E} \Big)   \Big)\nonumber\\
&=-H(\cN_{E}^{\otimes }(\rho))+\sum_j\Big(\,g( E_j+E)+q_j\ln\Big( \frac{ E_j+E+1}{ E_j+ E}\Big)\Big)\label{eq:addnoise}\,.
\end{align}
Therefore, as previously, by conjecture \eqref{EPNIadd} it is enough to prove that for any $j\in[n]$, denoting by $\beta:=g^{-1}\big(  \frac{H(\rho)}{n}\big)$,
\begin{align*}
g(E_j+E)-g\big(\beta+E\big)\le \frac{\ln\Big( \frac{ E_j+E+1}{ E_j+ E}\Big)}{\ln\Big( \frac{E_j+1}{E_j}\Big)}\,\big(  g(E_j)-g(\beta)\big)
\end{align*}
Then, invoking \eqref{additivenoiseineq} for $\alpha:=E_j$ and $\nu:=E$ allows us to conclude. 

\textbf{Step 2: $\eta_{\Re}\ge \eta_{\Re}^{(p)}$}: this is obvious by definition of the contraction coefficients.

\textbf{Step 3: Lower bounding $\eta_{\Re}^{(p)}$ for all $p\ge \max_j E_j$}: let us start with the additive noise channel $\cN_{E}$: choosing $\rho$ to be a tensor product state, we can reduce to the case $n=1$ without loss of generality. Moreover, we choose $\rho$ to be the thermal Gaussian state $\sigma(E')$. Then, we have for any displaced thermal state $\sigma(E')_z:=D(z)\,\sigma(E')\,D(z)^\dagger$, $z\in \CC$,
\begin{align*}
&\frac{D(\cN_E(\sigma(E')_z)\|\cN_{E}(\sigma(E_1)))}{D(\sigma(E')_z\|\sigma(E_1))} \\ &=\frac{-g(E'+E)+(1+E'+|z|^2+E)\ln({E_1+1+E})-(E'+E+|z|^2)\ln(E_1+E)}{-g(E')+(1+E'+|z|^2)\ln(E_1+1)-(E'+|z|^2)\ln(E_1)}\,.
\end{align*}
The lower bound follows after taking $E'=E_1-\delta$, $|z|^2=\delta$ and taking the limit $\delta\to 0$. Since $\tr(\sigma(E')_z\,N^{(1)})=E_1$ for all $\delta>0$, the result follows. The cases of the amplifier and of the attenuator follow the exact same computations. 
The other two cases follows the exact same argument.
\end{proof}

\begin{remark}
	In the case of the Gaussian attenuator channel $\cE_{\lambda,E}$, we find that the contraction coefficient $\eta_{\operatorname{Re}}(\sigma(E)^{\otimes n},\cE^{\otimes n}_{\lambda,E})=\lambda$. This is also a consequence of the conditional entropy power inequality, as previously observed in  \cite{de2018conditional,de2018conditional2}.
\end{remark}

\begin{remark}
	In the classical setting, it is well-known that the contraction coefficient $\eta(\mathcal{G})=1$ for the classical additive white Gaussian noise channel, even when restricting to inputs of  finite energy. In order to get nontrivial contraction, the following non-linear strong data processing inequality was proposed in \cite{7279116,8186217}: for the channel $\cG:X\to Y$
	\begin{align*}
	F(t,\cG,E):=\underset{\substack{U\to X\to Y\\I(U;X)\le t\\ \mathbb{E}[X^2]\le E}}{\sup}\,I(U;Y)\,,
	\end{align*}
	where the supremum is over all joint distributions $P_{UX}$ with constrained input energy over system $X$. In \cite{8186217}, it was shown that such curves are always strictly less than $t$ for all $t>0$. A similar analysis in the quantum setting will be the subject of  future work.
\end{remark}

\section{Conclusions}

In this work, we discussed the relationship between the concept of contraction coefficients and the less noisy partial order. The equivalence between their respective relative entropy and mutual information formulations allowed us a new point of view on their properties and bringing our knowledge in the quantum setting closer to what is currently available in the classical case. Given the central position data processing takes in the field of quantum information theory, we expect our results to lay the ground for further results and applications to plenty information processing tasks. Indeed, the techniques developed here already found applications in obtaining state-of-the-art bounds on the performance of noisy quantum devices~\cite{DePalma2022,Wang2021a},  novel capacity bounds~\cite{Hirche2022} and the study of quantum differential privacy~\cite{Hirche2022a}.

Nevertheless, many open questions remain. Most crucially, it is not yet clear whether the relations between the different partial orders are actually strict. For example, we know that 
$$\cn\deg\cm  \Rightarrow   \cn\cln\cm  \Rightarrow \cn\rln\cm$$
but it is not known whether the opposite directions might also hold, namely
$$\cn\rln\cm  \xRightarrow{?}   \cn\cln\cm  \xRightarrow{?} \cn\deg\cm\,.$$ 
Moreover, the fully quantum less noisy order and the regularized less noisy order sit somewhat parallel in the hierarchy of orders, compare Figure~\ref{Fig:RelationsPO}. In the exact case, both have been used to e.g. to give a condition for which the quantum capacity becomes single-letter. It would have interesting applications to find a direct relationship between the two orders, i.e. one or both of the following:
$$\cn\rln\cm  \xRightarrow{?}   \cn\fqln\cm\, ,  \qquad\quad \cn\fqln\cm  \xRightarrow{?}   \cn\rln\cm \,.$$

Finally, for SDPI constants and contraction coefficients, natural open problems include a better understanding of their tensorization properties and the possible equivalence of constants based on different divergences.

\section*{Acknowledgments} 
 CH and DSF acknowledge financial support from the VILLUM FONDEN via the QMATH Centre of Excellence (Grant no. 10059)  and the QuantERA ERA-NET Cofund in Quantum Technologies implemented within the European Union’s Horizon 2020 Programme (QuantAlgo project) via the Innovation Fund Denmark. DSF acknowledges financial support from the European Research Council (grant agreement no. 81876).
 CR is partially supported by a Junior Researcher START Fellowship from the MCQST.  CR acknowledges funding by the Deutsche Forschungsgemeinschaft (DFG, German Research Foundation) under Germanys Excellence Strategy EXC-2111 390814868.

\bibliographystyle{plainurl}
\bibliography{library_doi}

\end{document}